\def\confversion{0}
\def\ifconf{\ifnum\confversion=1}
\def\ifnotconf{\ifnum\confversion=0}

\documentclass[11 pt]{article}

\def\showauthornotes{0}
\def\showkeys{0}
\def\showdraftbox{0}

\usepackage{xspace,xcolor,enumerate}
\usepackage{amsmath,amssymb}
\usepackage{amsthm}
\usepackage[toc,page]{appendix}
\usepackage{thmtools}
\usepackage{thm-restate}
\usepackage{color,graphicx}
\usepackage{boxedminipage}
\usepackage{makecell}
\usepackage{tabularx}

\ifnum\showkeys=1
\usepackage[color]{showkeys}
\fi

\definecolor{darkred}{rgb}{0.5,0,0}
\definecolor{darkgreen}{rgb}{0,0.35,0}
\definecolor{darkblue}{rgb}{0,0,0.55}

\usepackage[pdfstartview=FitH,pdfpagemode=UseNone,colorlinks,linkcolor=darkblue,filecolor=darkred,citecolor=darkgreen,urlcolor=darkred,pagebackref]{hyperref}

\usepackage[capitalise,nameinlink]{cleveref}
\usepackage[T1]{fontenc}
\usepackage{mathtools,dsfont,bbm}
\usepackage{mathpazo}

\ifnum\showauthornotes=1
\newcommand{\Authornote}[2]{{\sf\small\color{red}{[#1: #2]}}}
\newcommand{\Authorcomment}[2]{{\sf \small\color{gray}{[#1: #2]}}}
\newcommand{\Authorfnote}[2]{\footnote{\color{red}{#1: #2}}}
\else
\newcommand{\Authornote}[2]{}
\newcommand{\Authorcomment}[2]{}
\newcommand{\Authorfnote}[2]{}
\fi

\ifnum\showdraftbox=1
\newcommand{\draftbox}{\begin{center}
  \fbox{%
    \begin{minipage}{2in}%
      \begin{center}%
        \begin{Large}%
          \textsc{Working Draft}%
        \end{Large}\\
        Please do not distribute%
      \end{center}%
    \end{minipage}%
  }%
\end{center}
\vspace{0.2cm}}
\else
\newcommand{\draftbox}{}
\fi

\newtheorem{theorem}{Theorem}[section]

\newtheorem{definition}[theorem]{Definition}
\newtheorem{lemma}[theorem]{Lemma}
\newtheorem{remark}[theorem]{Remark}

\newtheorem{corollary}[theorem]{Corollary}
\newtheorem{claim}[theorem]{Claim}
\newtheorem{fact}[theorem]{Fact}

\newtheorem{algo}[theorem]{Algorithm}

\newenvironment{algorithm}[3]
        {\noindent\begin{boxedminipage}{\textwidth}\begin{algo}[#1]\ \par
        {\begin{tabularx}{\textwidth}{r X}
        \textbf{Input} & #2\\
        \textbf{Output} & #3
        \end{tabularx}\par\enskip}}
        {\end{algo}\end{boxedminipage}}

\def\FullBox{\hbox{\vrule width 6pt height 6pt depth 0pt}}

\def\qed{\ifmmode\qquad\FullBox\else{\unskip\nobreak\hfil
\penalty50\hskip1em\null\nobreak\hfil\FullBox
\parfillskip=0pt\finalhyphendemerits=0\endgraf}\fi}

\def\qedsketch{\ifmmode\Box\else{\unskip\nobreak\hfil
\penalty50\hskip1em\null\nobreak\hfil$\Box$
\parfillskip=0pt\finalhyphendemerits=0\endgraf}\fi}

\def\to{\rightarrow}

\def\epsilon{\varepsilon}

\def\phi{\varphi}

\def\implies{\Rightarrow}

\renewcommand{\bar}{\overline} 

\newcommand{\ie}{i.e.,\xspace}

\newcommand{\etal}{et al.\xspace}

\newcommand{\mper}{\,.}
\newcommand{\mcom}{\,,}

\newcommand{\R}{{\mathbb R}}
\newcommand{\E}{{\mathbb E}}

\newcommand{\N}{{\mathbb{N}}}

\newcommand{\F}{{\mathbb F}}

\newcommand{\FF}{{\mathbb F}}

\usepackage{nicefrac}

\newcommand{\abs}[1]{\ensuremath{\left\lvert #1 \right\rvert}}

\newcommand{\norm}[1]{\ensuremath{\left\lVert #1 \right\rVert}}
\newcommand{\smallnorm}[1]{\ensuremath{\lVert #1 \rVert}}

\newcommand{\mydot}[2]{\ensuremath{\left\langle #1, #2 \right\rangle}}

\newcommand{\ip}[2] {\ensuremath{\langle #1 , #2 \rangle}}

\newcommand{\one}{{\mathbf{1}}}

\newcommand{\qary}{[q]}

\newcommand{\Esymb}{\mathbb{E}}
\newcommand{\Psymb}{\mathbb{P}}

\newcommand{\Varsymb}{\mathrm{Var}}
\DeclareMathOperator*{\ExpOp}{\Esymb}
\DeclareMathOperator*{\ProbOp}{\Psymb}

\makeatletter

\def\Ex#1{%
    \ProbabilityRender{\Esymb}{#1}%
}

\def\Var#1{%
    \ProbabilityRender{\Varsymb}{#1}%
}

\def\ProbabilityRender#1#2{
  \@ifnextchar\bgroup%
  {\renderwithdist{#1}{#2}}
   {\singlervrender{#1}{#2}}
}
\def\singlervrender#1#2{%
   \ensuremath{\mathchoice
       {{#1}\left[ #2 \right]}
       {{#1}[ #2 ]}
       {{#1}[ #2 ]}
       {{#1}[ #2 ]}
   }
}
\def\renderwithdist#1#2#3{%
   \@ifnextchar\bgroup
   {\superfancyrender{#1}{#2}{#3}}
   {\ensuremath{\mathchoice
      {\underset{#2}{#1}\left[ #3 \right]}
      {{#1}_{#2}[ #3 ]}
      {{#1}_{#2}[ #3 ]}
      {{#1}_{#2}[ #3 ]}
     }
   }
}
\def\superfancyrender#1#2#3#4#5{
   \ensuremath{\mathchoice
      {\underset{#1}{{#1}}\left#4 #3 \right#5}
      {{#1}_{#2}#4 #3 #5}
      {{#1}_{#2}#4 #3 #5}
      {{#1}_{#2}#4 #3 #5}
   }
}
\makeatother

\newfont{\inhead}{eufm10 scaled\magstep1}
\newcommand{\deffont}{\sf}

\newcommand{\poly}{{\mathrm{poly}}}
\newcommand{\polylog}{{\mathrm{polylog}}}

\DeclareMathOperator\supp{Supp}

\DeclareMathOperator{\cov}{\operatorname {Cov}}

\DeclareMathOperator*{\argmax}{\arg\!\max}

\renewcommand{\bar}[1]{\ensuremath{\overline{#1}}}

\newcommand{\problemmacro}[1]{\textsf{#1}}

\newcommand{\maxkcsp}{\problemmacro{MAX k-CSP}\xspace}

\newcommand{\inparen}[1]{\left(#1\right)}             
\newcommand{\inbraces}[1]{\left\{#1\right\}}           
\newcommand{\insquare}[1]{\left[#1\right]}             

\DeclareSymbolFont{extraup}{U}{zavm}{m}{n}
\DeclareMathSymbol{\varheart}{\mathalpha}{extraup}{86}
\DeclareMathSymbol{\vardiamond}{\mathalpha}{extraup}{87}

\def\RR{\mathbb R}

\def\matr#1{\mathsf{#1}}

\def\rand#1{\mathbf{#1}}
\def\rv#1{\rand #1}
\def\conj{\overline}

\def\Ray{\mathcal R}

\def\One{\mathbb 1}
\def\one{\mathbf 1}

\def\Pred{\mathcal P}
\def\Cc{\mathcal C}

\def\OPT{\mathsf{OPT}}
\def\SAT{\mathsf{SAT}}

\def\SDP{\mathsf{SDP}}

\def\Ins{\mathfrak I}
\def\assn{\sigma}
\def\tree{\mathcal T}

\def\Aye{\matr A}

\def\Ess{\matr S}

\def\Emm{\matr M}
\def\Jay{\matr J}

\def\aye{\mathfrak a}
\def\ess{\mathfrak s}
\def\tee{\mathfrak t}
\def\bee{\mathfrak b}

\DeclareMathOperator{\SwapT}{Swap}

\DeclareMathOperator{\Rank}{rank}
\DeclareMathOperator{\PExp}{\widetilde \E}
\DeclarePairedDelimiter\set{\lbrace}{\rbrace}
\DeclarePairedDelimiter\parens{\lparen}{\rparen}
\DeclarePairedDelimiter\sqbr{[ }{]}

\DeclareMathOperator{\bias}{bias}
\DeclareMathOperator{\lift}{lift}
\DeclareMathOperator{\dsum}{dsum}
\DeclareMathOperator{\dprod}{dprod}

\newcommand{\wswap}[2]{\Ess_{#1,#2}^{\circ}}

\newcommand{\fswap}[2]{\mathfrak{S}_{#1,#2}}

\newcommand{\lict}[1]{$(1/2-#1)$-\text{close}}

\usepackage{tikz}
\usepackage{float}

\usepackage[top=1in, bottom=1in, left=1.25in, right=1.25in]{geometry}

\usepackage{amsmath}
\usepackage{amssymb}
\usepackage{amsthm}
\usepackage{ytableau}
\usepackage{mathtools}

\begin{document}

\title{List Decoding of Direct Sum Codes}

\author{
Vedat Levi Alev\thanks{ Supported by NSERC Discovery Grant 2950-120715, NSERC Accelerator Supplement 2950-120719, and partially supported by NSF awards CCF-1254044 and CCF-1718820. {\tt University of Waterloo}. {\tt vlalev@uwaterloo.ca}.} 
\and 
Fernando Granha Jeronimo\thanks{Supported in part by NSF grants CCF-1254044 and CCF-1816372.  {\tt University of Chicago}. {\tt granha@uchicago.edu}. } 
\and
Dylan Quintana\thanks{{\tt University of Chicago.} {\tt dquintana@uchicago.edu}}
\and
Shashank Srivastava\thanks{{\tt TTIC.} {\tt shashanks@ttic.edu }}
\and
Madhur Tulsiani\thanks{Supported by NSF grants CCF-1254044 and CCF-1816372. {\tt TTIC}. {\tt madhurt@ttic.edu}} 
}

\setcounter{page}{0}

\date{}

\maketitle
\draftbox
\thispagestyle{empty}

We consider families of codes obtained by "lifting" a base code
$\mathcal{C}$ through operations such as $k$-XOR applied to "local
views" of codewords of $\mathcal{C}$, according to a suitable
$k$-uniform hypergraph. The $k$-XOR operation yields the direct sum
encoding used in works of [Ta-Shma, STOC 2017] and [Dinur and Kaufman,
FOCS 2017].

We give a general framework for list decoding such lifted codes, as
long as the base code admits a unique decoding algorithm, and the
hypergraph used for lifting satisfies certain expansion properties. We
show that these properties are indeed satisfied by the collection of
length $k$ walks on a sufficiently strong expanding graph, and by
hypergraphs corresponding to high-dimensional expanders.
Instantiating our framework, we obtain list decoding algorithms for
direct sum liftings corresponding to the above hypergraph families.
Using known connections between direct sum and direct product, we also
recover (and strengthen) the recent results of Dinur et al. [SODA
2019] on list decoding for direct product liftings.

Our framework relies on relaxations given by the Sum-of-Squares (SOS)
SDP hierarchy for solving various constraint satisfaction problems
(CSPs). We view the problem of recovering the closest codeword to a
given (possibly corrupted) word, as finding the optimal solution to an
instance of a CSP. Constraints in the instance correspond to edges of
the lifting hypergraph, and the solutions are restricted to lie in the
base code $\mathcal{C}$. We show that recent algorithms for
(approximately) solving CSPs on certain expanding hypergraphs by some
of the authors also yield a decoding algorithm for such lifted codes.

We extend the framework to list decoding, by requiring the SOS
solution to minimize a convex proxy for negative entropy. We show that
this ensures a covering property for the SOS solution, and the
"condition and round" approach used in several SOS algorithms can then
be used to recover the required list of codewords.

\newpage

\ifnotconf
\pagenumbering{roman}
\tableofcontents
\clearpage
\fi

\pagenumbering{arabic}
\setcounter{page}{1}

\section{Introduction}\label{sec:intro}

We consider the problem of list decoding binary codes obtained by
starting with a binary base code $\Cc$ and amplifying its distance by
``lifting'' $\Cc$ to a new code $\Cc'$ using an expanding or
pseudorandom structure.
Examples of such constructions include \emph{direct products} where
one lifts (say) $\Cc \subseteq \F_2^n$ to $\Cc' \subseteq
(\F_2^k)^{n^k}$ with each position in $y \in \Cc'$ being a $k$-tuple
of bits from $k$ positions in $z \in \Cc$.
Another example is \emph{direct sum} codes where
$\Cc' \subseteq \F_2^{n^k}$ and each position in $y$ is the parity of
a $k$-tuple of bits in $z \in \Cc$.
Of course, for many applications, it is interesting to consider a
small ``pseudorandom'' set of $k$-tuples, instead of considering the
complete set of size $n^k$.

This kind of distance amplification is well known in coding
theory~\cite{ABNNR92, IW97, GI01, Ta-Shma17} and it can draw on the
vast repertoire of random and pseudorandom expanding
objects \cite{HooryLW06,Lubotzky18}.
Such constructions are also known to have several applications to the
theory of Probabilitically Checkable Proofs
(PCPs) \cite{ImpagliazzoKW09, DinurS14, DDGEKS15, Chan16, A02:icm}.
However, despite having several useful properties, it might not always
be clear how to \emph{decode} the codes resulting from such
constructions, especially when constructed using sparse pseudorandom
structures.
An important example of this phenomenon is Ta-Shma's explicit
construction of binary codes of arbitrarily large distance near the
(non-constructive) Gilbert-Varshamov bound~\cite{Ta-Shma17}. Although
the construction is explicit, efficient decoding is not known.
Going beyond unique-decoding algorithms, it is also useful to have
efficient list-decoding algorithms for complexity-theoretic
applications \cite{S00,G01, STV99, T04:codes}.

The question of list decoding such pseudorandom constructions of
direct-product codes was considered by Dinur \etal \cite{DinurHKNT19},
extending a unique-decoding result of Alon \etal \cite{ABNNR92}. While
Alon \etal proved that the code is unique-decodable when the lifting
hypergraph (collection of $k$-tuples) is a good ``sampler'', Dinur \etal
showed that when the hypergraph has additional structure (which they
called being a ``double sampler'') then the code is also list
decodable. They also posed the question of understanding structural
properties of the hypergraph that might yield even unique decoding
algorithms for the \emph{direct sum} based liftings.

We develop a generic framework to understand properties of the
hypergraphs under which the lifted code $\Cc'$ admits efficient list
decoding algorithms, assuming only efficient unique decoding
algorithms for the base code $\Cc$. Formally, let $X$ be a
downward-closed hypergraph (simplicial complex) defined by taking the
downward closure of a $k$-uniform hypergraph, and let
$g:\F_2^k \to \F_2$ be any boolean function. $X(i)$ denotes the
collection of sets of size $i$ in $X$ and $X(\leq d)$ the collection
of sets of size at most $d$.  We consider the lift $\Cc'
= \lift_{X(k)}^g(\Cc)$, where $\Cc \subseteq \F_2^{X(1)}$ and
$\Cc' \subseteq \F_2^{X(k)}$, and each bit of $y \in \Cc'$ is obtained
by applying the function $g$ to the corresponding $k$ bits of
$z \in \Cc$. We study properties of $g$ and $X$ under which this
lifting admits an efficient list decoding algorithm.

We consider two properties of this lifting, \emph{robustness}
and \emph{tensoriality}, formally defined later, which are sufficient
to yield decoding algorithms. The first property (robustness)
essentially requires that for any two words in $\F_2^{X(1)}$ at a
moderate distance, the lifting amplifies the distance between them.
While the second property is of a more technical nature and is
inspired by the Sum-of-Squares (SOS) SDP hierarchy used for our
decoding algorithms, it is implied by some simpler combinatorial
properties. Roughly speaking, this combinatorial property,
which we refer to as \emph{splittability}, requires
that the graph on (say) $X(k/2)$ defined by connecting $\ess, \tee \in
X(k/2)$ if $\ess \cap \tee = \emptyset$ and $\ess \cup \tee \in X(k)$,
is a sufficiently good expander (and similarly for graphs on $X(k/4)$,
$X(k/8)$, and so on). Splittability requires that the $k$-tuples can be
(recursively) split into disjoint pieces such that at each step the
graph obtained between the pairs of pieces is a good expander.

\noindent \textbf{Expanding Structures.} \enspace
We instantiate the above framework with two specific structures: the
collection of $k$-sized hyperedges of a high-dimensional expander
(HDX) and the collection of length $k$ walks~\footnote{Actually, we
will be working with length $k-1$ walks which can be represented as
$k$-tuples, though this is an unimportant technicality.
The reason is to be consistent in the number of vertices
(allowing repetitions) with $k$-sized hyperedges.} on an expander graph.  HDXs are
downward-closed hypergraphs satisfying certain expansion properties.
We will quantify this expansion using Dinur and Kaufman's notion of a
$\gamma$-HDX~\cite{DinurK17}.

HDXs were proved to be splittable by some of the authors \cite{AJT19}.
For the expander walk instantiation, we consider a 
variant of splittability where a walk of length $k$ is split into two halves,
which are walks of length $k/2$ (thus we do \emph{not} consider all
$k/2$ size subsets of the walk). The spectrum of the graphs obtained
by this splitting can easily be related to that of the underlying
expander graph. In both cases, we take the function $g$ to be $k$-XOR
which corresponds to the direct sum lifting. We also obtain results
for direct product codes via a simple (and standard) reduction to the
direct sum case.

\noindent \textbf{Our Results.} \enspace
Now we provide a quantitative version of our main result. For this,
we split the main result into two cases (due to their difference in
parameters): HDXs and length $k$ walks on expander graphs. We start
with the former expanding object.

\begin{theorem}[Direct Sum Lifting on HDX (Informal)]\label{theo:dsum_hdx_intro}
  Let $\epsilon_0 < 1/2$ be a constant and $\epsilon \in (0,\epsilon_0)$.  Suppose
  $X(\le d)$ is a $\gamma$-HDX on $n$ vertices with $\gamma \le (\log(1/\epsilon))^{-O( \log(1/\epsilon))}$
  and $d =\Omega\left((\log(1/\epsilon))^2/\epsilon^2\right)$.

  For every linear code $\mathcal{C}_1 \subset \mathbb{F}_2^n$ with relative
  distance $\ge 1/2-\epsilon_0$, there exists a direct sum lifting
  $\mathcal{C}_k \subset \mathbb{F}_2^{X(k)}$
  with $k = O\left(\log(1/\epsilon)\right)$ and relative distance $\ge 1/2-\epsilon^{\Omega_{\epsilon_0}(1)}$
  satisfying the following:
  \begin{itemize}
    \item{[Efficient List Decoding]} If $\tilde{y}$ is \lict{\epsilon} to $\mathcal{C}_k$, then we can compute the list of all the codewords of
                                     $\mathcal{C}_k$ that are \lict{\epsilon} to $\tilde{y}$ in time $n^{\epsilon^{-O\left(1\right)}} \cdot f(n)$,
                                     where $f(n)$ is the running time of a unique decoding algorithm for $\mathcal{C}_1$.
    \item{[Rate]} The rate \footnote{In the rate computation, $X(k)$ is viewed as a multi-set where each $\ess \in X(k)$ is repeated a certain number of times for technical reasons.} $r_k$ of $\mathcal{C}_k$ is  $r_k = r_1 \cdot \abs{X(1)}/\abs{X(k)}$,
                  where $r_1$ is the rate of $\mathcal{C}_1$.
  \end{itemize}
\end{theorem}

A consequence of this result is a method of decoding the direct
product lifting on a HDX via a reduction to the direct sum case.

\begin{corollary}[Direct Product Lifting on HDX (Informal)]\label{cor:dprod_hdx_intro}
Let $\epsilon_0 < 1/2$ be a constant and $\epsilon > 0$. Suppose $X(\leq d)$ is a $\gamma$-HDX on $n$ vertices with
$\gamma \leq (\log(1/\epsilon))^{-O(\log(1/\epsilon))}$ and $d = \Omega((\log(1/\epsilon))^2/\epsilon^2)$.

For every linear code $\mathcal C_1 \subset \F_2^n$ with relative distance $\geq 1/2 - \epsilon_0$, there exists a
direct product encoding $\mathcal C_{\ell} \subset (\F_2^{\ell})^{X(\ell)}$ with $\ell = O(\log(1/\epsilon))$ that can
be efficiently list decoded up to distance $(1 - \epsilon)$.
\end{corollary}

\begin{remark}
  List decoding the direct product lifting was first established by
  Dinur et al. in \cite{DinurHKNT19} using their notion of double samplers.
  Since constructions of
  double samplers are only known using HDXs, we can compare some
  parameters. In our setting, we obtain $d =
  O(\log(1/\epsilon)^2/\epsilon^2)$ and $\gamma =
  (\log(1/\epsilon))^{-O(\log(1/\epsilon))}$ whereas
  in~\cite{DinurHKNT19} $d = O(\exp(1/\epsilon))$ and $\gamma =
  O(\exp(-1/\epsilon))$.
\end{remark}

Given a graph $G$, we denote by $W_G(k)$ the collection of
all length $k-1$ walks of $G$, which plays the role of the local views
$X(k)$. If $G$ is sufficiently expanding, we have the following
result.

\begin{theorem}[Direct Sum Lifting on Expander Walks (Informal)]\label{theo:dsum_expander_walk_intro}
  Let $\epsilon_0 < 1/2$ be a constant and $\epsilon \in (0,\epsilon_0)$.  Suppose
  $G$ is a $d$-regular $\gamma$-two-sided spectral expander graph on
  $n$ vertices with $\gamma \le \epsilon^{O(1)}$.

  For every linear code $\mathcal{C}_1 \subset \mathbb{F}_2^n$ with relative distance $\ge 1/2-\epsilon_0$,
  there exists a direct sum encoding $\mathcal{C}_k \subset \mathbb{F}_2^{W_G(k)}$ 
  with $k = O\left(\log(1/\epsilon)\right)$ and relative distance $\ge 1/2-\epsilon^{\Omega_{\epsilon_0}(1)}$ satisfying the following:

  \begin{itemize}
    \item{[Efficient List Decoding]} If $\tilde{y}$ is \lict{\epsilon} to $\mathcal{C}_k$, then we can compute the list of all the codewords of
                                     $\mathcal{C}_k$ that are \lict{\epsilon} to $\tilde{y}$ in time $n^{\epsilon^{-O\left(1\right)}} \cdot f(n)$,
                                     where $f(n)$ is the running time of a unique decoding algorithm for $\mathcal{C}_1$.
   
    \item{[Rate]} The rate $r_k$ of $\mathcal{C}_k$ is $r_k = r_1 / d^{k-1}$, where $r_1$ is the rate of $\mathcal{C}_1$.
  \end{itemize}
\end{theorem}

The results in \cref{theo:dsum_hdx_intro}, \cref{cor:dprod_hdx_intro},
and \cref{theo:dsum_expander_walk_intro} can all be extended (using a
simple technical argument) to nonlinear base codes $\mathcal C_1$ with
similar parameters. We also note that
applying \cref{theo:dsum_hdx_intro} to explicit objects derived from
Ramanujan complexes~\cite{LubotzkySV05a,LubotzkySV05b} and applying
\cref{theo:dsum_expander_walk_intro} to Ramanujan graphs~\cite{LPS88} yield explicit
constructions of codes with constant relative distance and rate,
starting from a base code with constant relative distance and rate.
With these constructions, the rate of the lifted code satisfies
$r_k \geq
r_1 \cdot \exp\left(-(\log(1/\epsilon))^{O(\log(1/\epsilon))}\right)$
in the HDX case and $r_k \geq
r_1 \cdot \epsilon^{O(\log(1/\epsilon))}$ for expander walks.  The
precise parameters of these applications are given
in~\cref{cor:ramanujan_instantiation} of~\cref{sec:list_dec_xor_hdx}
and in~\cref{cor:ramanujan_instantiation_exp}
of~\cref{sec:expander_walks}, respectively.

\noindent \textbf{Our techniques.} \enspace
We connect the question of decoding lifted codes to finding good
solutions for instances of Constraint Satisfaction Problems (CSPs)
which we then solve using the Sum-of-Squares (SOS) hierarchy. Consider
the case of direct sum lifting, where for the lifting $y$ of a
codeword $z$, each bit of $y$ is an XOR of $k$ bits from $z$. If an
adversary corrupts some bits of $y$ to give $\tilde{y}$, then finding
the closest codeword to $\tilde{y}$ corresponds to finding
$z' \in \Cc$ such that appropriate $k$-bit XORs of $z'$ agree with as
many bits of $\tilde{y}$ as possible. If the corruption is small, the
distance properties of the code ensure that the unique choice for
$z'$ is $z$. Moreover, the distance amplification (robustness)
properties of the lifting can be used to show that it suffices to
find \emph{any} $z'$ (not necessarily in $\Cc$) satisfying
sufficiently many constraints. We then use results by a subset of the
authors \cite{AJT19} showing that splittability (or the tensorial
nature) of the hypergraphs used for lifting can be used to yield
algorithms for approximately solving the related CSPs. Of course, the
above argument does not rely on the lifting being direct sum and works
for any lifting function $g$.

For list decoding, we solve just a single SOS program whose solution
is rich enough to ``cover'' the list of codewords we intend to
retrieve. In particular, the solutions to the CSP are obtained by
``conditioning'' the SDP solution on a small number of variables, and
we try to ensure that in the list decoding case, conditioning the SOS
solution on different variables yields solutions close to different
elements of the list.
To achieve this covering property we consider a convex proxy $\Psi$
for negative entropy measuring how concentrated (on a few codewords)
the SOS solution is.  Then we minimize $\Psi$ while solving the SOS
program.
A similar technique was also independently used by Karmalkar, Klivans,
and Kothari~\cite{KarmalkarKK19} and
Raghavendra--Yau~\cite{RaghavendraY19} in the context of learning
regression.
Unfortunately, this SOS cover comes with only some weak guarantees
which are, a priori, not sufficient for list decoding.
However, again using the robustness property of the lifting, we are
able to convert weak covering guarantees for the lifted code $\Cc'$ to
strong guarantees for the base code $\Cc$, and then appeal to the
unique decoding algorithm.
We regard the interplay between these two properties leading to the
final list decoding application as our main technical contribution. A
more thorough overview is given in~\cref{sec:strategy} after
introducing some objects and notation
in~\cref{sec:prelim}. In~\cref{sec:strategy}, we also give further
details about the organization of the document.

\noindent \textbf{Related work.} \enspace
The closest result to ours is the list decoding framework of Dinur et
al.~\cite{DinurHKNT19} for the direct product encoding, where the
lifted code is not binary but rather over the alphabet $\F_2^k$.
Our framework instantiated for the direct sum encoding on HDXs
(c.f.~\cref{theo:dsum_hdx_intro}) captures and strengthens some of
their parameters in~\cref{cor:dprod_hdx_intro}.
While Dinur \etal also obtain list decoding by solving an SDP for a
specific CSP (Unique Games), the reduction to CSPs in their case uses
the combinatorial nature of the double sampler instances and is also
specific to the direct product encoding.
They recover the list by iteratively solving many CSP instances, where
each newly found solution is pruned from the instance by reducing the
alphabet size by one each time.
On the other hand, the reduction to CSPs is somewhat generic in our
framework and the recovery of the list is facilitated by including an
entropic proxy in the convex relation.
As mentioned earlier, a similar entropic proxy was also
(independently) used by Karmalkar et al.~\cite{KarmalkarKK19} and
Raghavendra--Yau~\cite{RaghavendraY19} in the context of list decoding
for linear regression and mean estimation.
Direct products on expanders were also used as a building block by
Guruswami and Indyk \cite{GI03} who used these to
construct \emph{linear time} list decodable codes over large
alphabets. They gave an algorithm for recovering the list based on
spectral partitioning techniques.

\section{Preliminaries}\label{sec:prelim}

\subsection{Simplicial Complexes}\label{subsec:complexes}

It will be convenient to work with hypergraphs satisfying a certain
downward-closed property (which is straightforward to obtain).

\begin{definition}
A {\deffont simplicial complex} $X$ with ground set $[n]$ is a
downward-closed collection of subsets of $[n]$, \ie for all sets $\ess
\in X$ and $\tee \subseteq \ess$, we also have $\tee \in X$. The sets
in $X$ are referred to as {\deffont faces} of $X$.
We use the notation $X(i)$ for the set of all faces of a simplicial
complex $X$ with cardinality $i$ and $X(\leq d)$ for the set of all
faces of cardinality at most $d$.~\footnote{Note that it is more
  common to associate a geometric representation to simplicial
  complexes, with faces of cardinality $i$ being referred to as faces
  of \emph{dimension} $i-1$ (and the collection being denoted by
  $X(i-1)$ instead of $X(i)$). However, we prefer to index faces by
  their cardinality to improve readability of related expressions.}
By convention, we take $X(0) := \set{\varnothing}$.

A simplicial complex $X(\leq d)$ is said to be a {\deffont pure
  simplicial complex} if every face of $X$ is contained in some face
of size $d$. Note that in a pure simplicial complex $X(\leq d)$, the
top slice $X(d)$ completely determines the complex.
\end{definition}

Simplicial complexes are equipped with the following probability
measures on their sets of faces.

\begin{definition}[Probability measures $(\Pi_1, \ldots, \Pi_d)$]
Let $X(\leq d)$ be a pure simplicial complex and let $\Pi_d$ be an arbitrary probability measure on
$X(d)$. We define a coupled array of random variables $(\ess^{(d)}, \ldots,
\ess^{(1)})$ as follows: sample $\ess^{(d)} \sim \Pi_d$ and (recursively) for
each $i \in [d]$, take $\ess^{(i-1)}$ to be a uniformly random subset of
$\ess^{(i)}$ of size $i-1$.
The distributions $\Pi_{d-1}, \ldots, \Pi_1$ are then defined to be the marginal
distributions of the random variables $\ess^{(d-1)}, \ldots, \ess^{(1)}$. We also define
the joint distribution of $(\ess^{(d)}, \ldots, \ess^{(1)})$ as $\Pi$.
Note that the choice of $\Pi_d$ determines each other distribution $\Pi_i$ on $X(i)$.
\end{definition}

In order to work with the HDX and expander walk instantiations in a
unified manner, we will use also use the notation $X(k)$ to indicate
the set of all length $k-1$ walks on a graph $G$. In this case, $X(k)$
is a set of $k$-tuples rather than subsets of size $k$.  This
distinction will be largely irrelevant, but we will use $W_G(k)$ when
referring specifically to walks rather than subsets.  The set of walks
$W_G(k)$ has a corresponding distribution $\Pi_k$ as well (see
\cref{def:walk_collection}).

\subsection{Codes and Lifts}

\subsubsection*{Codes}

We briefly recall some standard code terminology. Let $\Sigma$ be a
finite alphabet with $q \in \mathbb{N}$ symbols. We will be mostly
concerned with the case $\Sigma=\mathbb{F}_2$. Given $z,z' \in
\Sigma^n$, recall that the relative Hamming distance between $z$ and
$z'$ is \text{$\Delta(z,z') \coloneqq \abs{\set{i \mid z_i\ne
      z_i'}}/n$}. Any set $\Cc \subset \Sigma^n$ gives rise to a
$q$-ary code. The distance of $\Cc$ is defined as $\Delta(\Cc)
\coloneqq \min_{z\ne z'} \Delta(z,z')$ where $z,z' \in \Cc$. We say
that $\Cc$ is a linear code~\footnote{In this case, $q$ is required to
  be a prime power.}  if $\Sigma = \mathbb{F}_q$ and $\Cc$ is a linear
subspace of $\mathbb{F}_q^n$. The rate of $\Cc$ is
$\log_q(\abs{\Cc})/n$.

Instead of discussing the distance of a binary code, it will often be more natural to phrase results in terms of its bias.

\begin{definition}[Bias]
  The {\deffont bias} of a word~\footnote{Equivalently, the bias of $z \in \{\pm 1\}^n$ is $\bias(z) \coloneqq \abs{\E_{i \in [n]} z_i}$.}
  $z \in \F_2^n$ is $\bias(z) \coloneqq \abs{\E_{i \in [n]} (-1)^{z_i}}$. The bias of a code $\Cc$ is the
  maximum bias of any non-zero codeword in $\Cc$.
\end{definition}

\subsubsection*{Lifts} \label{subsec:lifts}

Starting from a code $\Cc_1 \subset \Sigma^{X(1)}_1$, we amplify its
distance by considering a \textit{lifting} operation defined as
follows.

\begin{definition}[Lifting Function]
Let $g: \Sigma_1^k \to \Sigma_k$ and $X(k)$ be a collection of $k$-uniform hyperedges or walks of length $k-1$ on the set $X(1)$.
For $z \in \Sigma_1^{X(1)}$, we define $\lift_{X(k)}^g(z) = y$ such that $y_{\ess} = g(z\vert_{\ess})$ for all $\ess \in X(k)$, where $z\vert_{\ess}$ is the restriction of $z$ to the indices in $\ess$.

The lifting of a code $\mathcal C_1 \subseteq \Sigma_1^{X(1)}$ is
	$$\lift_{X(k)}^g(\mathcal C_1) = \{\lift_{X(k)}^g(z) \mid z \in \mathcal C_1\},$$
which we will also denote $\mathcal C_k$.
We will omit $g$ and $X(k)$ from the notation for lifts when they are clear from context.
\end{definition}

We will call liftings that amplify the distance of a code \emph{robust}.
\begin{definition}[Robust Lifting]\label{def:lift_rob}
  We say that $\lift_{X(k)}^g$ is $(\delta_0,\delta)$-robust
  if for every $z,z' \in \Sigma_1^{X(1)}$ we have
  $$
  \Delta(z,z') \ge \delta_0 \implies \Delta(\lift(z), \lift(z')) \ge \delta.
  $$
\end{definition}

For us the most important example of lifting is when the function $g$
is \text{$k$-XOR} and $\Sigma_1=\Sigma_k=\mathbb{F}_2$, which has been
extensively studied in connection with codes and
otherwise~\cite{Ta-Shma17, STV99, GNW95, ABNNR92}. In our language of
liftings, $k$-XOR corresponds to the \emph{direct sum lifting}.

\begin{definition}[Direct Sum Lifting]
 Let $\mathcal C_1 \subseteq \F_2^n$ be a base code on $X(1) = [n]$.
 The {\deffont direct sum lifting} of a word $z \in \F_2^n$ on a
 collection $X(k)$ is $\dsum_{X(k)}(z) = y$ such that $y_{\ess} = \sum_{i \in \ess} z_i$ 
 for all $\ess \in X(k)$.
\end{definition}

We will be interested in cases where the direct sum lifting reduces
the bias of the base code; in~\cite{Ta-Shma17}, structures with such a
property are called \emph{parity samplers}, as they emulate the reduction in
bias that occurs by taking the parity of random samples.

\begin{definition}[Parity Sampler]
  Let $g \colon \mathbb{F}_2^k \to \mathbb{F}_2$. We say that
    $\lift_{X(k)}^g$ is an $(\beta_0, \beta)$-parity sampler if
  for all $z \in \F_2^{X(1)}$ with $\bias(z) \leq \beta_0$, we have
  $\bias(\lift(z)) \leq \beta$.
\end{definition}

\subsection{Constraint Satisfaction Problems (CSPs)}
A $k$-CSP instance $\Ins(H, \Pred, w)$ with alphabet size $q$
consists of a $k$-uniform hypergraph $H$, a set of constraints 
\[ \Pred= \set*{ \Pred_\aye \subseteq [q]^{\aye} : \aye \in H},\]
and a non-negative weight function $w \in \RR_+^H$ on the
constraints satisfying $\sum_{\aye \in H}w(a) = 1$.

We will think of the constraints as predicates that are satisfied by an assignment
$\assn$ if we have $\assn|_{\aye} \in \Pred_\aye$,
\ie~the restriction of $\assn$ on $\aye$ is contained in
$\Pred_\aye$. We write $\SAT_\Ins(\assn)$ for the (weighted) fraction
of the constraints satisfied by the assignment $\assn$, \ie
\[ 
\SAT_\Ins(\assn) 
~=~ \sum_{\aye \in H} w(\aye) \cdot \one[\assn|_\aye \in \Pred_\aye] 
~=~ \Ex{\aye \sim w}{\one[\assn|_\aye \in \Pred_\aye]}
\mper
\]
We denote by $\OPT(\Ins)$ the maximum of $\SAT_\Ins(\assn)$ over all
$\assn \in [q]^{V(H)}$.

A particularly important class of $k$-CSPs for our work will be
$k$-XOR: here the input consists of a $k$-uniform hypergraph $H$ with
weighting $w$, and a (right-hand side) vector $r \in \FF_2^H$. The
constraint for each $\aye \in H$ requires
\[ \sum_{i \in \aye} \sigma(i) = r_\aye \pmod{2}.\]
In this case we will use the notation $\Ins(H, r, w)$ to refer to the
$k$-XOR instance. When the weighting $w$ is implicitly clear, we will
omit it and just write $\Ins(H, r)$. 

Any $k$-uniform hypergraph $H$ can be associated with a pure
simplicial complex in a canonical way by setting $X_{\Ins} = \set*{\bee
  : \exists\ \aye \in H ~\text{with}~ \aye \supseteq \bee }$; notice that
$X_{\Ins}(k) = H$. 
We will refer to this complex as the \emph{constraint complex} of the instance $\Ins$.
The probability distribution $\Pi_{k}$ on
$X_{\Ins}(k)$ will be derived from the weight function $w$ of the
constraint:
\[ \Pi_{k}(\aye) = w(\aye) \quad\forall \aye \in X_{\Ins}(k) = H.\]
\subsection{Sum-of-Squares Relaxations and $t$-local PSD Ensembles}\label{subsec:sos-relax}

\newcommand{\V}[2]{{v}_{(#1,#2)}}
\newcommand{\W}[2]{{w}_{(#1,#2)}}
\newcommand{\Vempty}{\V{\emptyset}{\emptyset}}

The Sum-of-Squares (SOS) hierarchy gives a sequence of increasingly
tight semidefinite programming relaxations for several optimization
problems, including CSPs. Since we will use relatively few facts about
the SOS hierarchy, already developed in the analysis of Barak,
Raghavendra, and Steurer \cite{BarakRS11}, we will adapt their notation
of \emph{$t$-local distributions} to describe the relaxations. For
a $k$-CSP instance $\Ins = (H,\Pred, w)$ on $n$ variables, we consider
the following semidefinite relaxation given by $t$-levels of the SOS
hierarchy, with vectors $\V{S}{\alpha}$ for all $S \subseteq [n]$ with
$\abs{S} \leq t$, and all $\alpha \in [q]^S$. Here, for $\alpha_1 \in
\qary^{S_1}$ and $\alpha_2 \in \qary^{S_2}$, $\alpha_1 \circ \alpha_2
\in \qary^{S_1 \cup S_2}$ denotes the partial assignment obtained by
concatenating $\alpha_1$ and $\alpha_2$.
\begin{table}[h]
\hrule
\vline
\begin{minipage}[t]{0.99\linewidth}
\vspace{-5 pt}
{\small
\begin{align*}
\mbox{maximize}\quad ~~
\Ex{\aye \sim w}{\sum_{\alpha \in \Pred_{\aye}} \smallnorm{\V{\aye}{\alpha}}^2}&~=:~ \SDP(\Ins)\\
\mbox{subject to}\quad \quad ~
 \mydot{\V{S_1}{\alpha_1}}{\V{S_2}{\alpha_2}} &~=~ 0 
  & \forall~ \alpha_1|_{S_1 \cap S_2} \neq \alpha_2|_{S_1 \cap S_2}\\
 \mydot{\V{S_1}{\alpha_1}}{\V{S_2}{\alpha_2}} 
  &~=~ \mydot{\V{S_3}{\alpha_3}}{\V{S_4}{\alpha_4}}
  & \forall~ S_1 \cup S_2 = S_3 \cup S_4, 
  ~\alpha_1 \circ \alpha_2 = \alpha_3 \circ \alpha_4 \\
 \sum_{j \in [q]} \smallnorm{\V{\{i\}}{j}}^2 &~=~ 1 &\forall i \in [n]\\
 \smallnorm{\Vempty}^2 &~=~ 1 &
\end{align*}}
\vspace{-14 pt}
\end{minipage}
\hfill\vline
\hrule
\end{table}

For any set $S$ with $|S| \leq t$, the vectors $\V{S}{\alpha}$ induce
a probability distribution $\mu_S$ over $\qary^S$ such that the
assignment $\alpha \in \qary^S$ appears with probability
$\smallnorm{\V{S}{\alpha}}^2$.
Moreover, these distributions are consistent on intersections: for
$T \subseteq S \subseteq [n]$, we have $\mu_{S|T} = \mu_T$, where
$\mu_{S|T}$ denotes the restriction of the distribution $\mu_S$ to the
set $T$.
We use these distributions to define a collection of random variables
$\rv Z_1, \ldots, \rv Z_n$ taking values in $\qary$, such that for any
set $S$ with $\abs{S} \leq t$, the collection of variables
$\inbraces{\rv Z_i}_{i \in S}$ has a joint distribution $\mu_S$. Note
that the entire collection $(\rv Z_1, \ldots, \rv Z_n)$ \emph{may not}
have a joint distribution: this property is only true for
sub-collections of size $t$. We will refer to the collection $(\rv
Z_1, \ldots, \rv Z_n)$ as a \emph{$t$-local ensemble} of random
variables.

We also have that that for any $T \subseteq [n]$ with $\abs{T} \leq
t-2$, and any $\xi \in \qary^T$, we can define a $(t-\abs{T})$-local
ensemble $(\rv Z_1', \ldots, \rv Z_n')$ by ``conditioning'' the local
distributions on the event $\rv Z_T = \xi$, where $\rv Z_T$ is
shorthand for the collection $\inbraces{\rv Z_i}_{i \in T}$. For any
$S$ with $\abs{S} \leq t-\abs{T}$, we define the distribution of $\rv
Z_S'$ as $\mu_S' := \mu_{S \cup T} | \set{\rv Z_T = \xi}$.
Finally, the semidefinite program also ensures that for any such
conditioning, the conditional covariance matrix
\[
\Emm_{(S_1, \alpha_1)(S_2,\alpha_2)} ~=~ \cov\inparen{\one[\rv Z_{S_1}' = \alpha_1],
  \one[\rv Z_{S_2}' = \alpha_2]}
\]
is positive semidefinite, where $\abs{S_1}, \abs{S_2} \leq
(t-\abs{T})/2$.  Here, for each pair $S_1, S_2$ the covariance is
computed using the joint distribution $\mu_{S_1 \cup S_2}'$.
In this paper, we will only consider $t$-local ensembles such that for
every conditioning on a set of size at most $t-2$, the conditional
covariance matrix is PSD. We will refer to these as \emph{$t$-local PSD ensembles}.
We will also need a simple corollary of the above definitions.
\begin{fact}\label{fact:set-ensemble}
  Let $(\rv Z_1, \ldots, \rv Z_n)$ be a $t$-local PSD ensemble, and
  let $X$ be any collection with $X(1)=[n]$.  Then, for all $s \leq
  t/2$, the collection $\inbraces{\rv Z_{\aye}}_{\aye \in X(\le s)}$
  is a $(t/s)$-local PSD ensemble, where $X(\le s) = \bigcup_{i = 1}^s
  X(i)$.
\end{fact}
For random variables $\rv Z_S$ in a $t$-local PSD ensemble, we use the
notation $\inbraces{\rv Z_S}$ to denote the distribution of $\rv Z_S$
(which exists when $\abs{S} \leq t$). We also define $\Var{\rv Z_S}$
as
$$
\Var{\rv Z_S} \coloneqq \sum_{\alpha \in {\qary^S}} \Var{\one\insquare{\rv Z_S = \alpha}}.
$$

\subsubsection*{Pseudo-expectation Formulation}

An equivalent way of expressing this local PSD ensemble is through the
use of a pseudo-expectation operator, which is also a language commonly
used in the SOS literature (e.g., \cite{BHKKMP16,BKS17}). The
exposition of some of our results is cleaner in this equivalent
language. Each variable $\rv Z_i$ with $i \in [n]$ is modeled by a
collection of indicator local random variables~\footnote{Note that
  $\set{\rv Z_{i,a}}_{i \in [n],a \in [q]}$ are formal variables in
  the SOS formulation.}  $\set{\rv Z_{i,a}}_{a \in [q]}$ with the
intent that $\rv Z_{i,a} = 1$ iff $\rv Z_i = a$. To ensure they behave
similarly to indicators we add the following restrictions to the SOS
formulation:
\begin{align*}
    \rv Z_{i,a}^2 & = \rv Z_{i,a} & \forall i \in [n], a \in [q]\\
    \sum_{a \in [q]} \rv Z_{i,a} & = 1 & \forall i \in [n]
\end{align*}
Let $\mathcal{R} = \mathbb{R}[\rv Z_{1,1},\dots, \rv Z_{n,q}]$ be the
ring of polynomials on $\set{\rv Z_{i,a}}_{i\in[n],a\in[q]}$. We will write
$\Ray^{\le d}$ for the restriction of $\Ray$ to polynomials of degree at most $d$. A
feasible solution at the $(2t)$-th level of the SOS hierarchy is a
linear operator $\widetilde{\mathbb{E}} : \mathcal{R}^{\le 2t} \to
\mathbb{R}$ called the pseudo-expectation operator. This operator
satisfies the following problem-independent constraints: (i)
$\widetilde{\mathbb{E}}[1] = 1$ (normalization) and (ii)
$\widetilde{\mathbb{E}}[P^2] \ge 0$ for every $P \in \mathcal{R}^{\le
  t}$ (non-negative on Sum-of-Squares)~\footnote{From condition (ii),
  we can recover the PSD properties from the local PSD ensemble
  definition.}. It also satisfies the problem-dependent constraints
$$
\widetilde{\mathbb{E}}\left[\rv Z_{i,a}^2 \cdot P\right] = \widetilde{\mathbb{E}}\left[\rv Z_{i,a} \cdot P\right] \qquad \text{and} \qquad \widetilde{\mathbb{E}}\left[\left(\sum_{a \in [q]} \rv Z_{i,a}\right) \cdot Q\right] = \widetilde{\mathbb{E}}\left[Q\right],
$$
for every $i \in [n]$, $a \in [q]$, $P \in \mathcal{R}^{\le 2t-2}$,
and $Q \in \mathcal{R}^{\le 2t-1}$. Note that for any collection of
local random variables $\rv Z_{i_1},\dots, \rv Z_{i_j}$ with $j \le
2t$ we have the joint distribution
$$
\ProbOp(\rv Z_{i_1}=a_1,\dots,\rv Z_{i_j}=a_j) = \widetilde{\mathbb{E}}\left[\rv Z_{i_1,a_1}\dots \rv Z_{i_j,a_j}\right].
$$
Even though we may not have a global distribution we can implement a
form of pseudo-expectation conditioning on a random variable $\rv Z_i$
taking a given value $a \in [q]$ as long as $\ProbOp\left[\rv Z_i =
  a\right] = \widetilde{\mathbb{E}}[\rv Z_{i,a} ] > 0$. This can be
done by considering the new operator $\widetilde{\mathbb{E}}_{\vert
  Z_{i}= a} \colon \mathcal{R}^{\le 2t-2} \to \mathbb{R}$ defined as
$\widetilde{\mathbb{E}}_{\vert \rv Z_{i}= a}[\cdot] = \widetilde{\mathbb{E}}[\rv Z_{i,a}^2 \cdot ]/\widetilde{\mathbb{E}}[\rv Z_{i,a}^2 ]$,
which is a valid pseudo-expectation operator at the $(2t-2)$-th level.
This conditioning can be naturally generalized to a set of variables
$S \subseteq [n]$ with $\abs{S} \le t$ satisfying $\rv Z_S = \alpha$
for some $\alpha \in [q]^{{S}}$.
\subsection*{Notation}\label{subsec:notation}

We make some systematic choices for our parameters in order to
syntactically stress their qualitative behavior.
\begin{itemize}\setlength\itemsep{0.0001in}
  \item $1/2-\epsilon_0$ is a lower bound on the distance of the base code $\mathcal{C}_1$.  
  \item $1/2-\epsilon$ is a lower bound on the distance of the lifted code $\mathcal{C}_k$.
  \item $\kappa$ is
      a parameter that will control the list-decodability of the lifted
        code $\Cc_k$.
  \item $\mu,\theta,\eta$ are parameters that can be made arbitrarily small by increasing
        the SOS degree and/or the quality of expansion.
  \item $\beta,\delta$ are arbitrary error parameters.
  \item $\lambda_1 \geq \lambda_2 \geq \cdots$ are the eigenvalues of a graph's adjacency matrix (in $[-1,1]$).
  \item $\sigma_1 \geq \sigma_2 \geq \cdots$ are the singular values of a graph's adjacency matrix (in $[0,1]$).
\end{itemize}

SOS is an analytic tool so we will identify~\footnote{For this, we can
  use any bijection from $\F_2 \to \set{\pm 1}$.} words over $\F_2$
with words over $\set{\pm 1}$. We also make some choices for words and
local variables to distinguish the ground space $\F_2^{X(1)}$ or
$\set{\pm 1}^{X(1)}$ form the lifted space $\F_2^{X(k)}$ or $\set{\pm
  1}^{X(k)}$.
\begin{itemize}\setlength\itemsep{0.0001in}
  \item $z,z',z'',\dots$ are words in the ground space $\F_2^{X(1)}$ or $\set{\pm 1}^{X(1)}$.
  \item $y,y',y'',\dots$ are words in the lifted space $\F_2^{X(k)}$ or $\set{\pm 1}^{X(k)}$
  \item $\rv Z \coloneqq \set{\rv Z_1,\dots, \rv Z_n}$ is a local PSD ensemble on the ground set $X(1)$.
  \item $\rv Y \coloneqq \set{\rv Y_{\ess} \coloneqq \left(\lift(\rv Z)\right)_{\ess} ~\vert~\ess \in X(k) }$ is a local ensemble on $X(k)$.
\end{itemize}

\section{Proof Strategy and Organization}\label{sec:strategy}

As discussed earlier, we view the problem of finding the closest
codeword(s) as that of finding suitable solution(s) to an instance of
a CSP (which is $k$-XOR in the case of direct sum). We now discuss
some of the technical ingredients required in the decoding procedure.

\paragraph{Unique Decoding.} Given $\Cc_k = \dsum_{X(k)}(\Cc_1)$ with the lifting function as
$k$-XOR, we can view the problem of finding the closest codeword to a
given $\tilde{y} \in \F_2^{X(k)}$ as that of finding the unique $z \in \Cc_1$
satisfying the maximum number of equations of the form
$\sum_{i \in \ess}z_i = \tilde{y}_{\ess}\pmod 2$, with one equation
for each $\ess \in X(k)$. By this property,
 $y=\dsum(z)$ is the unique codeword of $\Cc_k$ closest to
$\tilde{y}$. Using the results of \cite{AJT19}, it is
indeed possible to find $z' \in \F_2^n$ such that
$\Delta(\dsum(z'), \tilde y) \leq \Delta(\dsum(z), \tilde y)
+ \beta$ for any $\beta > 0$.
We then argue that $z'$ or its complement $\bar{z'}$ must be close to
$z \in \Cc_1$, which can then be recovered by unique decoding.

If this is not the case, then $z-z'$ must have bias bounded away from
1, which would imply by robustness (parity sampling property of the
hypergraph) that $\dsum(z-z')$ has bias close to zero, \ie
$\Delta(\dsum(z), \dsum(z')) \approx 1/2$. However, if
$\Delta(\tilde{y},\Cc_k) \leq \eta$, then we must have
\[
\Delta(\dsum(z), \dsum(z'))
~\leq~
\Delta(\dsum(z), \tilde{y})	+ \Delta(\dsum(z'), \tilde{y})	
~\leq~
2\eta + \beta \mcom
\]
which leads to a contradiction if $\eta$ is significantly below $1/4$
and $\beta$ is sufficiently small.

\paragraph{List Decoding.} 

We start by describing an abstract list decoding framework which only
assumes two general properties of a lifting $\lift_{X(k)}^g$: (i) it
is distance amplifying (\textit{robust}) and (ii) it is amenable to
SOS rounding (\textit{tensorial}).

Suppose $\tilde{y} \in \F_2^{X(k)}$ is a word promised to
be \lict{\sqrt{\epsilon}} to a lifted code $\mathcal{C}_k=\lift(\mathcal{C}_1)$
where $\mathcal{C}_k$ has distance at least $1/2-\epsilon$ and $\mathcal{C}_1$
has distance at least $1/2-\epsilon_0$. By list decoding $\tilde{y}$, we mean
finding a list $\mathcal{L} \subseteq \mathcal{C}_k$ of all codewords
\lict{\sqrt{\epsilon}} to $\tilde{y}$. 

Our framework for list decoding $\tilde{y}$ consists of three
stages. In the first stage, we set up and solve a natural SOS program
which we treat abstractly in this discussion~\footnote{The precise SOS
program used is given in~\cref{sec:sos_list_dec}.}.
One issue with using a rounding algorithm for this relaxation to do
list decoding is that this natural SOS program may return 
a solution that is ``concentrated'', e.g., a SOS solution
corresponding to single codeword in $\mathcal{L}$.
Such a solution will of course not have enough information to recover the entire list.
To address this issue we now ask not only for feasibility in our SOS
program but also to minimize a convex function $\Psi$ measuring how
concentrated the SOS solution is. Specifically, if $\rv Z$ is the PSD
ensemble corresponding to the solution of the SOS program and if $\rv
Y$ is the lifted ensemble, then we minimize $\Psi
~:=~ \E_{\ess, \tee \in X(k)}\sqbr*{\parens*{\PExp[{\rv Y}_{\ess} {\rv
Y}_{\tee}]}^2}$.

The key property of the function $\Psi$ is that if the SOS solution
``misses'' any element in the list $\mathcal{L}$ then it is possible
to decrease it. Since our solution is a minimizer~\footnote{Actually
an approximate minimizer is enough in our application.} of $\Psi$,
this is impossible. Therefore, our solution does ``cover'' the list
$\mathcal{L}$. Even with this SOS cover of $\mathcal{L}$, the list
decoding task is not complete. So far we have not talked about
rounding, which is necessary to extract codewords out of the
(fractional) solution. For now, we will simply assume that rounding is
viable (this is handled by the second stage of the framework) and
resume the discussion.

Unfortunately, the covering guarantee is somewhat weak, namely, for
$y \in \mathcal{L}$ we are only able to obtain a word $y' \in \F_2^{X(k)}$
with weak agreement $\abs{\ip{y'}{y}} \ge
2\cdot\epsilon$. Converting a word $y'$ from the cover into an actual
codeword $y$ is the goal of the third and final stage of the list
decoding framework, dubbed \textit{Cover Purification}.  At this point
we resort to the robustness properties of the lifting and the fact
that we actually have ``coupled'' pairs $(z,y=\lift(z))$ and
$(z',y'=\lift(z'))$ for some $z,z' \in \F_2^{X(1)}$. Due to
this robustness (and up to some minor technicalities) even a weak
agreement between $y$ and $y'$ in the lifted space translates into a
much stronger agreement between $z$ and $z'$ in the ground
space. Provided the latter agreement is sufficiently strong, $z'$ will
lie in the unique decoding ball centered at $z$ in $\mathcal{C}_1$. In
this case, we can uniquely recover $z$ and thus also
$y=\lift(z)$. Furthermore, if $\mathcal{C}_1$ admits an efficient
unique decoder, we can show that this step in list decoding $\tilde y$
can be done efficiently.

Now we go back to fill in the rounding step, which constitutes the
second stage of the framework, called \textit{Cover Retrieval}. We
view the SOS solution as composed of several ``slices'' from which the
weak pairs $(z',y')$ are to be extracted. Note that the framework
handles, in particular, $k$-XOR liftings where it provides not just a
single solution but a list of them. Hence, some structural assumption
about $X(k)$ is necessary to ensure SOS tractability.  Recall that
random $k$-XOR instances are hard for
SOS~\cite{Grigoriev01,KothariMOW17}.  For this reason, we impose a
sufficient tractability condition on $X(k)$ which we denote
the \textit{two-step tensorial} property. This notion is a slight
strengthening of a \textit{tensorial} property which was (implicitly)
first investigated by Barak et al.~\cite{BarakRS11} when $k=2$ and
later generalized for arbitrary $k \ge 2$ in~\cite{AJT19}. Roughly
speaking, if $X(k)$ is tensorial then the SOS local random variables
in a typical slice of the solution behave approximately as product
variables from the perspective of the local views $\ess \in X(k)$. A
two-step tensorial structure is a tensorial structure in which the
local random variables between pairs of local views $\ess,\tee \in
X(k)$ are also close to product variables, which is an extra property
required to perform rounding in this framework. With the two-step
tensorial assumption, we are able to round the SOS solution to obtain
a list of pairs $(z', y')$ weakly agreeing with elements of the code
list that will be refined during cover purification.

To recapitulate, the three stages of the abstract list decoding
framework are summarized in~\cref{fig:list_dec_framework} along with
the required assumptions on the lifting.

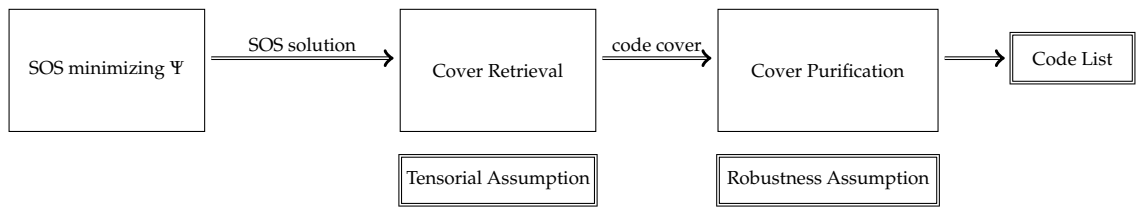
\begin{figure}[h!]
  \centering
  \begin{tikzpicture}[scale=0.65, every node/.style={scale=0.65}]]
    \draw  (-7,1) rectangle node {SOS minimizing $\Psi$}(-3,-1.5);
    \draw  (1,1) rectangle node {Cover Retrieval} (5,-1.5);
    \draw  (7.5,1) rectangle node {Cover Purification} (12,-1.5);
    \node (v1) at (-3,0) {};
    \node (v2) at (1,0) {};
    \node (v3) at (5,0) {};
    \node (v4) at (7.5,0) {};
    \draw [double] (v1) edge[double,->] node[above] {SOS solution} (v2);
    \draw [double] (v3) edge[double,->] node[above] {code cover} (v4);
    \node (v5) at (12,0) {};
    \draw [double] (1,-2) rectangle node {Tensorial Assumption} (5,-3);
    \draw [double] (7.5,-2) rectangle node {Robustness Assumption} (12,-3);
    \draw [double] (13.5,0.5) rectangle node {Code List} (16,-0.5);
    \node (v6) at (13.5,0) {};
    \draw [double] (v5) edge[double,->] (v6);
  \end{tikzpicture}
  \caption{List decoding framework with the assumptions required in
    each stage.}\label{fig:list_dec_framework}
\end{figure}

\paragraph{Finding suitable hypergraphs.} Fortunately, objects satisfying the necessary tensorial and robustness
assumptions do exist.  HDXs were shown to be tensorial
in~\cite{AJT19}, and here we strengthen this result to two-step
tensorial as well as prove that HDXs possess the particular robustness
property of parity sampling.
Walks on expander graphs are already known to be
robust~\cite{Ta-Shma17}, and we use a modified version of the methods
in~\cite{AJT19} to show they are also two-step tensorial.
For both HDXs and expander walks, we describe how to use known
constructions of these objects to get explicit direct sum encodings
that can be decoded using our abstract framework.

\paragraph{Reduction from direct product to direct sum.}
Finally, we describe how to use list decoding results for direct sum
codes to obtain results for direct product codes. Given a direct
product lifting $\Cc_k$ on the hypergraph $X(k)$, if
$\Delta(\tilde{y}, y) \leq 1-\epsilon$ for $y \in \Cc_k$, then we must
have that
\[
\Pr_{\ess \in X(k)}\insquare{y_{\ess} = \tilde{y}_{\ess}} 
~=~
\Ex{\ess \in X(k)}{\Ex{\tee \subseteq \ess}{\chi_{\tee}(y_{\ess}+\tilde{y}_{\ess})}} 
~\geq~ \epsilon \mper
\]
Since $\chi_{\tee}(y_{\ess})$ can be viewed as part of a direct sum lifting,
we get by grouping subsets $\tee$ by size that there must exist a size
$i$ such that the direct sum lifting using $X(i)$ has correlation at
least $\epsilon$ with the word $y'$ defined as $y'_{\tee}
= \chi_{t}(\tilde{y}_{\ess})$ for all $\tee \in X(i)$. We can then
apply the list decoding algorithm for direct sum codes on $X(i)$. A
standard concentration argument can also be used to control the size
$i$ to be approximately $k/2$.

\subsection*{Organization of Results}

In~\cref{sec:parity_sampling}, we show how the direct sum lifting on
HDXs can be used to reduce bias, establishing that HDXs are parity
samplers. This will give a very concrete running example of a lifting
that can be used in our framework. Before addressing list decoding, we
remark in~\cref{sec:unique_decoding} how this lifting can be used in
the simpler regime of unique decoding using a $k$-CSP algorithm on
expanding instances~\cite{AJT19}. The abstract list decoding framework
is given in~\cref{sec:list_dec}. Next, we instantiate the framework
with the direct sum lifting on HDXs in~\cref{sec:list_dec_xor_hdx}. As
an interlude between the first and second
instantiation,~\cref{sec:direct_prod} describes how the first concrete
instantiation of~\cref{sec:list_dec_xor_hdx} captures the direct
product lifting on HDXs via a reduction to the direct sum
lifting. Finally, in~\cref{sec:expander_walks}, we show how to
instantiate the framework with the direct sum lifting on the
collection of length $k-1$ walks of an expander graph.

\section{Pseudorandom Hypergraphs and Robustness of Direct Sum}\label{sec:parity_sampling}

The main robustness property we will consider is parity sampling
applied to the case of the direct sum lifting.  As this section
focuses on this specific instance of a lifting, here we will say that
a collection $X(k)$ is a parity sampler if its associated direct sum
lifting $\dsum_{X(k)}$ is a parity sampler.  Recall that for such a
parity sampler, the direct sum lifting brings the bias of a code close
to zero, which means it boosts the distance almost to 1/2. 

\subsection{Expander Walks and Parity Sampling}\label{subsec:expander_walks_parity_sampling}

A known
example of a parity sampler is the set $X(k)$ of all walks of length
$k$ in a sufficiently expanding graph, as shown by Ta-Shma.

\begin{theorem}[Walks on Expanders are Parity Samplers~\cite{Ta-Shma17}]\label{theo:expander_parity_sampler}
  Suppose $G$ is a graph with second largest singular value at most $\lambda$, and let $X(k)$ be the set of all walks of
  length $k$ on $G$.  Then $X(k)$ is a $(\beta_0, (\beta_0 +
  2\lambda)^{\lfloor k/2 \rfloor})$-parity sampler.
\end{theorem}

Our goal in this section is to prove a similar result for
high-dimensional expanders, where $X(k)$ is the set of $k$-sized
faces.

\subsection{High-dimensional Expanders}

A high-dimensional expander (HDX) is a particular kind of simplicial
complex satisfying an expansion requirement.  We recall the notion of
high-dimensional expansion considered in~\cite{DinurK17}.  For a
complex $X(\leq d)$ and $\ess \in X(i)$ for some $i \in [d]$, we
denote by $X_{\ess}$ the \emph{link complex}
\[
  X_\ess ~:=~ \set*{ \tee \backslash \ess \mid \ess \subseteq \tee \in X } \mper 
\]

When $\abs{\ess} \leq d-2$, we also associate a natural weighted graph
$G(X_{\ess})$ to a link $X_{\ess}$, with vertex set $X_{\ess}(1)$ and
edge set $X_{\ess}(2)$. The edge weights are taken to be proportional
to the measure $\Pi_2$ on the complex $X_{\ess}$, which is in turn
proportional to the measure $\Pi_{\abs{\ess}+2}$ on $X$. The graph
$G(X_{\ess})$ is referred to as the \emph{skeleton} of $X_{\ess}$.

Dinur and Kaufman~\cite{DinurK17} define high-dimensional expansion in
terms of spectral expansion of the skeletons of the links.
\begin{definition}[$\gamma$-HDX from~\cite{DinurK17}]\label{def:hdx_dk} 
  A simplicial complex $X(\leq d)$ is said to be {\deffont $\gamma$-High Dimensional Expander ($\gamma$-HDX)}
  if for every $0 \le i \le d-2$ and for every $\ess \in X(i)$, the graph $G(X_\ess)$ satisfies
  $\sigma_2(G(X_{\ess})) \leq \gamma$.
\end{definition}

We will need the following
theorem relating $\gamma$ to the spectral properties of the graph
between two layers of an HDX.

\begin{theorem}[Adapted from~\cite{DinurK17}]\label{theo:exp_up_walk_hdx}
  Let $X$ be a $\gamma$-HDX and let $M_{1,d}$ be the weighted
  bipartite containment graph between $X(1)$ and $X(d)$, where each edge
  $(\set{i}, \ess)$ has weight $(1/d) \Pi_d(\ess)$. Then the second largest singular value
  $\sigma_2$ of $M_{1,d}$ satisfies
	$$\sigma_2^2 \leq \frac{1}{d} + O(d \gamma).$$
\end{theorem}

We will be defining codes using HDXs by associating each face in some
$X(i)$ with a position in the code.  The distance between two
codewords does not take into account any weights on their entries,
which will be problematic when decoding since the distributions
$\Pi_i$ are not necessarily uniform.  To deal with this issue, we will
work with HDXs where the distributions $\Pi_i$ satisfy a property only
slightly weaker than uniformity.
\begin{definition}[Flatness (from~\cite{DinurHKNT19})]
  We say that a distribution $\Pi$ on a finite probability space
  $\Omega$ is $D$-flat if there exits $N$ such that each singleton
  $\omega \in \Omega$ has probability in $\{1/N, \dots, D/N\}$.
\end{definition}

Using the algebraically deep construction of Ramanujan complexes by
Lubotzky, Samuels, and Vishne \cite{LubotzkySV05a,LubotzkySV05b}, Dinur
and Kaufman \cite{DinurK17} showed that sparse $\gamma$-HDXs do exist,
with flat distributions on their sets of faces.  The following lemma
from~\cite{DinurHKNT19} is a refinement of~\cite{DinurK17}.

\begin{lemma}[Extracted from~\cite{DinurHKNT19}]\label{lemma:hdx_existence}
  For every $\gamma > 0$ and every $d \in \mathbb{N}$ there exists an
  explicit infinite family of bounded degree $d$-sized complexes
  which are $\gamma$-HDXs. Furthermore, there
  exists a $D \le (1/\gamma)^{O(d^2/\gamma^2)}$ such that
  $$
  \frac{\vert X(d) \vert}{\vert X(1) \vert} ~\le~ D,
  $$
  the distribution $\Pi_1$ is uniform, and the other distributions $\Pi_d, \dots, \Pi_2$ are $D$-flat.
\end{lemma}

For a $D$-flat distribution $\Pi_i$, we can duplicate each face in
$X(i)$ at most $D$ times to make $\Pi_i$ the same as a uniform
distribution on this multiset.  We will always perform such a
duplication implicitly when defining codes on $X(i)$.

\subsection{HDXs are Parity Samplers}

To prove that sufficiently expanding HDXs are parity samplers, we
establish some properties of the complete complex and then explore the
fact that HDXs are locally complete~\footnote{This a recurring theme
in the study of HDXs~\cite{DinurK17}.}. We first show that the
expectation over $k$-sized faces of a complete complex $X$ on $t$
vertices approximately splits into a product of $k$ expectations over
$X(1)$ provided $t \gg k^2$.

\begin{claim}[Complete complex and near independence]\label{claim:complete_near_indep}
Suppose $X$ is the complete complex of dimension at least $k$ with $\Pi_k$ uniform over $X(k)$ and $\Pi_1$ uniform over $X(1) = [t]$.
For a function $f: X(1) \to \R$, let
	$$\mu_k = \E_{\ess \sim \Pi_k} \left[\prod_{i \in \ess} f(i) \right] \qquad \text{and} \qquad \mu_1 = \E_{i \sim \Pi_1} \left[f(i)\right].$$
Then
 	$$\abs{\mu_k - \mu_1^k} \leq \frac{k^2}{t} \norm{f}_{\infty}^k.$$
\end{claim}

\begin{proof}
Let $\mathcal{E} = \{(i_1, \dots, i_k) \in X(1)^k \mid i_1, \dots, i_k \text{ are distinct}\}$,~$\delta = \ProbOp_{i_1, \dots, i_k \sim \Pi_1}[(i_1, \dots, i_k) \notin \mathcal{E}]$, and $\eta = \E_{(i_1, \dots, i_k) \in X(1)^k \setminus \mathcal{E}}[f(i_1) \cdots f(i_k)]$.
Then
\begin{align*}
	\mu_1^k &= \E_{i_1, \dots, i_k \sim \Pi_1}[f(i_1) \cdots f(i_k)] \\
	&= (1-\delta) \cdot \E_{(i_1, \dots, i_k) \in \mathcal{E}}[f(i_1) \cdots f(i_k)] + \delta \cdot \E_{(i_1, \dots, i_k) \in X(1)^k \setminus \mathcal{E}}[f(i_1) \cdots f(i_k)] \\
	&= (1-\delta) \cdot \mu_k + \delta \cdot \eta,
\end{align*}
where the last equality follows since $\Pi_k$ is uniform and the product in the expectation is symmetric.
As $i_1, \dots, i_k$ are sampled independently from $\Pi_1$, which is uniform over $X(1)$,
	$$\delta = 1 - \prod_{j < k} \left(1 - \frac jt \right) \leq \sum_{j < k} \frac jt = \frac{k(k-1)}{2t},$$
so we have
	$$\abs{\mu_k - \mu_1^k} = \delta \abs{\mu_k - \eta} \leq \frac{k^2}{2t} \left(2 \norm{f}_{\infty}^k \right).$$
\end{proof}

We will derive parity sampling for HDXs from their behavior as
samplers.  A sampler is a structure in which the average of any
function on a typical local view is close to its overall average.
More precisely, we have the following definition.

\begin{definition}[Sampler]
  Let $G = (U,V,E)$ be a bipartite graph with a probability distribution $\Pi_U$ on $U$.
  Let $\Pi_V$ be the distribution on $V$ obtained by choosing $u \in U$ according to $\Pi_U$, then a uniformly random neighbor $v$ of $u$.
  We say that $G$ is an $(\eta,\delta)$-sampler if for every function $f \colon V \to [0,1]$ with $\mu = \E_{v\sim \Pi_V} f(v)$,
	$$\ProbOp_{u \sim \Pi_U}\left[\left\vert \E_{v \sim u} [f(v)] - \mu \right\vert \ge \eta \right] ~\le~ \delta.$$
\end{definition}

To relate parity sampling to spectral expansion, we use the following
fact establishing that samplers of arbitrarily good parameters
$(\eta,\delta)$ can be obtained from sufficiently expanding bipartite
graphs.  This result is essentially a corollary of the expander mixing
lemma.

\begin{fact}[From Dinur et al.~\cite{DinurHKNT19}]\label{fact:exp_to_sampler}
A weighted bipartite graph with second singular value $\sigma_2$ is an $(\eta, \sigma_2^2/\eta^2)$-sampler.
\end{fact}

Using \cref{claim:complete_near_indep}, we show that the graph
between $X(1)$ and $X(k)$ obtained from a HDX is a parity sampler,
with parameters determined by its sampling properties.

\begin{claim}[Sampler bias amplification]\label{claim:sampler_bias}
Let $X(\leq d)$ be a HDX such that the weighted bipartite graph $M_{1,d}$ between $X(1) = [n]$ and $X(d)$ is an $(\eta, \delta)$-sampler.
For any $1 \leq k \leq d$, if $z \in \F_2^n$ has bias at most $\beta_0$, then
	$$\bias(\dsum_{X(k)}(z)) \leq (\beta_0 + \eta)^k + \frac{k^2}{d} + \delta.$$
\end{claim}

\begin{proof}
By downward closure, the subcomplex $X|_{\tee}$ obtained by
restricting to edges contained within some $\tee \in X(d)$ is a
complete complex on the ground set $\tee$.  Since $M_{1,d}$ is an
$(\eta, \delta)$-sampler, the bias of $z|_{\tee}$ must be within
$\eta$ of $\bias(z)$ on all but $\delta$ fraction of the edges $\tee$.
Hence
\begin{align*}
	\bias(\dsum_{X(k)}(z)) &= \abs{\E_{\{i_1,\dots,i_k\} \sim \Pi_k} (-1)^{z_{i_1} + \dots + z_{i_k}}} \\
	&= \abs{\E_{\tee \sim \Pi_d} \E_{\{i_1,\dots,i_k\} \in X\vert_{\tee}(k)} (-1)^{z_{i_1} + \dots + z_{i_k}}} \\
	&\leq \abs{\E_{\tee \sim \Pi_d} \E_{\{i_1,\dots,i_k\} \in X\vert_{\tee}(k)} (-1)^{z_{i_1} + \dots + z_{i_k}} \One_{[\bias(z|_{\tee}) \leq \beta_0 + \eta]}} \\
	& \qquad + \ProbOp_{t \sim \Pi_d}[\bias(z|_{\tee}) > \beta_0 + \eta]\\
	&\leq \E_{\tee \sim \Pi_d} \One_{[\bias(z|_{\tee}) \leq \beta_0 + \eta]} \abs{\E_{\{i_1,\dots,i_k\} \in X\vert_{\tee}(k)} (-1)^{z_{i_1} + \dots + z_{i_k}}} + \delta.
\end{align*}

By \cref{claim:complete_near_indep}, the magnitude of the expectation
of $(-1)^{z_i}$ over the edges of size $k$ in the complete complex
$X|_{\tee}$ is close to $\abs{\E_{i \sim X|_{\tee}(1)} (-1)^{z_i}}$,
which is just the bias of $z|_{\tee}$.  Then
\begin{align*}
	\bias(\dsum_{X(k)}(z)) &\leq \E_{\tee \sim X(d)} \One_{[\bias(z|_{\tee}) \leq \beta_0 + \eta]} \bias(z|_{\tee})^k + \frac{k^2}{d} + \delta \\    
	&\leq\left(\beta_0 + \eta \right)^k + \frac{k^2}{d} + \delta
\end{align*}
\end{proof}

Now we can compute the parameters necessary for a HDX to be an
$(\beta_0, \beta)$-parity sampler for arbitrarily small
$\beta$.

\begin{lemma}[HDXs are parity samplers]\label{lemma:hdx_parity_samplers}
Let $0 < \beta \leq \beta_0 < 1$, $0 < \theta < (1/\beta_0) - 1$, and $k \geq \log_{(1+\theta) \beta_0}(\beta/3)$.
If $X(\le d)$ is a $\gamma$-HDX with $d \geq \max\{3k^2/\beta, 6/(\theta^2 \beta_0^2 \beta)\}$ and $\gamma = O\left(1/d^2\right)$, then $X(k)$ is a $(\beta_0, \beta)$-parity sampler.
\end{lemma}

\begin{proof}
  Suppose the graph $M_{1,d}$ between $X(1)$ and $X(d)$ is an $(\eta, \delta)$-sampler.
We will choose $d$ and $\gamma$ so that $\eta = \theta\beta_0$ and $\delta = \beta/3$.
Using~\cref{fact:exp_to_sampler} to obtain a sampler with these parameters, we need the second singular value $\sigma_2$ of $M_{1,d}$ to be bounded as  
	$$\sigma_2 \leq \theta \beta_0 \sqrt{\frac{ \beta}{3}}.$$
By the upper bound on $\sigma_2^2$ from~\cref{theo:exp_up_walk_hdx}, it suffices to have
	$$\frac{1}{d} + O\left(d \gamma \right) \le \frac{\theta^2 \beta_0^2 \beta}{3},$$
which is satisfied by taking $d \ge 6/\left(\theta^2 \beta_0^2 \beta \right)$ and $\gamma = O\left(1/d^2\right)$.

By \cref{claim:sampler_bias}, $X(k)$ is a $(\beta_0, (\beta_0 + \eta)^k + k^2/d + \delta)$-parity sampler.
The first term in the bias is $(\beta_0 + \eta)^k = ((1+\theta) \beta_0)^k$, so we require $(1+\theta) \beta_0 < 1$ to amplify the bias by making $k$ large.
To make this term smaller than $\beta/3$, $k$ must be at least $\log_{(1+\theta) \beta_0}{\left(\beta/3\right)}$.
We already chose $\delta = \beta/3$, so ensuring $d \geq 3k^2/\beta$ gives us a $(\beta_0, \beta)$-parity sampler.
\end{proof}

\subsection{Rate of the Direct Sum Lifting}

By applying the direct sum lifting on a HDX to a base code $\mathcal
C_1$ with bias $\beta_0$, parity sampling allows us to obtain a
code $\mathcal C_k = \dsum_{X(k)}(\mathcal C_1)$ with arbitrarily
small bias $\beta$ at the cost of increasing the length of the
codewords.  The following lemma gives a lower bound on the rate of the
lifted code $\mathcal C_k$.

\begin{lemma}[Rate of direct sum lifting for a HDX]\label{lemma:direct_sum_rate} 
Let $\beta_0 \in (0,1)$ and $\theta \in (0, (1/\beta_0) - 1)$ be constants, and let $\mathcal{C}_1$ be an $\beta_0$-biased binary linear code with relative rate $r_1$.
For $\beta \in (0, \beta_0]$, suppose $k$, $d$, and $\gamma$ satisfy the hypotheses of \cref{lemma:hdx_parity_samplers}, with $k$ and $d$ taking the smallest values that satisfy the lemma.
The relative rate $r_k$ of the code $\mathcal C_k = \dsum_{X(k)}(\mathcal C_1)$ with bias $\beta$ constructed on a HDX with these parameters satisfies
 	$$r_k \geq r_1 \cdot \gamma^{O((\log(1/\beta))^4/(\beta^2 \gamma^2))}.$$
If $\gamma = C/d^2$ for some constant $C$, then this becomes
	$$r_k \geq r_1 \cdot \left(\frac{\beta^2}{(\log(1/\beta))^4}\right)^{O((\log(1/\beta))^{12}/\beta^6)}.$$
\end{lemma}

\begin{proof}
Performing the lifting from $\mathcal C_1$ to $\mathcal C_k$ does not change the dimension of the code, but it does increase the length of the codewords from $n$ to $|X(k)|$, where $|X(k)|$ is the size of the multiset of edges of size $k$ after each edge has been copied a number of times proportional to its weight.
Using the bound and flatness guarantee from \cref{lemma:hdx_existence}, we can compute
	$$r_k = \frac{r_1 n}{|X(k)|} \geq \frac{r_1}{D^2},$$
where $D \leq (1/\gamma)^{O(d^2/\gamma^2)}$.
Treating $\beta_0$ and $\theta$ as constants, the values of $k$ and $d$ necessary to satisfy \cref{lemma:hdx_parity_samplers} are
	$$k = \log_{(1+\theta) \beta_0}(\beta/3) = O(\log(1/\beta))$$
and
	$$d = \max \left\{\frac{3k^2}{\beta}, \frac{6}{\theta^2 \beta_0^2 \beta} \right\} = O\left(\frac{(\log(1/\beta))^2}{\beta} \right).$$
Putting this expression for $d$ into the inequality for $D$ yields
	$$D \leq (1/\gamma)^{O((\log(1/\beta))^4/(\beta^2 \gamma^2))},$$
from which the bounds in the lemma statement follow.
\end{proof}

From \cref{lemma:direct_sum_rate}, we see that if $\mathcal{C}_1$ has constant rate, then $\mathcal{C}_k$ has a rate which is constant with respect to $n$.  However, the dependence of the rate on the bias $\beta$ is quite poor.
This is especially striking in comparison to the rate achievable using Ta-Shma's expander walk construction described in \cref{subsec:expander_walks_parity_sampling}.

\begin{lemma}[Rate of direct sum lifting for expander walks~\cite{Ta-Shma17}]
Let $\beta_0 \in (0,1)$ be a constant and $\mathcal{C}_1$ be an $\beta_0$-biased binary linear code with relative rate $r_1$.
Fix $\beta \in (0,\beta_0]$.
Suppose $G$ is a graph with second largest singular value $\lambda = \beta_0/2$ and degree $d \leq 4/\lambda^2$.
Let $k = 2 \log_{2\beta_0}(\beta) + 1$ and $X(k)$ be the set of all walks of length $k$ on $G$.
Then the direct sum lifting $\mathcal C_k = \dsum_{X(k)}(\mathcal C_1)$ has bias $\beta$ and rate $r_k \geq r_1 \cdot \beta^{O(1)}$.
\end{lemma}

\begin{proof}
From \cref{theo:expander_parity_sampler} with this choice of $\lambda$ and $k$, the direct sum lifting $\mathcal{C}_k$ has bias $\beta$.
For the rate, observe that the lifting increases the length of the codewords from $n$ to the number of walks of length $k$ on $G$, which is $nd^k$.
Thus the rate of $\mathcal{C}_k$ is
	$$r_k = \frac{r_1 n}{nd^k} = \frac{r_1}{d^k}$$
As $d \leq 16/\beta_0$, which is a constant, and $k = O(\log(1/\beta))$, the rate satisfies $r_k \geq r_1 \cdot \beta^{O(1)}$.
\end{proof}

\section{Unique Decoding}\label{sec:unique_decoding}

In this section, we will show how parity sampling
and the ability to solve $k$-XOR instances with $X(k)$
as their constraint complex allow us to decode the direct sum lifting $\Cc_k
= \dsum_{X(k)}(\Cc_1)$ of a linear base code $\Cc_1 \in \FF_2^n$.
With a more technical argument, we can also handle different kinds of liftings
and non-linear codes, but for clarity of exposition we
restrict our attention to the preceding setting.

\subsection{Unique Decoding on Parity Samplers}

Our approach to unique decoding for $\Cc_k$ is as follows.
Suppose a received word $\tilde y \in \F_2^{X(k)}$ is close to $y^\star \in \Cc_k$, which is the direct sum lifting of some $z^\star \in \Cc_1$ on $X(k)$.
We first find an approximate solution $z \in \F_2^n$ to the $k$-XOR instance $\Ins(X(k), \tilde y)$ with predicates
	\[ \sum_{i  \in \ess} z_{i}  = {\tilde y}_\ess \pmod 2 \]
for every $\ess \in X(k)$.
Note that $z$ being an approximate solution to $\Ins(X(k), \tilde y)$ is equivalent to its lifting $\dsum_{X(k)}(z)$ being close to $\tilde y$.
In \cref{lem:robust-approx-ud}, we show that if $\dsum_{X(k)}$ is a sufficiently strong parity sampler, either $z$ or its complement $\conj{z}$ will be close to $z^\star$.
Running the unique decoding algorithm for $\Cc_1$ on $z$ and $\conj{z}$ will recover $z^\star$, from which we can obtain $y^\star$ by applying the direct sum lifting.

\begin{lemma}\label{lem:robust-approx-ud}
    Let $0 < \epsilon  < 1/2$ and $0 < \beta < 1/4-\epsilon/2$. Suppose $\Cc_1$ is a
    linear code that is efficiently uniquely
    decodable within radius $1/4 - \mu_0$ for some $\mu_0 >
    0$, and $\Cc_k = \dsum_{X(k)}(\Cc_1)$ where $\dsum_{X(k)}$ is a
    $(1/2+2\mu_0, 2\epsilon)$-parity sampler. Let $\tilde
    y \in \FF_2^{X(k)}$ be a word that has distance strictly less than $(1/4 - \epsilon/2 - \beta)$
    from $\Cc_k$, and let $y^\star = \dsum_{X(k)}(z^\star) \in \Cc_k$ be the word closest to $\tilde y$.
    
    Then, for any $z \in \FF_2^n$ satisfying
        $$\Delta(\dsum_{X(k)}(z), \tilde y) < \frac{1}{4} - \frac{\epsilon}{2},$$
    we have either
    \[ \Delta(z^\star,z) \le \frac{1}{4} - \mu_0 ~~\textrm{ or
    }~~\Delta(z^\star,\conj{z}) \le \frac{1}{4} - \mu_0.\]
    In particular, either $z$ or $\conj{z}$ can be efficiently
    decoded in $\Cc_1$ to obtain $z^\star \in \Cc_1$.
\end{lemma}

\begin{remark}
    Since $\dsum_{X(k)}$ is a $(1/2 + 2\mu_0, 2\epsilon)$-parity sampler, the code
    $\Cc_k$ has distance $\Delta(\Cc_k) \ge 1/2 - \epsilon$. This implies that $z^\star \in \Cc_1$
    is unique, since its direct sum lifting $y^\star$ is within distance $\Delta(\Cc_k)/2$ of $\tilde y$.
\end{remark}

\begin{proof}
  Let $y = \dsum_{X(k)}(z)$. We have
  $$
  \Delta(y^\star, y) \leq \Delta(y^\star, \tilde y) + \Delta(y, \tilde y) < \frac 12 - \epsilon.
  $$
  By linearity of $\dsum_{X(k)}$, $\Delta(\dsum_{X_k}(z^\star-z),0) < 1/2-\epsilon$, so $\bias(\dsum_k(z^\star-z)) > 2\epsilon$.
  From the $(1/2+2\mu_0, 2\epsilon)$-parity sampling assumption, $\bias(z^\star-z) > 1/2 + 2\mu_0$.
  Translating back to distance, either $\Delta(z^\star, z) < 1/4 - \mu_0$ or $\Delta(z^\star, z) > 3/4 + \mu_0$,
  the latter being equivalent to $\Delta(z^\star, \conj{z}) < 1/4-\mu_0$.
\end{proof}

To complete the unique decoding algorithm, we need only describe how a
good enough approximate solution $z \in \mathbb{F}_2^n$ to a $k$-XOR
instance $\Ins(X(k), \tilde y)$ allows us to recover $z^{\star} \in \Cc_1$
provided $\tilde y$ is sufficiently close to $\Cc_k$.

\begin{corollary} \label{cor:unique_decoding_sat}
  Suppose $\Cc_1$, $X(k)$, $z^{\star}$, $y^{\star}$ and $\tilde y$ are as in the
  assumptions of~\cref{lem:robust-approx-ud}. If $z \in \mathbb{F}_2^n$
  is such that
  $$
  \SAT_{\Ins(X(k), \tilde y)}(z) \ge \OPT_{\Ins(X(k), \tilde y)} - \beta,
  $$
  then unique decoding either $z$ or $\conj{z}$ gives $z^{\star} \in \Cc_1$.
  Furthermore, if such a $z$ can be found efficiently, so can $z^{\star}$.
\end{corollary}

\begin{proof}
  By the assumption on $z$, we have
  \begin{align*}
    1 - \Delta(\dsum_{X(k)}(z), \tilde y) &=\SAT_{\Ins(X(k), \tilde y)}(z)\\
                                         &\ge \OPT_{\Ins(X(k), \tilde y)} - \beta\\
                                         &\ge \SAT_{\Ins(X(k), \tilde y)}(z^{\star}) - \beta \\
                                         & = 1 - \Delta(y^\star, \tilde y) -\beta,
  \end{align*}
  implying $\Delta(\dsum_{X(k)}(z), \tilde y) \le \Delta(y^\star, \tilde y) + \beta$.
  Using the assumption that $\tilde y$ has distance strictly less than $(1/4 - \epsilon/2 - \beta)$
  from $\Cc_k$, we get that $\Delta(\dsum_{X(k)}(z), \tilde y) < 1/4 -\epsilon/2$, in which case
  we satisfy all of the conditions required for \cref{lem:robust-approx-ud}.
\end{proof}

\subsection{Concrete Instantiations}

\subsubsection*{High Dimensional Expanders}
If $X(k)$ is the collection of $k$-faces of a sufficiently expanding
$\gamma$-HDX, we can use the following algorithm to approximately solve the $k$-XOR instance $\Ins(X(k), \tilde y)$ and obtain $z \in \F_2^n$.

\begin{theorem}[\cite{AJT19}] \label{theo:hdx_csp}
Let $\Ins$ be an instance of \maxkcsp on $n$ variables taking values over an
alphabet of size $q$, and let $\beta > 0$. Let the simplicial complex $X_{\Ins}$
be a $\gamma$-HDX with $\gamma = \beta^{O(1)} \cdot (1/(kq))^{O(k)}$. 

There is an algorithm based on $(k/\beta)^{O(1)} \cdot q^{O(k)}$ levels of the
    Sum-of-Squares hierarchy which produces an assignment satisfying at least an $(\OPT_{\Ins} - \beta)$
fraction of the constraints in time $n^{(k/\beta)^{O(1)} \cdot q^{O(k)}}$.
\end{theorem}

If $X$ is a HDX with the parameters necessary to both satisfy this theorem and be a $(1/2 + 2 \mu_0, 2\epsilon)$ parity sampler, we can combine this with \cref{cor:unique_decoding_sat} to achieve efficient unique decodability of $\Cc_k = \dsum_{X(k)}(\Cc_1)$.

\begin{corollary}\label{cor:cooked-ud-hdx}
    Let $X(\le d)$ be a $d$-dimensional $\gamma$-HDX satisfying the premises of \cref{lemma:hdx_parity_samplers} that would guarantee that $X(k)$ is a 
    $(1/2 + 2\mu_0, 2\epsilon)$-parity sampler, and let $\Cc_1 \subseteq \F_2^n$ be a linear code which is efficiently unique decodable within radius $1/4 - \mu_0$ for some $\mu_0 > 0$.
    Then the code $\Cc_k = \dsum_{X(k)}(\Cc_1)$ can be unique decoded within distance $1/4 - \epsilon/2 - \beta$ in time $n^{(k/\beta)^{O(1)} \cdot 2^{O(k)}}, $\footnote{Here we are assuming that uniquely decoding $\Cc_1$ within
    radius $1/4 - \mu_0$ takes time less than this.} where we have
    \[ \beta = (\gamma \cdot (2k)^{O(k)})^{\frac{1}{O(1)}}. \]
\end{corollary}

\begin{proof}
    By \cref{lemma:hdx_parity_samplers}, we can achieve $(1/2+2\mu_0, 2\epsilon)$-parity sampling by taking $0 < \theta  < \frac{2}{1 + 4\mu_0} - 1$,
    $k \ge \log_{(1 + \theta) \cdot (\frac{1}{2} + 2\mu_0) }(2\epsilon/3)$, $d \ge \max\set*{ \frac{3k^2}{2\epsilon},\frac{3}{\theta^2 (1/2 + 2\mu_0)^2\epsilon}}$, and $\gamma= O(1/d^2)$.
    Let $\tilde y \in \F_2^{X(k)}$ be a received word with distance less than $(1/4 - \epsilon/2 - \beta)$ from $\Cc_k$.
    Applying \cref{theo:hdx_csp} to $\Ins(X(k), \tilde y)$ with $q=2$ and the given value of $\beta$, we obtain a $z \in \F_2^n$ with $\SAT_{\Ins(X(k), \tilde y)}(z) \ge \OPT_{\Ins(X(k), \tilde y)} - \beta$.
    This $z$ can be used in \cref{cor:unique_decoding_sat} to find $z^*$ and uniquely decode $\tilde y$ as $y^* = \dsum_{X(k)}(z^*)$.
\end{proof}

\subsubsection*{Expander Walks}

In \cref{sec:expander_walks}, we will show that the algorithmic results of \cite{AJT19} can be modified to work when $X(k)$ is a set of tuples of size $k$ which is sufficiently splittable (\cref{cor:alg-cooked-baked2}), which occurs when $X(k)$ is a set of walks on on a suitably strong expander (\cref{cor:cite-ahead-split}).
In particular, we have the following.

\begin{theorem} \label{theo:expander_walk_csp}
    Let $G = (V, E)$ be a graph with $\sigma_2(G) = \lambda$ and $k$ be
    a given parameter. Let $\Ins$ be a $k$-CSP instance over an
    alphabet of size $q$ whose constraint graph is the set of walks
    on $G$ of length $k$. Let $\beta > 0$ be such that $\lambda =
    O(\beta^2/(k^2 \cdot q^{2k}))$.

    There exists an algorithm based on $O\parens*{\frac{q^{4k}
    k^7}{\beta^5}}$ levels of the Sum-of-Squares hierarchy which
    produces an assignment satisfying at least an $(\OPT_{\Ins}
    - \beta)$ fraction of the constraints in time $n^{O(q^{4k} \cdot k^7 / \beta^5)}$.
\end{theorem}
Using this result, one can efficiently unique decode $\Cc_k = \dsum_{X(k)}(\Cc_1)$ when $X(k)$ is
the set of walks of length $k$ on an expander strong enough to achieve the necessary parity sampling property.

\begin{corollary}\label{cor:cooked-ud-expander}
    Let $X(k)$ be the set of walks on a graph $G$ with $\sigma_2(G) = \lambda$ such that $\dsum_{X(k)}$ is a $(1/2+2\mu_0, 2\epsilon)$ parity sampler, and let $\Cc_1 \subseteq \F_2^n$ be a linear code which is efficiently unique decodable within radius $1/4 - \mu_0$ for some $\mu_0 > 0$.
    Then the code $\Cc_k = \dsum_{X(k)}(\Cc_1)$ can be unique decoded within radius
    $1/4 - \epsilon/2 - \beta$ in time $n^{O(2^{4k} \cdot k^7/\beta^5)}$, where we have
    \[ \beta = O(\lambda \cdot k^2 \cdot 2^k).\]
\end{corollary}

\begin{proof}
    By \cref{theo:expander_parity_sampler}, we can obtain a $(1/2+2\mu_0, 2\epsilon)$-parity sampler by ensuring $1/2 + \mu_0 + 2\lambda < 1$ and $k \ge 2 \log_{1/2 + \mu_0 + 2\lambda}(2\epsilon) + 1$.
    Let $\tilde y \in \F_2^{X(k)}$ be a received word with distance less than $(1/4 - \epsilon/2 - \beta)$ from $\Cc_k$.
    Applying \cref{theo:expander_walk_csp} to $\Ins(X(k), \tilde y)$ with $q=2$ and the given value of $\beta$, we obtain a $z \in \F_2^n$ with $\SAT_{\Ins(X(k), \tilde y)}(z) \ge \OPT_{\Ins(X(k), \tilde y)} - \beta$.
    This $z$ can be used in \cref{cor:unique_decoding_sat} to find $z^*$ and uniquely decode $\tilde y$ as $y^* = \dsum_{X(k)}(z^*)$.
\end{proof}

\begin{remark}
    In both \cref{cor:cooked-ud-hdx}
    and \cref{cor:cooked-ud-expander}, when $\mu_0$ and $\epsilon$ are
    constants, $k$ can be constant, which means
    we can decode $\Cc_k$ from a radius arbitrarily close to $1/4 - \epsilon/2$ if
    we have strong enough guarantees on the quality of the expansion
    of the high-dimensional expander or the graph, respectively.

    Notice, however, that the unique decodability radius of the code
    $\Cc_k$ is potentially larger than
    $1/4 - \epsilon/2$. Our choice of $(1/2 + 2\mu_0,
    2\epsilon)$-parity sampling is needed to ensure that the approximate $k$-CSP
    solutions lie within the unique decoding radius of $\Cc_1$. Since
    the bias of the code $\Cc_1$ will generally be smaller than the parity sampling requirement of $1/2 + 2\mu_0$,
    the bias of the code $\Cc_k$ will be smaller than
    $2\epsilon$. In this case, the maximum distance at which our unique decoding algorithm works will be
    smaller than $\Delta(\Cc_k)/2$.
\end{remark}

\section{Abstract List Decoding Framework}\label{sec:list_dec}

In this section, we present the abstract list decoding framework with
its requirements and prove its guarantees. We introduce the entropic
proxy $\Psi$ in~\cref{sec:entropic_func} and use it to define the SOS
program for list decoding in~\cref{sec:sos_list_dec}.
In~\cref{sec:entropic_prop}, we establish key properties of $\Psi$
capturing its importance as a list decoding tool. We recall the
Propagation Rounding algorithm in~\cref{sec:propagation_rounding} and
formalize the notion of a slice as a set of assignments to variables
in the algorithm. Then, considerations of SOS
tractability of the lifting related to tensorial properties are dealt with
in~\cref{sec:tensorial}. Now, assuming we have a fractional SOS
solution to our program, the analysis of its covering properties and
the precise definition and correctness of the two later stages of the
framework are given in~\cref{sec:list_dec_blocks_and_analysis}. This
abstract framework will be instantiated using the direct sum lifting:
on HDXs in~\cref{sec:list_dec_xor_hdx} and on expander walks
in~\cref{sec:expander_walks}.

\subsection{Entropic Proxy}\label{sec:entropic_func}

In our list decoding framework via SOS, we will solve a single
optimization program whose resulting pseudo-expectation will in a
certain sense be rich enough to cover all intended solutions at once.
To enforce this covering property we rely on an analytical artifice,
namely, we minimize a convex function $\Psi$ that provides a proxy to
how concentrated the SOS solution is. More precisely, we use $\Psi$
from~\cref{def:neg_entropic_func}. A similar list decoding technique
was also (independently) used by Karmalkar et al.~\cite{KarmalkarKK19}
and Raghavendra--Yau~\cite{RaghavendraY19}, but in the context of
learning.

\begin{definition}[Entropic Proxy]\label{def:neg_entropic_func}
  Let $\rv Y = \{\rv Y_{\ess}\}_{\ess \in X(k)}$ be a $t$-local PSD
  ensemble with $t \ge 2$. We define $\Psi = \Psi\left(\{\rv
  Y_{\ess}\}_{\ess \in X(k)}\right)$ as
  $$
  \Psi ~\coloneqq~ \E_{\ess,\tee \sim \Pi_k} \left(\widetilde{\E}\left[\rv Y_{\ess} \rv Y_{\tee} \right] \right)^2.
  $$
  We also denote $\Psi$ equivalently as $\Psi
  = \Psi\left(\widetilde{\E}\right)$ where $\widetilde{\E}$ is the
  pseudo-expectation operator associated to the ensemble $\rv Y$.
\end{definition}

\subsection{SOS Program for List Decoding}\label{sec:sos_list_dec}

Let $\tilde{y} \in \{\pm 1\}^{X(k)}$ be a word promised to be
\lict{\sqrt{\epsilon}} to a lifted code $\mathcal{C}_k
= \lift(\mathcal{C}_1)$. The word $\tilde{y}$ is to be regarded as a
(possibly) corrupted codeword for which we want to do list
decoding. We consider the following SOS program.

\begin{table}[H]
\hrule
\vline
\begin{minipage}[t]{0.99\linewidth}
\vspace{-5 pt}
{\small
\begin{align*}
    &\mbox{minimize}\quad ~~ \Psi\left(\{\rv Y_{\ess}\}_{\ess \in X(k)}\right)\tag{List Decoding Program}\label{sos:list_dec}\\
&\mbox{subject to}\quad \quad ~\\
    &\qquad \E_{\ess \sim \Pi_k} \widetilde{\E} \left[ \tilde{y}_{\ess} \cdot \rv Y_{\ess}\right] \ge 2 \sqrt{\epsilon}\label{cons:agreement-ld}\tag{Agreement Constraint}\\
&\qquad \rv Z_1,\dots, \rv Z_n \textup{ being $(L+2k)$-local PSD ensemble}
\end{align*}}
\vspace{-14 pt}
\end{minipage}
\hfill\vline
\hrule
\caption{List decoding SOS formulation for $\tilde{y}$.}
\end{table}

\subsection{Properties of the Entropic Proxy}\label{sec:entropic_prop}

We establish some key properties of our negative entropic function
$\Psi$. First, we show that $\Psi$ is a convex function. Since the
feasible set defined by the SOS~\ref{sos:list_dec} is convex and
admits an efficient separation oracle~\footnote{In our setting the
pseudo-expectation has trace bounded by $n^{O(t)}$ in which case
semidefinite programming can be solved
efficiently~\cite{GM12,RW17:sos}.}, the convexity of $\Psi$ implies
that the~\ref{sos:list_dec} can be efficiently solved within
$\eta$-optimality in time $n^{O(t)} \cdot
\polylog(\eta^{-1})$ where $t$ is the SOS degree.
\begin{lemma}[Convexity]
  $\Psi$ is convex, i.e., for every pair of pseudo-expectations
    $\widetilde{\E}_1$ and $\widetilde{\E}_2$ and $\alpha \in [0,1]$,
  $$
  \Psi\left(\alpha \cdot \widetilde{\E}_1 + (1-\alpha)\cdot \widetilde{\E}_2 \right) ~\le~ \alpha \cdot \Psi\left(\widetilde{\E}_1\right) + (1-\alpha) \cdot \Psi\left(\widetilde{\E}_2\right).
  $$
\end{lemma}

\begin{proof}
  Suppose $\ess \cup \tee = \set{i_1,\dots,i_t}$. By definition $\rv
  Y_{\ess} \rv Y_{\tee} = \lift(\rv
  Z)_\ess \cdot \lift(\rv Z)_\tee$, i.e., $\rv Y_{\ess} \rv
  Y_{\tee}$ is a function $f$ on input $\rv Z_{i_1},\dots, \rv
  Z_{i_t} \in \set{\pm 1}$. Let
  $$
  f(\rv Z_{i_1},\dots, \rv Z_{i_t}) = \sum_{S \subseteq \ess \cup \tee} \widehat{f}(S) \cdot \prod_{i \in S} \rv Z_i,
  $$
  be the Fourier decomposition of $f$. Then
  $$
  \widetilde{\E}\left[\rv Y_{\ess} \rv Y_{\tee}\right] = \widetilde{\E}\left[f\right] = \sum_{S \subseteq \ess \cup \tee} \widehat{f}(S) \cdot \widetilde{\E}\left[\prod_{i \in S} \rv Z_i \right].
  $$
  Since $\widetilde{\E}\left[\rv Y_{\ess} \rv Y_{\tee}\right]$ is a
  linear function of $\widetilde{\E}$, we obtain
  $\left(\widetilde{\E}\left[\rv Y_{\ess} \rv
    Y_{\tee}\right]\right)^2$ is convex. Now, the convexity of $\Psi$    
    follows by noting
  that $\Psi$ is a convex combination of convex functions.
\end{proof}

The (sole) problem-specific constraint appearing in the SOS
\ref{sos:list_dec} allows us to deduce a lower bound on
$\Psi$. This lower bound will be important later to show that a
feasible solution that does not cover all our intended solutions must
have $\Psi$ bounded away from $0$ so that we still have room to
decrease $\Psi$. We note that an improvement in the conclusion of the
following lemma would directly translate to stronger list decoding
parameters in our framework.

\begin{lemma}[Correlation $\implies$ entropic bound]\label{lemma:cor_imp_entropic_bound}
  Let $\{\rv Y_{\ess}\}_{\ess \in X(k)}$ be $t$-local PSD ensemble
  with $t \ge 2$. If there is some $y \in \{\pm 1\}^{X(k)}$ such that
  $$
  \left\vert \E_{\ess \sim \Pi_k}\widetilde{\E} \left[ y_{\ess} \cdot \rv Y_{\ess}\right]\right\vert \ge \beta,
  $$
  then
  $$
  \Psi\left(\{\rv Y_{\ess}\}_{\ess \in X(k)}\right) ~\ge~ \beta^4.
  $$
\end{lemma}

\begin{proof}
    We calculate
 \begin{align*}
     \E_{\ess,\tee \sim \Pi_k} \left(\widetilde{\E}\left[\rv Y_{\ess} \rv Y_{\tee}\right] \right)^2 &~=~  \E_{\ess,\tee \sim \Pi_k} \left(\widetilde{\E}\left[\left(y_{\ess} \rv Y_{\ess}\right)\left(y_{\tee} \rv Y_{\tee}\right) \right] \right)^2\\
    &~\ge~  \left(\E_{\ess,\tee\sim \Pi_k} \widetilde{\E}\left[\left(y_{\ess} \rv Y_{\ess}\right)\left(y_{\tee} \rv Y_{\tee}\right)\right]\right)^2 && \text{(Jensen's Inequality)}\\    
    &~=~  \left( \widetilde{\E}\left[ (\E_{\ess \sim \Pi_k} y_{\ess} \cdot \rv Y_{\ess})^2\right]\right)^2\\
    &~\ge~  \left( \widetilde{\E} \left[ \E_{\ess \sim \Pi_k}\left[ y_{\ess} \cdot \rv Y_{\ess}\right] \right] \right)^4 && \text{(Cauchy--Schwarz Inequality)}\\
    &~=~  \left( \E_{\ess \sim \Pi_k}  \widetilde{\E}\left[ y_{\ess} \cdot \rv Y_{\ess}\right]\right)^4 ~\ge~ \beta^4.
 \end{align*}
\end{proof}

We now show the role of $\Psi$ in list decoding: if an intended
solution is not represented in the pseudo-expectation $\PExp$, we can
get a new pseudo-expectation $\PExp'$ which attains a smaller value
of $\Psi$.
\begin{lemma}[Progress lemma]\label{lemma:progress_lemma}
  Suppose there exist $z \in \{\pm 1\}^{X(1)}$ and $y=\lift(z) \in \{\pm 1 \}^{X(k)}$ satisfying
  $$
  \PExp \left[\left(\E_{\ess \sim \Pi_k} y_{\ess} \cdot \rv Y_{\ess} \right)^2 \right] \leq \delta^2.
  $$
  If $\Psi \ge \delta^2$, then there exists a pseudo-expectation $\widetilde{\E}'$ such that
  $$
  \E_{\ess,\tee\sim \Pi_k} \left(\widetilde{\E}' \left[\rv Y_{\ess} \rv Y_{\tee} \right] \right)^2 ~\le~ \Psi - \frac{\left(\Psi - \delta^2\right)^2}{2}.
  $$
  In particular, if $\Psi \ge 2 \delta^2$, then
  $$
  \E_{\ess,\tee\sim \Pi_k} \left(\widetilde{\E}' \left[\rv Y_{\ess} \rv Y_{\tee} \right] \right)^2 ~\le~ \Psi - \frac{\delta^4}{2}.
  $$  
\end{lemma}

\begin{proof}
  Let $\widetilde{\E}'$ be the pseudo-expectation~\footnote{By summing
    the pseudo-expectation $\widetilde{\E}$ and actual expectation
    $\E_{\delta_z}$, we mean that we are summing $\widetilde{\E}$ to
    pseudo-expectation of the same dimensions obtained from operator
    $\E_{\delta_z}$.}
  $$
  \widetilde{\E}' ~\coloneqq~ (1-\alpha) \cdot \widetilde{\E} + \alpha \cdot \E_{\delta_z},
  $$
  where $\E_{\delta_z}$ is the expectation of the delta distribution on $z$ and
  $\alpha \in (0,1)$ is to be defined later. We have
  \begin{align*}
    \E_{\ess,\tee\sim \Pi_k} \left(\widetilde{\E}' \left[\rv Y_{\ess} \rv Y_{\tee} \right] \right)^2 & ~=~ \E_{\ess,\tee\sim \Pi_k} \left((1-\alpha) \cdot \widetilde{\E} \left[\rv Y_{\ess} \rv Y_{\tee} \right] + \alpha \cdot y_{\ess} y_{\tee} \right)^2\\
     & ~=~ (1-\alpha)^2 \cdot \Psi + \alpha^2 \cdot \E_{\ess,\tee\sim \Pi_k}(y_{\ess} y_{\tee})^2 + 2\alpha(1-\alpha)\cdot  \E_{\ess, \tee \sim \Pi_k}\left[\PExp[\rv Y_{\ess} \rv Y_{\tee}] y_{\ess} y_{\tee}\right] \\
     & ~\le~ (1-\alpha)^2 \cdot \Psi + \alpha^2  + 2\alpha(1-\alpha) \cdot \delta^2.
  \end{align*}
  The value of $\alpha$ minimizing the quadratic expression of the RHS above is
  $$
  \alpha^{\star} ~=~ \frac{\Psi - \delta^2}{1 + \Psi - 2\delta^2}.
  $$
  Using this value yields
  \begin{align*}
    \E_{\ess,\tee\sim \Pi_k} \left(\widetilde{\E}' \left[\rv Y_{\ess} \rv Y_{\tee} \right] \right)^2 &~\le~ \Psi - \frac{\left(\Psi - \delta^2\right)^2}{1+\Psi - 2\delta^2} \\
    &~\le~ \Psi - \frac{\left(\Psi - \delta^2\right)^2}{2},
  \end{align*}
  where in the last inequality we used $\Psi \le 1$.
\end{proof}

\subsection{Propagation Rounding}\label{sec:propagation_rounding}

A central algorithm in our list decoding framework is the Propagation
Rounding \cref{algo:prop-rd}. It was studied by Barak et
al.~\cite{BarakRS11} in the context of approximating $2$-CSPs on low
threshold rank graphs and it was later generalized to HDXs (and low
threshold rank hypergraphs) in the context of $k$-CSPs \cite{AJT19}.

Given an $(L+2k)$-local PSD ensemble $\set{\rv Z_1, \ldots, \rv Z_n}$,
the Propagation Rounding~\cref{algo:prop-rd} chooses a subset of
variables $S \subseteq [n]$ at random. Then it samples a joint
assignment $\sigma$ to the variables in $S$ according to $\set{\rv
Z_S}$. The value of the remaining variables $\rv Z_i$ are sampled
according to the conditional marginal distributions $\set{\rv
Z_i \vert \rv Z_S = \sigma}$. An important byproduct of this algorithm
is the $2k$-local PSD ensemble $\rv Z' = \set{\rv
Z_1,\dots,\rv Z_n \vert \rv Z_S = \sigma}$.

The precise description of the Propagation Rounding~\cref{algo:prop-rd} follows.

\begin{algorithm}{Propagation Rounding Algorithm}{An $(L+2k)$-local PSD ensemble $\set{\rv Z_1, \ldots, \rv Z_n}$ and some
    distribution $\Pi_k$ on $X(k)$.}{A random assignment $(\assn_1, \ldots, \assn_n) \in  [q]^n$ and $2k$-local PSD ensemble $\rv Z'$.}\label{algo:prop-rd}
    \begin{enumerate}
        \item Choose $m \in \set*{1, \ldots, L/k}$ uniformly at random.
        \item Independently sample $m$ $k$-faces, $\ess_j \sim \Pi_k$ for $j = 1, \ldots, m$.
        \item Write $S = \bigcup_{j = 1}^m \ess_j$, for the set of the seed vertices.
        \item Sample assignment $\assn: S \to [q]$ according to the local distribution $\set{\rv Z_{S}}$.
        \item Set $\rv Z' = \set{\rv Z_1, \ldots, \rv Z_n | \rv Z_S = \assn}$, i.e.~the local ensemble
            $\rv Z$ conditioned on agreeing with $\assn$.
        \item For all $j \in [n]$, sample independently $\assn_j \sim \set{\rv Z'_j}$.
        \item Output $(\assn_1, \ldots, \assn_n)$ and $\rv Z'$.
    \end{enumerate}
\end{algorithm}

To our list decoding task we will show that an ensemble minimizing
$\Psi$ covers the space of possible solutions in the sense that for
any intended solution there will be a choice of $S$ and $\sigma$ such
that the conditioned ensemble $\rv Z'$ enables the sampling of a word
within the unique decoding radius in $\mathcal{C}_1$ of this intended
solution.

An execution of the \cref{algo:prop-rd} is completely determined by
the tuple $(m,S,\sigma)$ which we will refer to as a slice of the PSD
ensemble.

\begin{definition}[Slice]
  We call a tuple $(m,S,\sigma)$ obtainable by~\cref{algo:prop-rd}
  a \textit{slice} and let $\Omega$ denote the set of all slices obtainable
  by~\cref{algo:prop-rd}.
\end{definition}

We can endow $\Omega$ with a natural probability distribution, where
the measure of each $(m, S, \sigma)$ is defined as the probability
that this slice is picked during an execution of \cref{algo:prop-rd}.
We also define a pseudo-expectation operator for each slice.

\begin{definition}[Pseudo-Expectation Slice]
  Given a slice $(m,S,\sigma)$, we define the pseudo-expectation
  operator $\widetilde{\E}_{\vert_{S,\sigma}}$ which is the
  pseudo-expectation operator of the conditioned local PSD ensemble
  $\set{\rv Z_1,\dots, \rv Z_n\vert \rv Z_s = \sigma}$.
\end{definition}

\subsection{Tensorial Structures}\label{sec:tensorial}

In general, a local PSD ensemble $\rv Z' = \set{\rv Z_1', \dots, \rv
Z_n'}$ output by the Propagation Rounding \cref{algo:prop-rd} may be
far from corresponding to any underlying joint global
distribution~\footnote{In fact, if this was the case, then we would be
able to approximate any $k$-CSP with SOS degree $(L+2k)$. However,
even for $L$ as large as linear in $n$ this is impossible for
SOS~\cite{Grigoriev01,KothariMOW17}.}. In our application, we will be
interested in the case where the ensemble approximately behaves as
being composed of independent random variables over the collection of
``local views'' given by the hyperedges in $X(k)$. In such case,
rounding the SOS solution via independent rounding is
straightforward. A collection of local views admitting this property
with a given SOS degree parameter $L$ is denoted \textit{tensorial}
(variables behave as products over the local views).

\begin{definition}[Tensorial Hypergraphs]
  Let $X(k)$ be a collection of $k$-uniform hyperedges endowed with a
  distribution $\Pi_k$, $\mu \in [0,1]$, and $L \in \mathbb{N}$.  We
  say that $X(k)$ is $(\mu,L)$-tensorial if the local PSD ensemble
  $\rv Z'$ returned by Propagation Rounding \cref{algo:prop-rd} with
  SOS degree parameter $L$
  satisfies
  \begin{equation}
    \ExpOp_{\Omega} \ExpOp_{\aye \sim \Pi_k}{ \norm{\set{\rv Z_{\aye}'} - \set*{\rv Z_{a_1}'}\cdots \set*{\rv Z_{a_k}'}}_1} \le \mu.
  \end{equation}
\end{definition}

To analyze the potential $\Psi$ we will need that the variables
between pairs of local views, i.e., pairs of hyperedges, behave as
product.
\begin{definition}[Two-Step Tensorial Hypergraphs]\label{def:two_step_tensorial}
  Let $X(k)$ be a collection of $k$-uniform hyperedges endowed with a
  distribution $\Pi_k$, $\mu \in [0,1]$, and $L \in \mathbb{N}$.  We
  say that $X(k)$ is $(\mu,L)$-two-step tensorial if it is $(\mu, L)$-tensorial and the PSD ensemble
  $\rv Z'$ returned by Propagation Rounding~\cref{algo:prop-rd} with
  SOS degree parameter $L$ satisfies
    $$\ExpOp_{\Omega} \ExpOp_{\ess,\tee \sim \Pi_k}{ \norm{\set{\rv Z_{\ess}' \rv Z_{\tee}'} - \set*{\rv Z_{\ess}'}\set*{\rv Z_{\tee}'}}_1 } \le \mu.$$         
\end{definition}

In~\cref{subsec:hdx_two_step_tensorial}, we establish the
relationship between the parameters $\mu$ and $L$
and the expansion that will ensure HDXs are
$(\mu,L)$-two-step tensorial. Similarly,
in~\cref{subsec:expander_walk_two_step_tensorial} we provide this
relationship when $X(k)$ is the collection of walks of an expander
graph.

\subsubsection*{Tensorial over Most Slices}

By choosing $\mu$ sufficiently small it is easy to show that most
slices $(m,S,\sigma)$ satisfy the tensorial (or two-step tensorial)
statistical distance condition(s) with a slightly worse parameter
$\tilde{\mu}$ such that $\tilde{\mu} \to 0$ as $\mu \to 0$. If we
could construct tensorial (or two-step tensorial) objects for
arbitrarily small parameter $\mu$ with $L = O_{k,q,\mu}(1)$, then we
would be able to obtain $\tilde{\mu}$ arbitrarily
small.~\cref{lemma:double_close_to_product} establishes that HDXs of
appropriate expansion satisfy this assumption, and
\cref{lemma:exp_walk_double_close_to_product} does the same for
walks on expanders.

We introduce two events. The first event captures when a slice
$(m,S,\sigma)$ leads to the conditioned local variables $\rv
Z_1',\dots, \rv Z_n'$ being close to $k$-wise independent over the
$k$-sized hyperedges.

\begin{definition}[Ground Set Close to $k$-wise Independent]
   Let $\mu \in (0,1]$. We define the event $K_{\mu}$ as
   $$
   K_{\mu} \coloneqq \left\{(m,S,\sigma) \in \Omega~\vert~ \ExpOp_{\aye \sim \Pi_k}{ \norm{\set{\rv Z_{\aye}\vert \rv Z_S = \sigma} - \set*{\rv Z_{a_1}\vert \rv Z_S = \sigma}\cdots \set*{\rv Z_{a_k}\vert \rv Z_S = \sigma}}_1 } < \mu^2/2 \right\}.
   $$
\end{definition}

The second event captures when the variables between pairs of
hyperedges are close to independent.

\begin{definition}[Lifted Variables Close to Pairwise Independent]
   Let $\mu \in (0,1]$. We define the event $P_{\mu}$ as
   $$
   P_{\mu} \coloneqq \left\{(m,S,\sigma) \in \Omega~\vert~ \ExpOp_{\ess,\tee \sim \Pi_k}{ \norm{\set{\rv Z_{\ess} \rv Z_{\tee} \vert \rv Z_S = \sigma} - \set*{\rv Z_{\ess}\vert \rv Z_S = \sigma}\set*{\rv Z_{\tee} \vert \rv Z_S = \sigma}}_1 } < \mu^2/2 \right\}.
   $$
\end{definition}

These events satisfy a simple concentration property.
\begin{claim}[Concentration]\label{claim:glorified_concentration}
    Suppose a simplicial complex $X(\le k)$ with $X(1) = [n]$ and an
    $(L+2k)$-local PSD ensemble $\rv Z = \set{\rv Z_1, \ldots, \rv Z_n}$
    are given as input to Propagation Rounding~\cref{algo:prop-rd}.
    Let $\mu \in (0,1]$. If $X(k)$ is $(\mu^4/4,L)$-two-step
    tensorial, then
    \begin{equation}\label{eq:first_glorified_concentration}
        \ProbOp_{(m,S,\sigma)\sim\Omega}\left[ K_{\mu}^c \right] \le \frac{\mu^2}{2},
    \end{equation}
    and
    \begin{equation}\label{eq:second_glorified_concentration}
        \ProbOp_{(m,S,\sigma)\sim\Omega}\left[ P_{\mu}^c \right] \le \frac{\mu^2}{2}.
    \end{equation}
\end{claim}

\begin{proof}
   We only prove~\cref{eq:first_glorified_concentration} since the
   proof of~\cref{eq:second_glorified_concentration} is similar.
   Define the random variable $\rv
   R \coloneqq \ExpOp_{\aye \sim \Pi_k}{\norm{\set{\rv Z_{\aye}'}
   - \set*{\rv Z_{a_1}'}\cdots \set*{\rv Z_{a_k}'}}_1}$ on the sample
   space $\Omega = \{(m,S,\sigma)\}$. From our $(\mu^4/4,L)$-two-step
   tensorial assumption we have
   $$
   \E_{\Omega}\left[ \rv R \right] \le \frac{\mu^4}{4}.
   $$
   Now, we can conclude
   $$
   \ProbOp_{(m,S,\sigma)\sim\Omega}\left[ K_{\mu}^c \right] = \ProbOp_{(m,S,\sigma)\sim\Omega} \left[\rv R \ge  \frac{\mu^2}{2} \right] \le \frac{\mu^2}{2},
   $$
   using Markov's inequality.
\end{proof}

\subsection{Further Building Blocks and Analysis}\label{sec:list_dec_blocks_and_analysis}

Before we delve into further phases of the list decoding framework, we
introduce some notation for the list of codewords we want to retrieve.

\begin{definition}[Code list]
  Given $\tilde{y} \in \set{\pm 1}^{X(k)}$ and a code $\mathcal{C}$ on
  $X(k)$ with relative distance at least $1/2-\epsilon$, we define the
  list $\mathcal{L}(\tilde{y},\mathcal{C})$ as
  $$
  \mathcal{L}(\tilde{y},\mathcal{C})~\coloneqq~\left\{y \in \mathcal{C} ~\vert~ \Delta\parens{y, \tilde{y}} \le \frac{1}{2} - \sqrt{\epsilon}\right\}.
  $$
\end{definition}

Under these assumptions the Johnson bound establishes that the list
size is constant whenever $\epsilon > 0$ is constant. 

\begin{remark}\label{rem:Johnson_list_size}
  The Johnson bound~\cite{GRS:coding:notes} guarantees that
  $$
  \left\vert \mathcal{L}(\tilde{y},\mathcal{C})\right\vert \le \frac{1}{2 \cdot \epsilon}
  $$
  provided the relative distance of $\mathcal{C}$ is at least $1/2-\epsilon$.
\end{remark}

In the case of lifted codes, it is more appropriate to consider a list
of pairs $\mathcal{L}(\tilde{y},\mathcal{C}_1,\mathcal{C}_k)$ defined as follows.
\begin{definition}[Coupled code list]\label{def:coupled_code_list}
  Given $\tilde{y} \in \set{\pm 1}^{X(k)}$ and a lifted code
  $\mathcal{C}_k$ on $X(k)$ with relative distance at least
  $1/2-\epsilon$, we define the coupled code list
  $\mathcal{L}(\tilde{y},\mathcal{C}_1,\mathcal{C}_k)$ as
  $$
    \mathcal{L}(\tilde{y},\mathcal{C}_1,\mathcal{C}_k)~\coloneqq~\left\{(z,\lift(z)) ~\vert~ z \in \mathcal{C}_1 ~~\textrm{ and }~~\Delta\parens*{\lift(z),\tilde{y}} \le \frac{1}{2} - \sqrt{\epsilon}\right\}.
  $$
\end{definition}
\noindent Recovering this list
$\mathcal{L}(\tilde{y},\mathcal{C}_1,\mathcal{C}_k)$ is the main goal
of this section. This task will be accomplished
by~\cref{algo:list_decoding} stated below whose building blocks and
analysis we develop in this section.

\begin{algorithm}{List Decoding Algorithm}{A word $\tilde{y} \in \set{\pm 1}^{X(k)}$ which is \lict{\sqrt{\epsilon}} to $\mathcal{C}_k=\lift(\mathcal{C}_1)$.}
                                          {Coupled code list $\mathcal{L}(\tilde{y},\mathcal{C}_1,\mathcal{C}_k)$.}\label{algo:list_decoding}
    \begin{enumerate}                   
        \item Solve the~\ref{sos:list_dec} with $\eta$-accuracy, obtaining $\rv Z$, where $\eta = \epsilon^8/2^{22}$ 
        \item Let $\mathcal{M}$ be the output of the Cover Retrieval~\cref{algo:cover_retrieval} on $\rv Z$
        \item Let $\mathcal{L'}$ be the output of the Cover Purification~\cref{algo:purification} on $\mathcal{M}$
        \item Let $\mathcal{L}'' = \{(z,y) \in \mathcal{L}'~\vert~ \Delta\parens{\tilde{y},y} \le 1/2 - \sqrt{\epsilon} \}$          
        \item Output $\mathcal{L}''$
    \end{enumerate}
\end{algorithm}

As shown in~\cref{fig:list_dec_framework} of~\cref{sec:strategy}, the
first step is to solve the~\ref{sos:list_dec} which results in
a pseudo-expectation ``covering'' the list
$\mathcal{L}(\tilde{y},\mathcal{C})$ as we will make precise. A
precursor property to covering and some considerations about SOS
rounding are treated in~\cref{subsec:sos_round_and_rec}. Next, the
formal definition of cover is presented
in~\cref{subsec:coupled_pairs_list_and_cover} and we have all the
elements to present the Cover Retrieval~\cref{algo:cover_retrieval}
with its correctness in~\cref{subsec:cover_retrieval}. Then, we use
the \textit{robustness} properties of the lifting to purify the cover
in~\cref{subsec:cover_purification}.
Finally, in~\cref{subsec:list_dec_correctness}, we assemble the
building blocks and prove the main technical
result, \cref{theo:list_dec}, whose proof follows easily once the
properties of the building blocks are in place.

Note
that~\cref{theo:list_dec} embodies an abstract list decoding framework
which relies only on the $\textit{robustness}$ and $\textit{tensorial}$
properties of the lifting. We provide a concrete instantiation of the
framework to the direct sum lifting on HDXs in~\cref{sec:list_dec_xor_hdx}
and to the direct sum lifting on
expander walks in~\cref{sec:list_dec_xor_expander_walks}.

\begin{theorem}[List Decoding Theorem]\label{theo:list_dec}
  Suppose that $\lift$ is a $(1/2-\epsilon_0,1/2-\epsilon)$-robust
  $(\epsilon^8/2^{22},L)$-two-step tensorial lifting from
  $\mathcal{C}_1$ to $\mathcal{C}_k$ which is either
  \begin{itemize}
    \item linear and a $(1/2+\epsilon_0, 2\cdot \epsilon)$-parity sampler; or
    \item $(1/4 - \epsilon_0, 1/2 - \epsilon/2)$-robust and odd.
  \end{itemize}
    Let $\tilde{y} \in \set{\pm 1}^{X(k)}$ be \lict{\sqrt{\epsilon}} to
  $\mathcal{C}_k$. Then w.v.h.p.~the List
  Decoding~\cref{algo:list_decoding} returns the coupled code list
  $\mathcal{L}(\widetilde{y},\mathcal{C}_1,\mathcal{C}_k)$. Furthermore, the
  running time is
  $$
  n^{O(L+k)}\left(\polylog(\epsilon^{-1}) + f(n)\right),
  $$
  where $n=\vert X(1) \vert$ and $f(n)$ is the running time of a
  unique decoding algorithm of $\mathcal{C}_1$.
\end{theorem}

\begin{remark}
  Regarding~\cref{theo:list_dec}, we stress that although the lifting
  is $(1/2-\epsilon_0,1/2-\epsilon)$-robust and we can perform list decoding
  at least up to distance $1/2-\sqrt{\epsilon}$, our framework does not
  recover the Johnson bound. The issue is that our framework requires
  one of the additional amplification guarantees of~\cref{theo:list_dec},
  which both make the distance of $\mathcal{C}_k$ become
  $1/2-\epsilon^{\Omega_{\epsilon_0}(1)} > 1/2 - \epsilon$.
  Efficiently recovering the Johnson
  bound remains an interesting open problem.
\end{remark}

We observe that the algorithms themselves used in this framework are
quite simple (although their analyses might not be). Moreover, the tasks
of cover retrieval and purification are reasonably
straightforward. However,~\cref{subsec:sos_round_and_rec} combines
$\textit{tensorial}$ properties of the lifting with properties of
$\Psi$, requiring a substantial analysis. The list decoding
framework is divided into stages to make it modular so that key
properties are isolated and their associated functionality can be
presented in a simple manner. Most of the power of this framework
comes from the combination of these blocks and the concrete expanding
objects capable of instantiating it.

\subsubsection{SOS Rounding and Recoverability}\label{subsec:sos_round_and_rec}

We show that if a slice $(m,S,\sigma)$ ``captures'' an intended
solution $y \in \set{\pm 1}^{X(k)}$ (this notion is made precise in
the assumptions of~\cref{lemma:fractional_to_integral}), then we can
retrieve a $z \in \set{\pm 1}^{X(1)}$ such that $\lift(z)$ has some
agreement with $y$. This agreement is somewhat weak, but combined with
the robustness of the lifting, it will be enough for our purposes. In this
subsection, we first explore how to recover such words within a slice,
which can be seen as local
rounding in the slice. Next, we establish sufficient conditions for an
intended solution to be recoverable, now not restricted to a given
slice but rather with respect to the full pseudo-expectation. Finally,
we use all the tools developed so far to show that by minimizing
$\Psi$ in a two-step tensorial structure we end up with a
pseudo-expectation in which all intended solutions are
recoverable. The interplay between weak agreement and robustness of
the lifting is addressed in~\cref{subsec:cover_purification}.

We will be working with two-step tensorial structures where the
following product distribution associated to a slice naturally
appears.
\begin{definition}[Product Distribution on a Slice]\label{def:prod_dist_on_slice}
 We define $\set{\rv Z^{\otimes}\vert_{(S,\sigma)}}$ to be the product distribution on the marginals
 $\set{\rv Z_i\vert \rv Z_S = \sigma}_{i \in X(1)}$, i.e.,
 $\set{\rv Z^{\otimes}\vert_{(S,\sigma)}} \coloneqq \prod_{i\in X(1)} \set{\rv Z_i\vert \rv Z_S = \sigma}$.
\end{definition}

Under appropriate conditions,~\cref{lemma:fractional_to_integral}
shows how to round the pseudo-expectation in a slice.
\begin{lemma}[From fractional to integral in a slice]\label{lemma:fractional_to_integral}
    Let $(m,S,\sigma) \in \Omega$ be a slice. Suppose
    \begin{equation}\label{eq:brs_hdx}
        \ExpOp_{\aye \sim \Pi_k}{ \norm{\set{\rv Z_{\aye}\vert \rv Z_S = \sigma} - \set*{\rv Z_{a_1}\vert \rv Z_S = \sigma }\cdots \set*{\rv Z_{a_k}\vert \rv Z_S = \sigma}}_1 } \le \mu,
    \end{equation}
    and
    \begin{equation}\label{eq:hdx_extra_indep}
        \ExpOp_{\ess,\tee \sim \Pi^2_k}{ \norm{\set{\rv Z_{\ess}\rv Z_{\tee}\vert \rv Z_S = \sigma} - \set*{\rv Z_{\ess}\vert \rv Z_S = \sigma}\set*{\rv Z_{\tee}\vert \rv Z_S = \sigma}}_1 } \le \mu.
    \end{equation}    
    For $\beta \in (0,1)$,
    if $\mu \le \beta \cdot \kappa^2/6$ and $y \in \set{\pm 1}^{X(k)}$ is such that
    $$
    \E_{\ess,\tee \sim \Pi^2_k} \widetilde{\E}_{\vert S,\sigma} \left[ y_{\ess} y_{\tee} \rv Y_{\ess} \rv Y_{\tee} \right] \ge \kappa^2,
    $$
    then
    \begin{equation}\label{eq:fract_to_int_conclusion_1}
      \ProbOp_{z \sim \set{\rv Z^{\otimes}\vert_{(S,\sigma)}}}\left[\left\vert \E_{\ess \sim \Pi_k} y_{\ess} \cdot \lift(z)_\ess \right\vert \ge \sqrt{1-\beta} \cdot \kappa \right] \ge \frac{\beta \cdot \kappa^2}{4}.
    \end{equation}
\end{lemma}

\begin{proof}
  Let $\mu_{\ess,\tee} \coloneqq \norm{\set{\rv Z_{\ess} \rv Z_{\tee}\vert \rv Z_s=\sigma} - \prod_{i\in \ess} \set*{\rv Z_i \vert \rv Z_s=\sigma} \prod_{i \in \tee} \set*{\rv Z_i \vert \rv Z_s=\sigma}}_1$.
  Using triangle inequality and simplifying, we get
  \begin{align*}
    \mu_{\ess,\tee} \le & \norm{\set{\rv Z_{\ess} \rv Z_{\tee}\vert \rv Z_s=\sigma} - \set*{\rv Z_{\ess} \vert \rv Z_s=\sigma} \set*{\rv Z_{\tee} \vert \rv Z_s=\sigma}}_1\\
    & + \norm{\set{\rv Z_{\ess} \vert \rv Z_s=\sigma} - \prod_{i\in \ess} \set*{\rv Z_i \vert \rv Z_s=\sigma}}_1  + \norm{\set{\rv Z_{\tee} \vert \rv Z_s=\sigma} - \prod_{i\in \tee} \set*{\rv Z_i \vert \rv Z_s=\sigma}}_1.
  \end{align*}
  From our assumptions~\cref{eq:brs_hdx} and~\cref{eq:hdx_extra_indep},
  it follows that $\ExpOp_{\ess,\tee \sim \Pi_k^{
      2}} \mu_{\ess,\tee} \le 3\cdot \mu$. Using the fact that
  $\abs{y_{\ess} y_{\tee}} = 1$ and H\"older's inequality, we get
  \begin{align*}
    \E_{\ess,\tee \sim \Pi^2_k} \E_{\set{\rv Z^{\otimes}\vert_{(S,\sigma)}}} \left[ y_{\ess} y_{\tee} \rv Y_{\ess} \rv Y_{\tee} \right] &\ge \E_{\ess,\tee \sim \Pi^2_k} \widetilde{\E}_{\vert S,\sigma} \left[ y_{\ess} y_{\tee} \rv Y_{\ess} \rv Y_{\tee} \right] -
        \E_{\ess,\tee \sim \Pi^2_k} \mu_{\ess,\tee}\\
                                                                                         &\ge \E_{\ess,\tee \sim \Pi^2_k} \widetilde{\E}_{\vert S,\sigma} \left[ y_{\ess} y_{\tee} \rv Y_{\ess} \rv Y_{\tee} \right] - 3 \cdot \mu \ge \left(1 - \frac{\beta}{2}\right) \cdot \kappa^2.
  \end{align*}
  Alternatively,
  $$
  \E_{\ess,\tee \sim \Pi^2_k} \E_{\set{\rv Z^{\otimes}\vert_{(S,\sigma)}}} \left[ y_{\ess} y_{\tee} \rv Y_{\ess} \rv Y_{\tee} \right] = \E_{z \sim \set{\rv Z^{\otimes}\vert_{(S,\sigma)}}} \left(\E_{\ess \sim \Pi_k} y_{\ess} \cdot \lift(z)_\ess \right)^2 \ge \left(1 - \frac{\beta}{2}\right) \cdot \kappa^2.
  $$
  Define the random variable $\rv R \coloneqq \left(\E_{\ess \sim \Pi_k} [y_{\ess} \cdot \lift(z)_\ess] \right)^2$.
  Using~\cref{claim:first_moment_bound} with approximation parameter $\beta/2$, we get
  $$
  \E\left[\rv R \right] \ge \left(1 - \frac{\beta}{2}\right) \cdot \kappa^2 \implies \Pr\left[\rv R \ge (1 - \beta) \cdot \kappa^2 \right] \ge \frac{\beta \cdot \kappa^2}{4},
  $$
  from which~\cref{eq:fract_to_int_conclusion_1} readily follows.
\end{proof}

To formalize the notion of a word being recoverable with respect to
the full pseudo-expectation rather than in a given slice we will need
two additional events. The first event captures correlation as
follows.

\begin{definition}[$y$-Correlated Event]
   Let $\kappa \in (0,1]$ and $y \in \{\pm 1\}^{X(k)}$. We define the event
   $C_{\kappa}(y)$ as
   $$
   C_{\kappa}(y) \coloneqq \left\{(m,S,\sigma) \in \Omega~\vert~ \E_{\ess,\tee \sim \Pi^2_k} \widetilde{\E}_{\vert S,\sigma} \left[ y_{\ess} y_{\tee} \rv Y_{\ess} \rv Y_{\tee} \right] \ge \kappa^2 \right\}.
   $$
\end{definition}

The second event is a restriction of the first where we also require
the slice to satisfy the two-step tensorial condition
from~\cref{def:two_step_tensorial}.

\begin{definition}[$y$-Recoverable Event]
  Let $\kappa,\mu \in (0,1]$ and $y \in \{\pm 1\}^{X(k)}$. We define the
   event $R_{\kappa,\mu}(y)$ as
  $$
  R_{\kappa,\mu}(y) \coloneqq K_{\mu} \cap P_{\mu} \cap C_{\kappa}(y).
  $$
\end{definition}

\cref{lemma:fractional_to_integral} motivates the following
``recoverability'' condition.

\begin{definition}[Recoverable Word]
  Let $\kappa,\mu \in (0,1]$ and $y \in \set{\pm 1}^{X(k)}$. We say that $y$
  is \text{$(\kappa,\mu)$-recoverable} provided
  $$
  \ProbOp_{(m,S,\sigma)\sim\Omega}\left[ R_{\kappa,\mu}(y) \right] > 0.
  $$
\end{definition}

One of the central results in our framework is the following
``recoverability'' lemma. It embodies the power SOS brings to our
framework.

\begin{lemma}[Recoverability lemma]\label{lemma:recoverability}
  Let $\mathcal{C}_k$ be a lifted code on $X(\le k)$ with $X(1) = [n]$
  and distance at least $1/2 -\epsilon$. Let $\tilde{y} \in \set{\pm
    1}^{X(k)}$ be a word promised to be \lict{\sqrt{\epsilon}} to
  $\mathcal{C}_k$ and let $\mathcal{L} =
  \mathcal{L}(\tilde{y},\mathcal{C}_k)$ be its code list.

  Let $\theta \in (0,1]$ be arbitrary and set $\mu = \kappa \cdot
    \theta/2$ and $\kappa = (4 -\theta) \cdot \epsilon$. Suppose $\rv Z =
    \set{\rv Z_1, \ldots, \rv Z_n}$ is an $(L+2k)$-local PSD ensemble
    which is a solution to the~\ref{sos:list_dec} with
    objective value $\Psi$ within $\eta$ additive value from the
    optimum where $0 \le \eta \le \theta^2 \cdot \epsilon^4$.

    If $X(k)$ is $(\mu^4/4,L)$-two-step tensorial, then every $y \in
    \mathcal{L}$ is $(\kappa,\mu)$-recoverable. In particular, for
    every $\theta \in (0,1)$ and under the preceding assumptions, we
    have that every $y \in \mathcal{L}$ is $\left((4-\theta)\cdot
    \epsilon,\theta\right)$-recoverable.
\end{lemma}

\begin{proof}
  First observe that since $\tilde{y}$ is \lict{\sqrt{\epsilon}} to
  $\mathcal{C}_k$ the~\ref{sos:list_dec} is feasible and so
  the solution $\rv Z$ is well defined.  Towards a contradiction with the
  $\eta$-optimality of the SOS solution $\rv Z$, suppose there exists a
  word $y \in \mathcal{L}$ that is not $(\kappa,\mu)$-recoverable. Let
  $z \in \set{\pm 1}^{X(1)}$ be such that $y
  = \lift(z)$. Then \begin{align*} 1 ~=~ \ProbOp_{(m,S,\sigma)\sim\Omega}\left[
  R_{\kappa,\mu}(y)^c \right] ~\le~ \ProbOp_{(m,S,\sigma)\sim\Omega}\left[
  K_{\mu}^c \right] + \ProbOp_{(m,S,\sigma)\sim\Omega}\left[ P_{\mu}^c \right]
  + \ProbOp_{(m,S,\sigma)\sim\Omega}\left[ C_{\kappa}(y)^c \right].  \end{align*}
  
  Using~\cref{claim:glorified_concentration}, we get
  \begin{equation}\label{eq:prob_lower_bound}
    \ProbOp_{(m,S,\sigma)\sim\Omega}\left[ C_{\kappa}(y)^c \right] ~\ge~ 1 - \mu^2.
  \end{equation}

  Since $\widetilde{\E}$ is a valid solution to
  the~\ref{sos:list_dec},~\cref{lemma:cor_imp_entropic_bound}
  implies the lower bound 
  \begin{equation}\label{eq:entropic_lower_bound}
    \Psi\left(\{\rv Y_{\ess}\}_{\ess \in X(k)}\right) ~\ge~ 16 \cdot \epsilon^2.
  \end{equation}

  By definition, for $(m,S,\sigma) \in C_{\kappa}(y)^c$ we have
  $$
  \E_{\ess,\tee \sim \Pi^2_k} \widetilde{\E}_{\vert S,\sigma} \left[ y_{\ess} y_{\tee} \rv Y_{\ess} \rv Y_{\tee} \right] \le \kappa^2,
  $$
  implying
  $$
  \widetilde{\E} \left[\left( \E_{\ess \sim \Pi_k} y_{\ess} \cdot \rv Y_{\ess} \right)^2 \right] \le \E_{m,S,\sigma} \widetilde{\E}_{\vert S,\sigma} \left[\left( \E_{\ess \sim \Pi_k} y_{\ess} \cdot \rv Y_{\ess} \right)^2 \cdot \one_{C_{\kappa(y)}^c} \right] + \ProbOp_{(m,S,\sigma)\sim\Omega}\left[ C_{\kappa}(y) \right] \le \kappa^2 + \mu^2.
  $$
  
  Let $\widetilde{\E}$ be the pseudo-expectation of the ground set
  ensemble $\rv Z$ and let $\E'$ be the expectation on the delta
  distribution $\delta_z$. Note that the pseudo-expectation obtained
  from $\E'$ is a valid solution to the~\ref{sos:list_dec}. Since
  $$
  \kappa^2+\mu^2 \le \left(1+\frac{\theta^2}{4}\right) \cdot \kappa^2 = \left(1+\frac{\theta^2}{4}\right) \cdot \left(4 - \theta \right)^2 \cdot \epsilon^2 \le \left(16 - 2 \cdot \theta \right)\cdot \epsilon^2,
  $$
  and $\theta \ge 0$,~\cref{lemma:progress_lemma} gives that there
  is a convex combination of $\widetilde{\E}$ and $\E'$ such that the
  new $\Psi$, denoted $\Psi'$, can be bounded as
  $$
  \Psi' ~\le~ \Psi - \frac{\left(\Psi - \left(\kappa^2 + \mu^2\right)\right)^2}{2} \le \Psi - 2 \cdot \theta^2 \cdot \epsilon^4,
  $$
  contradicting the $\eta$-optimality of the SOS solution $\rv Z$ since $\eta \le \theta^2 \cdot \epsilon^4$.
\end{proof}

\subsubsection{Coupled Pairs, Coupled Lists, and Covers}\label{subsec:coupled_pairs_list_and_cover}

The~\ref{sos:list_dec} minimizing $\Psi$ was instrumental
to ensure that every $y' \in \mathcal{L}(\tilde{y},\mathcal{C}_k)$ is
recoverable in the sense of the conclusion
of~\cref{lemma:recoverability}. Unfortunately, this guarantee is
somewhat weak, namely, associated to every $y' \in
\mathcal{L}(\tilde{y},\mathcal{C}_k)$ there is a slice $(m,S,\sigma)$
from which we can sample $y$ (our approximation of $y'$) satisfying
\begin{equation}\label{eq:sos_weak_guarantee}
  \lvert \E_{\ess \sim \Pi_k} y_{\ess} \cdot y_{\ess}' \rvert > C \cdot \epsilon,
\end{equation}
where $C$ is a constant strictly smaller than $4$. A priori this seems
insufficient for our list decoding task. However, there are two
properties which will help us with list decoding. The first is that SOS finds not only $y$
but also $z \in \set{\pm 1}^{X(1)}$ such that $y =
\lift(z)$.  The second property is that the lifting is robust: even
the weak agreement given by~\cref{eq:sos_weak_guarantee} translates
into a much stronger agreement in the ground set between $z$ and $z'
\in \mathcal{C}_1$ where $y'=\lift(z')$. This stronger agreement on
the ground set can be used to ensure that $z$ (or $-z$) lies inside
the unique decoding ball of $z'$ in the base code $\mathcal{C}_1$.

To study this coupling phenomenon between words in the lifted space
$\set{\pm 1}^{X(k)}$ and on the ground space $\set{\pm 1}^{X(1)}$ we
introduce some terminology. The most fundamental one is a coupled
pair.

\begin{definition}[Coupled Pair]
  Let $z \in \set{\pm 1}^{X(1)}$ and $y \in \set{\pm 1}^{X(k)}$. We
  say that $(z,y)$ is a coupled pair with respect to a lift function
  $\lift$ provided $y = \lift(z)$.
\end{definition}

\begin{remark}
  If the function $\lift$ is clear in the context, we may
  assume that the coupled pair is with respect to this function.
\end{remark}

Coupled pairs can be combined in a list.
\begin{definition}[Coupled List]
  We say that a list $\mathcal{M} =
  \set{(z^{(1)},y^{(1)}),\dots,(z^{(h)},y^{(h)})}$ is coupled with
  respect to lift function $\lift$ provided
  $(z^{(i)},y^{(i)})$ is a coupled pair for every $i$ in $[h]$.
\end{definition}

A coupled list can ``cover'' a list of words in the lifted space
$\set{\pm 1}^{X(k)}$ as defined next. 

\begin{definition}[Coupled Bias Cover]
  Let $\mathcal{M} = \set{(z^{(1)},y^{(1)}),\dots,(z^{(h)},y^{(h)})}$ be a coupled list and $\mathcal{L} \subset \set{\pm 1}^{X(k)}$.
  We say that $\mathcal{M}$ is a $\delta$-bias cover of $\mathcal{L}$ provided
  $$
  \left(\forall y' \in \mathcal{L} \right)\left(\exists (z,y) \in \mathcal{M}\right)\left(\lvert \E_{\ess \sim \Pi_k}  y_{\ess}' \cdot y_{\ess} \rvert > \delta \right).
  $$
\end{definition}

A $\delta$-bias cover for ``small'' $\delta$ might seem a rather weak
property, but as alluded to, when combined with enough robustness of
the lifting, it becomes a substantial guarantee enabling list
decoding.

\subsubsection{Cover Retrieval}\label{subsec:cover_retrieval}

When the code list $\mathcal{L}(\tilde{y},\mathcal{C}_k)$ becomes
recoverable in the SOS sense as per~\cref{lemma:recoverability}, we
still need to conduct local rounding on the slices to collect a bias
cover.  Recall that this local rounding is probabilistic
(c.f.~\cref{lemma:fractional_to_integral}), so we need to repeat this
process a few times to boost our success probability\footnote{In fact,
  this process can be derandomized using standard techniques in our
  instantiations. See~\cref{lemma:derand_maj} for detais.}. This is accomplished
by~\cref{algo:cover_retrieval}.

\begin{algorithm}{Cover Retrieval Algorithm}{An $(L+2k)$-local PSD ensemble $\rv Z$ which is a $(\theta^2\epsilon^4)$-optimal solution to the \ref{sos:list_dec}.}
                 {A $2\epsilon$-bias cover $\mathcal{M}$ for $\mathcal{L}(\tilde{y},\mathcal{C}_k)$.}\label{algo:cover_retrieval}
    \begin{enumerate}
       \item Let $\mathcal{M} = \emptyset$
        \item Let $T = 4 \cdot \ln(\vert\Omega \vert) \cdot n / (\beta \cdot \epsilon^2)$
        \item For $(m,S,\sigma) \in \Omega$ do
        \item  \qquad If $(m,S,\sigma) \in K_{\mu} \cap P_{\mu}$ then
        \item  \qquad \qquad Run Propagation Rounding $T$ times conditioned on $(m,S,\sigma)$
        \item  \qquad \qquad Let $\mathcal{M}\vert_{m,S,\sigma} = \set{(z^{(1)}, y^{(1)}), \dots,(z^{(T)}, y^{(T)})}$ be the coupled list
        \item  \qquad \qquad Set $\mathcal{M} = \mathcal{M} \cup \mathcal{M}\vert_{m,S,\sigma}$
        \item Output $\mathcal{M}$.
    \end{enumerate}
\end{algorithm}

The correctness of~\cref{algo:cover_retrieval} follows easily given
the properties established so far.

\begin{lemma}[Cover lemma]\label{lemma:cover}
  Let $\beta \in (0,1)$.  Suppose that $\lift$ is a
  $(1/2-\epsilon_0,1/2-\epsilon)$-robust $(\beta^4 \cdot
  \epsilon^8/2^{18},L)$-two-step tensorial lifting from $\mathcal{C}_1$
    to $\mathcal{C}_k$. Let $\tilde{y} \in \set{\pm 1}^{X(k)}$ be
  \lict{\sqrt{\epsilon}} to $\mathcal{C}_k$. If $\theta \le \beta
  \cdot \epsilon/2^4$, then w.v.h.p.\footnote{The abbreviation
    w.v.h.p.~stand for \textit{with very high probability} and means
    with probability $1-\exp(-\Theta(n))$.}~the Cover Retrieval
  algorithm~\ref{algo:cover_retrieval} returns a $\delta$-bias cover
  $\mathcal{M}$ of the code list
  $\mathcal{L}(\widetilde{y},\mathcal{C}_k)$ where $\delta =
  (4-\beta) \cdot \epsilon$. Furthermore, the running time is at most
  $n^{O(L+k)}/(\beta \cdot \epsilon^2)$ where $n=\vert X(1) \vert$.
\end{lemma}

\begin{proof}
  Let $\rv Z = \set{\rv Z_1, \ldots, \rv Z_n}$ be an $\eta$-optimum
  solution to the~\ref{sos:list_dec} where $\eta \le \theta^2
  \cdot \epsilon^4$ and $\theta = \beta \cdot \epsilon/2^4$. By our
  $(\beta^4 \cdot \epsilon^8/2^{18},L)$-two-step tensorial assumption
  and our
  choice of SOS degree for the~\ref{sos:list_dec}, we can
  apply~\cref{lemma:recoverability} to conclude that every $y \in
  \mathcal{L} = \mathcal{L}(\tilde{y},\mathcal{C}_k)$ is $((4-\theta)
  \cdot \epsilon,(4-\theta) \cdot \epsilon \cdot \theta/2)$-recoverable.
  Then for $y
  \in \mathcal{L}$, there exists $(m,S,\sigma) \in \Omega$ such
  that~\cref{lemma:fractional_to_integral} yields
  $$
  \ProbOp_{z \sim \set{\rv Z^{\otimes}\vert_{(S,\sigma)}}}\left[\left\vert \E_{\ess \sim \Pi_k} y_{\ess} \cdot \lift(z)\right\vert \ge (4-\beta)\cdot \epsilon \right] \ge \frac{\beta \cdot (4-\theta)^2 \cdot \epsilon^2}{32} \ge \frac{\beta\cdot \epsilon^2}{4}.
  $$
  where $\set{\rv Z^{\otimes}\vert_{(S,\sigma)}}$ (c.f.~\cref{def:prod_dist_on_slice}) is the
  product distribution of the marginal distributions after
  conditioning the ensemble on slice $(m,S,\sigma)$.  By sampling
  $\set{\rv Z^{\otimes}\vert_{(S,\sigma)}}$ independently $T$ times we
  obtain $z^{(1)},\dots,z^{(T)}$ and thus also the coupled list
  $$
  \mathcal{M}\vert_{m,S,\sigma} = \set{(z^{(1)},y^{(1)}), \dots,(z^{(T)},y^{(T)})},
  $$
  where $y^{(i)} = \lift(z^{(i)})$. Then
  \begin{align*}
    \ProbOp_{z^{(1)},\dots,z^{(T)} \sim \set{\rv Z^{\otimes}\vert_{(S,\sigma)}}^{\otimes T}}\left[\forall i \in [T]~\colon~ \left\vert \E_{\ess \sim \Pi_k} y_{\ess} \cdot \lift\left(z^{(i)}\right) \right\vert < (4-\beta) \cdot \epsilon \right] & ~\le~ \exp{\left(- \frac{\beta \cdot \epsilon^2 \cdot T}{4}\right)}\\ & ~\le~ \frac{\exp{\left(-n\right)}}{\vert \Omega \vert},
  \end{align*}
  where the last inequality follows from our choice of $T$. Then by
  union bound
  $$
  \ProbOp\left[\mathcal{M} \textup{ is not a } 2\epsilon\textup{-bias cover of } \mathcal{L}\right] \le \vert \mathcal{L} \vert \cdot \frac{\exp{\left(-n\right)}}{\vert \Omega \vert} \le \exp{\left(-n\right)},
  $$
  concluding the proof.
\end{proof}

\subsubsection{Cover Purification and Robustness}\label{subsec:cover_purification}

Now we consider the third and final stage of the list decoding
framework. We show how despite the weak guarantee of the bias cover
returned by the Cover Retrieval~\cref{algo:cover_retrieval} we can do
a further processing to finally obtain the coupled code list
$\mathcal{L}(\tilde{y},\mathcal{C}_1,\mathcal{C}_k)$ provided the
lifting admits some \textit{robustness} properties. We first develop
these properties 
and later present this process, denoted Cover Purification.

\subsubsection*{Further Lifting Properties}\label{subsubsec:extra_rob_prop}

Given two coupled pairs $(z,y=\lift(z))$ and $(z',y'=\lift(z'))$
(where $z \in \mathcal{C}_1$), we show how weak agreement between $y$
and $y'$ on the lifted space is enough to provide non-trivial
guarantees between $z$ and $z'$ as long as the lifting admits
appropriate \textit{robustness}.

\begin{claim}[Coupled unique decoding from distance]\label{claim:coupled_unique_decoding_dist}
  Suppose that $\lift$ is a $(1/4-\epsilon_0/2,1/2-\epsilon)$-robust
  lifting from $\mathcal{C}_1$ to $\mathcal{C}_k$. Let $(z,y)$ and
  $(z',y')$ be coupled pairs. If $y \in \mathcal{C}_k$ (equivalently
  $z \in \mathcal{C}_1$) and $\Delta\parens{y,y'} < 1/2 - \epsilon$, then
  $\Delta\parens*{z,z'} \le 1/4 - \epsilon_0/2$, i.e., $z'$ is within the unique
  decoding radius of $z$.
\end{claim}

\begin{proof}
  Towards a contradiction suppose that $\Delta\parens*{z, z} \ge 1/4 -
  \epsilon_0/2$. Since the lifting is $(1/4-\epsilon_0/2,1/2-\epsilon)$-robust,
  this implies that $\Delta\parens*{y, y'} \ge 1/2 - \epsilon$ contradicting our
  assumption.
\end{proof}

From bias amplification (i.e., parity sampling), we
deduce~\cref{claim:coupled_unique_decoding_bias_1}.

\begin{claim}[Coupled unique decoding from bias I]\label{claim:coupled_unique_decoding_bias_1}
  Suppose $\lift$ is a $(1/2 - \epsilon_0,1/2 -\epsilon)$-robust linear
  lifting from $\mathcal{C}_1$ to $\mathcal{C}_k$ which is also
  a $(1/2 + \epsilon_0, 2 \cdot \epsilon)$-parity sampler. Let
  $(z,y)$ and $(z',y')$ be coupled pairs.  If $y \in \mathcal{C}_k$
  (equivalently $z \in \mathcal{C}_1$) and $\lvert \E_{\ess \sim
    \Pi_k} [y_{\ess} \cdot y'_{\ess}] \rvert > 2 \cdot \epsilon$, then
  $$
  \lvert \E_{i \sim \Pi_1} [z_i\cdot z'_i] \rvert \ge 1/2 + \epsilon_0,
  $$
  i.e., either $z'$ or $-z'$ is within the unique decoding radius of
  $z$.
\end{claim}

\begin{proof}
  The verification follows easily from our assumptions. Towards a
  contradiction suppose that $\lvert \E_{i \sim \Pi_1} [z_i\cdot
  z'_i] \rvert < 1/2 + \epsilon_0$, i.e., the word $z'' = z \cdot z'$ has
  bias at most $1/2 + \epsilon_0$. Using the assumption that the lift
  is linear, we have $\lift(z'') = \lift(z) \cdot \lift(z')$.  Since
  the lifting takes bias $1/2 + \epsilon_0$ to $2\cdot \epsilon$, we have
  $$
  \textup{bias}(\lift(z) \cdot \lift(z')) = \textup{bias}(\lift(z'')) \le 2 \cdot \epsilon,
  $$
  or equivalently $\lvert \E_{\ess \sim \Pi_k} [y_{\ess} \cdot
  y'_{\ess}] \rvert \le 2 \cdot \epsilon$ contradicting our assumption.
\end{proof}

If the lifting function is odd, then we
obtain~\cref{claim:coupled_unique_decoding_bias_2}.

\begin{claim}[Coupled unique decoding from bias II]\label{claim:coupled_unique_decoding_bias_2}
  Suppose $\lift$ is a $(1/4 - \epsilon_0/2, 1/2 -\epsilon)$-robust 
  lifting from $\mathcal{C}_1$ to $\mathcal{C}_k$ which is odd, i.e.,
  $\lift(-z) = -\lift(z)$. Let
  $(z,y)$ and $(z',y')$ be coupled pairs.  If $y \in \mathcal{C}_k$
  (equivalently $z \in \mathcal{C}_1$) and $\lvert \E_{\ess \sim
    \Pi_k} [y_{\ess} \cdot y'_{\ess}] \rvert > 2 \cdot \epsilon$, then
  either $z'$ or $-z'$ is within the unique decoding radius of $z$.
\end{claim}

\begin{proof}
  Since $\lvert \E_{\ess \sim \Pi_k} [y_{\ess} \cdot y'_{\ess}] \rvert > 2 \cdot \epsilon$ and the lifting
  is odd, either
  $$
  \E_{\ess \sim \Pi_k} [y_{\ess} \cdot \lift(z')_\ess] > 2 \cdot \epsilon,
  $$
  or
  $$
    \E_{\ess \sim \Pi_k} [y_{\ess} \cdot \lift(-z')_\ess] = \E_{\ess \sim \Pi_k} \left[-y_{\ess} \cdot \lift(z')_\ess \right] > 2 \cdot \epsilon.
  $$
  Then either $\Delta\parens{y,\lift(z')} \le 1/2-\epsilon$ or $\Delta\parens{y ,\lift(-z')} \le 1/2-\epsilon$.
  Using~\cref{claim:coupled_unique_decoding_dist} we conclude the proof.
\end{proof}

\subsubsection*{Cover Purification}\label{subsubsec:cover_purification}

A $\delta$-bias cover $\mathcal{M}$ of $\mathcal{L}$ for small
$\delta$ may require further processing in order to actually retrieve
$\mathcal{L}$. Provided the lifting is sufficiently robust, trying to
unique decode $z$ for $(z,y) \in \mathcal{M}^{\pm}$, where
$\mathcal{M}^{\pm}$ is the sign completion as defined next, and then
lifting the decoded word yields a new coupled list that contains
$\mathcal{L}$. This process is referred to as cover purification and
its formalization is the object of this section.

\begin{definition}[Sign Completion]
  Let $\mathcal{M}$ be coupled list. We say that $\mathcal{M}^{\pm}$ defined as
  $$
  \mathcal{M}^{\pm} \coloneqq \left\{(z,\lift(z)), (-z,\lift(-z)) ~\vert~ (z,y) \in \mathcal{M}\right\},
  $$
  is the sign completion of $\mathcal{M}$.
\end{definition}

The correctness of the cover purification process is established next.
\begin{lemma}[Purification lemma]\label{lemma:purification}
  Suppose $\lift$ is a $(1/2 - \epsilon_0,1/2 -\epsilon)$-robust
  lifting from $\mathcal{C}_1$ to $\mathcal{C}_k$ which is either
  \begin{itemize}
    \item linear and a $(1/2 + \epsilon_0, 2\cdot \epsilon)$-parity sampler; or
    \item $(1/4-\epsilon_0/2)$-robust and odd.
  \end{itemize}
  Let $\tilde{y} \in \set{\pm 1}^{X(k)}$ be \lict{\sqrt{\epsilon}} to
  $\mathcal{C}_k$ and $\mathcal{L}
  = \mathcal{L}(\tilde{y}, \mathcal{C}_k)$ be its code list. If
  $\mathcal{M} = \set{(z^{(i)},y^{(i)})\vert i \in [h]}$ is a
  $2\epsilon$-bias cover of $\mathcal{L}$, then
  $$
  \mathcal{L} \subseteq \left\{\lift(z) ~\vert~ z \in \textup{Dec}_{\mathcal{C}_1}\left(\bf{\textup{P}}_1\left( \mathcal{M}^{\pm} \right)\right)  \right\} \eqqcolon \mathcal{L}',
  $$  
  where $\bf{\textup{P}}_1$ is the projection on the first coordinate and $\textup{Dec}_{\Cc_1}$
  is a unique decoder for $\Cc_1$.
  Furthermore, $\mathcal{L}'$ can be computed in time $O\left(\left\vert
  \mathcal{M} \right\vert \cdot f(n) \right)$ where $f(n)$ is the running time
  of a unique decoding algorithm of $\mathcal{C}_1$.
\end{lemma}

\begin{proof}
  Let $y \in \mathcal{L}$. By the $2\epsilon$-cover property, there exists
  a coupled pair $(z',y') \in \mathcal{M}$ satisfying $\lvert \E_{\ess \sim \Pi_k} [y_{\ess} \cdot y'_{\ess}] \rvert >
  2 \cdot \epsilon$. Combining this bound with the appropriate robustness assumptions,
  \cref{claim:coupled_unique_decoding_bias_1} or \cref{claim:coupled_unique_decoding_bias_2}
  yields that either $z'$ or $-z'$ can be uniquely decoded in $\mathcal{C}_1$. Then
  $$
  y \in \left\{\lift(z) ~\vert~ z \in \textup{Dec}_{\mathcal{C}_1}\left( \bf{\textup{P}}_1\left( \mathcal{M}^{\pm} \right) \right)  \right\}.
  $$
  Finally, observe that computing $\mathcal{L}'$ with the claimed
  running time is straightforward.
\end{proof}

Algorithmically, cover purification works by running the unique decoding algorithm of $\Cc_1$ on every element of the sign completion $\mathcal{M}^{\pm}$, described below in \cref{algo:purification}.

\begin{algorithm}{Cover Purification Algorithm}{A $2\epsilon$-bias cover $\mathcal{M}$ for $\mathcal{L}(\tilde{y},\mathcal{C}_k)$.}
                                         {Coupled List $\mathcal{L}'$ s.t. $\textup{P}_2(\mathcal{L}') \supseteq \mathcal{L}(\tilde{y},\mathcal{C}_k)$.}\label{algo:purification}
    \begin{enumerate}
       \item Let $\mathcal{L}' = \emptyset$
        \item For $(z',y') \in \mathcal{M}^{\pm}$ do
        \item  \qquad If $z'$ is uniquely decodable in $\mathcal{C}_1$ then
        \item  \qquad \qquad Let $z = \textup{UniqueDecode}_{\mathcal{C}_1}(z')$
        \item  \qquad \qquad Let $y = \lift(z)$
        \item  \qquad \qquad Set $\mathcal{L}' = \mathcal{L}' \cup \{(z,y)\}$
        \item Output $\mathcal{L}'$.
    \end{enumerate}
\end{algorithm}

\subsubsection{Correctness of the List Decoding Algorithm}\label{subsec:list_dec_correctness}

The building blocks developed so far are assembled to form the final list decoding algorithm (\cref{algo:list_decoding}), which is restated below for
convenience.

\begin{algorithm}{List Decoding Algorithm}{A word $\tilde{y} \in \set{\pm 1}^{X(k)}$ \lict{\sqrt{\epsilon}} to $\mathcal{C}_k=\lift(\mathcal{C}_1)$}
                                          {Coupled code list $\mathcal{L}(\tilde{y},\mathcal{C}_1,\mathcal{C}_k)$.}\label{algo:list_decoding_restated}
    \begin{enumerate}                   
        \item Solve the~\ref{sos:list_dec} with $\eta$-accuracy obtaining $\rv Z$ where $\eta = \epsilon^8/2^{22}$ 
        \item Let $\mathcal{M}$ be the output of the Cover Retrieval~\cref{algo:cover_retrieval} on $\rv Z$
        \item Let $\mathcal{L'}$ be the output of the Cover Purification~\cref{algo:purification} on $\mathcal{M}$
        \item Let $\mathcal{L}'' = \{(z,y) \in \mathcal{L}'~\vert~ \Delta\parens{\tilde{y},y} \le 1/2 - \sqrt{\epsilon} \}$          
        \item Output $\mathcal{L}''$.
    \end{enumerate}
\end{algorithm}

We are ready to prove the main theorem of the abstract list decoding
framework which follows easily from the properties developed so far.
\begin{theorem}[List Decoding Theorem (Restatement of~\cref{theo:list_dec})]\label{theo:list_dec_restated}
  Suppose that $\lift$ is a $(1/2-\epsilon_0,1/2-\epsilon)$-robust
  $(\epsilon^8/2^{22},L)$-two-step tensorial lifting from
  $\mathcal{C}_1$ to $\mathcal{C}_k$ which is either
  \begin{itemize}
    \item linear and a $(1/2+\epsilon_0, 2\cdot \epsilon)$-parity sampler; or
    \item $(1/4 - \epsilon_0, 1/2 - \epsilon/2)$-robust and odd.
  \end{itemize}
    Let $\tilde{y} \in \set{\pm 1}^{X(k)}$ be \lict{\sqrt{\epsilon}} to
  $\mathcal{C}_k$. Then w.v.h.p.~the List
  Decoding~\cref{algo:list_decoding} returns the coupled code list
  $\mathcal{L}(\widetilde{y},\mathcal{C}_1,\mathcal{C}_k)$. Furthermore, the
  running time is
  $$
  n^{O(L+k)}\left(\polylog(\epsilon^{-1}) + f(n)\right),
  $$
  where $n=\vert X(1) \vert$ and $f(n)$ is the running time of a
  unique decoding algorithm of $\mathcal{C}_1$.
\end{theorem}

\begin{proof}
  Under the assumptions of the theorem,~\cref{lemma:cover} establishes
  that the Cover Retrieval~~\cref{algo:cover_retrieval} returns
  w.v.h.p. a $2\epsilon$-bias cover. Then, \cref{lemma:purification}
  states that providing this $2\epsilon$-bias cover as input to the Cover
  Purification~\cref{algo:purification} yields a coupled list
  containing the code list
  $\mathcal{L}(\widetilde{y},\mathcal{C}_1,\mathcal{C}_k)$. Finally, the
  last step in~\cref{algo:list_decoding} ensures the output is
  precisely $\mathcal{L}(\widetilde{y},\mathcal{C}_1,\mathcal{C}_k)$.
\end{proof}

\section{Instantiation I: Direct Sum on HDXs}\label{sec:list_dec_xor_hdx}

We instantiate the list decoding framework to the direct sum lifting
on HDXs obtaining~\cref{theo:direct_sum_general_codes_hdx}, which is
the main result in this section. For this instantiation we need to
establish that HDXs are two-step tensorial which will be done
in~\cref{subsec:hdx_two_step_tensorial}.

\begin{theorem}[Direct Sum Lifting on HDX]\label{theo:direct_sum_general_codes_hdx}
  Let $\epsilon_0 < 1/2$ be a constant and $\epsilon \in (0,\epsilon_0)$.  There
  exist universal constants $c,C > 0$ such that for any $\gamma$-HDX
  $X(\le d)$ on ground set $X(1)=[n]$ and $\Pi_1$ uniform, if
  $$
  \gamma \le (\log(1/\epsilon))^{-C \cdot \log(1/\epsilon)}\quad\text{ and }\quad d \ge c \cdot \frac{(\log(1/\epsilon))^2}{\epsilon},
  $$
  then the following holds:

  For every binary code $\mathcal{C}_1$ with $\Delta(\mathcal{C}_1) \ge 1/2-\epsilon_0$ on $X(1)=[n]$,
  there exists a binary lifted code $\mathcal{C}_k=\dsum_{X(k)}(\varphi(\mathcal{C}_1))$ 
  with $\Delta(\mathcal{C}_k) \ge 1/2-\epsilon^{\Omega_{\epsilon_0}(1)}$ on $X(k)$ where $k = O\left(\log(1/\epsilon)\right)$,
  $\varphi$ is an explicit linear projection, and
  \begin{itemize}
    \item{[Efficient List Decoding]} If $\tilde{y}$ is \lict{\sqrt{\epsilon}} to $\mathcal{C}_k$, then we
         can compute the list $\mathcal{L}(\tilde{y},\mathcal{C}_1,\mathcal{C}_k)$ (c.f.~\cref{def:coupled_code_list}) in time
         $$
         n^{\epsilon^{-O\left(1\right)}} \cdot f(n),
         $$
         where $f(n)$ is the running time of a unique decoding algorithm for $\mathcal{C}_1$.
    \item{[Rate]} The rate $r_k$ of $\mathcal{C}_k$ satisfies  $r_k = r_1 \cdot \abs{X(1)}/\abs{X(k)}$ where $r_1$ is the relative
         rate of $\mathcal{C}_1$.
    \item{[Linearity]} If $\mathcal{C}_1$ is linear, then $\varphi$ is the identity and $\mathcal{C}_k= \dsum_{X(k)}(\mathcal{C}_1)$
                       is linear.
  \end{itemize}
\end{theorem}

In particular, invoking~\cref{theo:direct_sum_general_codes_hdx} on
HDXs extracted from Ramanujan complexes (as in \cref{lemma:hdx_existence}), we
obtain~\cref{cor:ramanujan_instantiation}.

\begin{corollary}\label{cor:ramanujan_instantiation}
  Let $\epsilon_0 < 1/2$ be a constant and $\epsilon \in (0,\epsilon_0)$. There is an infinite
  sequence of HDXs $X_1,X_2,\dots$ on ground sets of size $n_1,n_2,\dots$ such that
  the following holds:
  
  For every sequence of binary codes $\mathcal{C}_1^{(i)}$ on $[n_i]$ with
  rate and distance uniformly bounded by $r_1$ and $(1/2-\epsilon_0)$
  respectively, there exists a sequence of binary lifted codes
  $\mathcal{C}_k^{(i)} =\dsum_{X(k)}(\varphi(\mathcal{C}_1^{(i)}))$ on a collection $X_i(k)$
  with $\Delta(\mathcal C_k^{(i)}) \geq 1/2-\epsilon^{\Omega_{\epsilon_0}(1)}$ where $\varphi$ is an explicit linear projection and
  \begin{itemize}
    \item{[Efficient List Decoding]} If $\tilde{y}$ is \lict{\sqrt{\epsilon}} to $\mathcal{C}_k$, then we can compute
          the list $\mathcal{L}(\tilde{y},\mathcal{C}_1,\mathcal{C}_k)$ (c.f.~\cref{def:coupled_code_list}) in time $n^{\epsilon^{-O\left(1\right)}} \cdot f(n)$, where
          $f(n)$ is the running time of a unique decoding algorithm of $\mathcal{C}_1$.
        \item{[Explicit Construction]} The collection $X_i(k)$ is part of an explicit $\gamma$-HDX $X_i(\le d)$ where $k = O\left(\log(1/\epsilon)\right)$,
              $d=O\left((\log(1/\epsilon))^2/\epsilon \right)$, and $\gamma = (\log(1/\epsilon))^{-O\left(\log(1/\epsilon)\right)}$.
    \item{[Rate]} The rate $r_k^{(i)}$ of $\mathcal{C}_k^{(i)}$ satisfies $r_k^{(i)} \ge r_1 \cdot \exp\left(-(\log(1/\epsilon))^{O\left( \log(1/\epsilon) \right)}\right)$.
    \item{[Linearity]} If $\mathcal{C}_1^{(i)}$ is linear, then $\varphi$ is the identity and $\mathcal{C}_k^{(i)}= \dsum_{X(k)}(\mathcal{C}_1^{(i)})$
                       is linear.    
  \end{itemize}
\end{corollary}

\begin{proof}
Efficient list decoding and linearity follow directly from \cref{theo:direct_sum_general_codes_hdx}, and the parameters of the explicit construction match the requirements of the theorem.
The only thing left to do is to calculate the rate.
Since the lifting $\dsum_{X_i(k)}$ needs to be a $(2\epsilon_0, 2\epsilon)$-parity sampler to achieve the promised distance, by \cref{lemma:direct_sum_rate} the rate $r_k^{(i)}$ of $\mathcal C_k^{(i)}$ satisfies
	$$r_k^{(i)} \geq r_1 \cdot \gamma^{O((\log(1/\epsilon))^4/(\epsilon^2 \gamma))}$$
Since $\gamma = (\log(1/\epsilon))^{-O(\log(1/\epsilon))}$, this reduces to
	$$r_k^{(i)} \geq r_1 \cdot (\log(1/\epsilon))^{-O((\log(1/\epsilon))^5/(\epsilon^2 \cdot \gamma))} = r_1 \cdot \exp(1/\gamma^{O(1)})
		= r_1 \cdot \exp\left(-(\log(1/\epsilon))^{O\left( \log(1/\epsilon) \right)}\right).$$
\end{proof}

\subsection{HDXs are Two-Step Tensorial}\label{subsec:hdx_two_step_tensorial}

\cref{theo:hdx_tensorial} proven in~\cite{AJT19} establishes that HDXs
of appropriate expansion parameter are tensorial objects for constant
$L = O_{k,q,\mu}(1)$.

\begin{theorem}[HDXs are Tensorial]\label{theo:hdx_tensorial}
  There exist some universal constants $c' \ge 0$ and $C' \ge 0$
  satisfying the following: If $L \ge c' \cdot (q^{k} \cdot
  k^5/\mu^4)$, $\supp(\rv Z_j) \le q$ for all $j \in [n]$, and $X$ is
  a $\gamma$-HDX for $\gamma \le C' \cdot \mu^4/(k^{8 + k} \cdot
  2^{6k} \cdot q^{2k})$ and size $\ge k$, then $X(k)$ endowed with a
  distribution $\Pi_k$ is $(\mu, L)$-tensorial.
\end{theorem}

The next result shows that HDXs are also two-step tensorial objects with the same parameters as above.
\begin{lemma}[HDXs are two-step tensorial]\label{lemma:double_close_to_product}    
    There exist some universal constants $c' \ge 0$ and $C' \ge 0$ satisfying the following:
    If $L \ge c' \cdot (q^{k} \cdot k^5/\mu^4)$, $\supp(\rv Z_j) \le q$
    for all $j \in [n]$, and $X$ is a $\gamma$-HDX for $\gamma \le C' \cdot \mu^4/(k^{8 + k} \cdot 2^{6k} \cdot q^{2k})$ and size $\ge k$, then $X(k)$ is $(\mu,L)$-two-step tensorial.
\end{lemma}

\begin{proof}
  Under our assumptions the $(\mu,L)$-tensorial property follows
  from~\cref{theo:hdx_tensorial} (this is the only place where the
  assumption on $\gamma$ is used), so we only need to show  
  $$
  \ExpOp_{\ess,\tee \sim \Pi_k}{ \norm{\set{\rv Z_{\ess}' \rv Z_{\tee}'} - \set*{\rv Z_{\ess}'}\set*{\rv Z_{\tee}'}}_1 } \le \mu,
  $$
  which can be proven by adapting a
  potential argument technique from~\cite{BarakRS11}. First, set the potential
    \begin{equation}
        \Phi_m = \ExpOp_{S \sim \Pi_k^m}{\ExpOp_{\sigma \sim \set{\rv Z_S}}{\ExpOp_{\ess \sim \Pi_k}{\Var{\rv Z_\ess \mid \rv Z_S = \sigma }}}},\label{eq:pot-def}
    \end{equation}
    and consider the error term
    \begin{equation}
        \mu_m \coloneqq \ExpOp_{S \sim \Pi_k^m}{\ExpOp_{\sigma \sim \set{\rv Z_S}}{D(S,\sigma)}},\label{eq:assume-large}
    \end{equation}
    where $D(S,\sigma) \coloneqq \Ex{\ess,\tee \sim \Pi_k}{\norm{\set{\rv Z_{\ess}\rv Z_{\tee} \mid \rv Z_S = \sigma} - \set{\rv Z_{\ess} | \rv Z_S = \sigma}\set{\rv Z_{\tee} | \rv Z_S = \sigma} }_1}$.
    If $\mu_m \ge \mu/2$, then
    \[ \ProbOp_{S \sim \Pi_k^m, \sigma \sim \{\rv Z_S\}}\left[ D(S,\sigma) \ge \mu_m/2 \right] \ge \frac{\mu}{4}. \] 
    
    Let $G = (V=X(k),E)$ be the weighted graph where $E
    = \{\{\ess,\tee\} \mid \ess,\tee \in X(k)\}$ and each edge
    $\{\ess,\tee\}$ receives weight $\Pi_k(\ess) \cdot \Pi_k(\tee)$.
    Local correlation (expectation over the edges) on this
    graph $G$ is the same as to global correlation (expectation over two independent
    copies of vertices). Then, we obtain~\footnote{See \cite{AJT19} or~\cite{BarakRS11} for the details.}
    \[
    \Phi_m -\Phi_{m + 1} \ge \ProbOp_{S \sim \Pi_k^m,\sigma \sim \{\rv Z_S\}}\left[D(S,\sigma) \ge \mu_m/2 \right] \cdot \frac{\mu^2}{2q^{2k}}.
    \]
    Since $1 \ge \Phi_1 \ge \cdots \ge \Phi_{L/k} \ge 0$, there can be
    at most $8 q^{2k} /\mu^3$ indices $m \in [L/k]$ such that
    $\mu_m \ge \mu/2$.  In particular, since the total number of indices is $L/k$, we have
    \[ \ExpOp_{m \in [L/k]}{\mu_m } \le \frac{\mu}{2} + \frac{k}{L} \cdot \frac{8 q^{2k}}{\mu^3}. \]
    Our choice of $L$ is more than enough to ensure $\Ex{m \in [L/k]}{\mu_m} \le \mu$.
\end{proof}

\subsection{Instantiation to Linear Base Codes}

First, we instantiate the list decoding framework to the seemingly
simpler case of binary \emph{linear} base codes
in~\cref{lemma:direct_sum_linear_biased_codes_hdx}. As we show later, with
a simple observation we can essentially use the proof
of~\cref{lemma:direct_sum_linear_biased_codes_hdx} to
obtain~\cref{theo:direct_sum_general_codes_hdx} for general codes.

\begin{lemma}[Direct sum lifting of linear biased codes]\label{lemma:direct_sum_linear_biased_codes_hdx}
  Let $\epsilon_0 < 1/2$ be a constant and $\epsilon \in (0,\epsilon_0)$.  There
  exist universal constants $c,C > 0$ such that for any $\gamma$-HDX
  $X(\le d)$ on ground set $X(1)=[n]$ and $\Pi_1$ uniform, if
  $$
  \gamma \le \log(1/\epsilon)^{-C \cdot (\log(1/\epsilon))}\quad\text{ and }\quad d \ge c \cdot \frac{(\log(1/\epsilon))^2}{\epsilon},
  $$
  then the following holds:

  For every binary \text{$2\epsilon_0$-biased} linear code $\mathcal{C}_1$  on $X(1)=[n]$,
  there exists a $2\epsilon$-biased binary lifted linear code
  $\mathcal{C}_k=\dsum_{X(k)}(\mathcal{C}_1)$ on $X(k)$ where $k = O\left(\log(1/\epsilon)\right)$ and
  \begin{itemize}
    \item{[Efficient List Decoding]} If $\tilde{y}$ is \lict{\sqrt{\epsilon}} to $\mathcal{C}_k$, then we
         can compute the list $\mathcal{L}(\tilde{y},\mathcal{C}_1,\mathcal{C}_k)$ (c.f.~\cref{def:coupled_code_list}) in time
         $$
         n^{\epsilon^{-O\left(1\right)}} \cdot f(n),
         $$
         where $f(n)$ is the running time of a unique decoding algorithm for $\mathcal{C}_1$.
    \item{[Rate]} The rate~\footnote{For the rate computation, we assume that $X(k)$ can be expressed as a
                  multi-set such that the uniform distribution on it
                  coincides with $\Pi_k$, which is true in the case that $\Pi_k$ is $D$-flat.}
                  $r_k$ of $\mathcal{C}_k$ satisfies  $r_k = r_1 \cdot \abs{X(1)}/\abs{X(k)}$ where $r_1$ is the relative
                  rate of $\mathcal{C}_1$.
    \item{[Linear]} The lifted code $\mathcal{C}_k$ is linear.
  \end{itemize}
\end{lemma}

\begin{proof}  
  We show that under our assumption on the $\gamma$-HDX $X(\le d)$ we
  obtain sufficient \textit{robustness} and \textit{tensorial}
  parameters to apply~\cref{theo:list_dec}. In this application, we
  will rely on parity sampling for robustness. If $\dsum_{X(k)}$ is a
  $(2\epsilon_0,2\epsilon)$-parity sampler, using the linearity of
  $\mathcal{C}_1$ we obtain a lifted code
  $\mathcal{C}_k=\dsum_{X(k)}(\mathcal{C}_1)$ which is linear and has bias
  $2\epsilon$; thus the lifting is indeed
  $(1/2-\epsilon_0,1/2-\epsilon)$-robust. If we want to fully rely on parity
  sampling in~\cref{theo:list_dec}, the lifting must be a
  $(\beta_0=1/2+\epsilon_0,\beta=2\epsilon)$-parity sampler, which is
  more stringent than the first parity sampling requirement.~\footnote{Recall that
  this strengthening is used in our list decoding framework.}
  To invoke~\cref{lemma:hdx_parity_samplers} and obtain this
  $(\beta_0,\beta)$-parity sampler, we need to choose a
  parameter $\theta$ (where $0 < \theta < (1
  - \beta_0)/\beta_0$) and \begin{align*} &
  k \ge \log_{(1+\theta) \beta_0}(\beta/3),\\ &
  d \ge \max\left(\frac{3 \cdot
  k^2}{\beta}, \frac{6}{\theta^2 \beta_0^2 \beta} \right), \text{
  and}\\ & \gamma = O\left(\frac{1}{d^2}\right).  \end{align*} To get a
  $(\mu,L)$-tensorial HDX,~\cref{theo:hdx_tensorial} requires
  $$
  L \ge \frac{c' \cdot 2^k \cdot k^5}{\mu^4} \quad \text{ and } \quad \gamma \le \frac{C' \cdot \mu^4}{k^{8+k} \cdot 2^{8k}}.
  $$
  where we used that our alphabet is binary (i.e., $q=2$) and $c',C' >
  0$ are constants. Finally,~\cref{theo:list_dec} requires $\mu \le
  \epsilon^8/2^{22}$. The conceptual part of the proof is essentially
  complete and we are left to compute parameters.
  Set $\zeta_0 = 3/4 + \epsilon_0 - \epsilon_0^2$. We choose $\theta = 1/2
  - \epsilon_0$ which makes $(1+\theta)\beta_0$ equal to $\zeta_0$
  (provided $\epsilon_0 < 1/2$ we have $\zeta_0 < 1$). This choice
  results in
  $$
  k \ge \lceil \log_{\zeta_0}(2\epsilon/3) \rceil \quad \text{ and } \quad d = O\left(\max\left(\frac{\log_{\zeta_0}(2\epsilon/3)}{\epsilon}, \frac{1}{(1/4 - \epsilon_0^2)^4 \cdot \epsilon} \right) \right).
  $$
  Combining the parity sampling and tensorial requirements and after
  some simplification, the expansion $\gamma$ is constrained as
  $$
  \gamma \le C''\cdot \min\left(\frac{\epsilon^{32}}{k^{8+k} \cdot 2^{8k}}, \frac{\epsilon^2}{k^4}, \left(1/4-\epsilon_0^2\right)^4 \cdot \epsilon^2 \right),
  $$
  where $C'' > 0$ is a constant. We deduce that taking $\gamma$ as
  $$
  \gamma \le C''\cdot \frac{\left(1/4-\epsilon_0^2\right)^4 \cdot \epsilon^{32}}{k^{8+k} \cdot 2^{8k}},
  $$
  is sufficient. Further simplifying the above bound gives
  $$
  \gamma = O\left(\frac{\left(1/4 - \epsilon_0^2\right)^4 \cdot \epsilon^{32}}{\left(\log_{\zeta_0}(2\epsilon/3)\right)^{8+\log_{\zeta_0}(2\epsilon/3)} \cdot \left(2\epsilon/3\right)^{8/\log(\zeta_0)}} \right).
  $$
  Now, we turn to the SOS-related parameter $L$ which is constrained
  to be
  $$
  L \ge c'' \cdot \frac{2^k\cdot k^5}{\epsilon^{32}},
  $$
  where $c'' > 0$. Note that in this case the exponent $O(L+k)$
  appearing in the running time of~\cref{theo:list_dec} becomes $O(L)$.
  Similarly, further simplification leads to
  $$
  L =  O\left(\frac{\left(\log_{\zeta_0}(2\epsilon/3)\right)^5 \cdot \left(3/2\epsilon\right)^{-1/\log(\zeta_0)}}{\epsilon^{32}}  \right).
  $$
  Taking $\epsilon_0$ to be a constant and simplifying yields the claimed parameters.
\end{proof}

\subsection{Instantiation to General Base Codes}

We can extend~\cref{lemma:direct_sum_linear_biased_codes_hdx} to an
arbitrary (not necessarily linear) binary base code $\mathcal{C}_1$ with
the natural caveat of no longer obtaining linear lifted code
$\mathcal{C}_k=\dsum_{X(k)}(\mathcal{C}_1)$. However, even if $\mathcal C_1$ has
small bias, it might not be the case that the difference of any two codewords
will have small bias, which is required for list decoding. To this end we modify
the code $\mathcal C_1$ by employing a projection $\phi$ which converts a
condition on the distance of the code to a condition on the bias of the difference
of any two codewords.

\begin{claim}\label{claim:zero_padding}
  If $\mathcal{C}_1$ is binary code on $[n]$ with relative distance
  $\delta$ and rate $r$, then there exists an explicit linear
  projection $\varphi \colon \mathbb{F}_2^n \to \mathbb{F}_2^n$
  such that the code $\mathcal{C}_1' = \varphi(\mathcal{C}_1)$
  has relative distance at least $\delta/2$ and rate $r$. Furthermore,
  for every $z,z' \in \mathcal{C}_1'$ we have
  $$
  \bias(z-z') \le 1 - \frac{\delta}{2}.
  $$
\end{claim}

\begin{proof}
  Take $\varphi$ to be the projector onto $\mathbb{F}_2^{n-s} \oplus \set{0}^s$
  where $s = \lfloor \delta n/2 \rfloor$. Then
  $$
  \mathcal{C}_1' \coloneqq \varphi(\mathcal{C}_1) = \set{(z_1,\dots, z_{n-s}, \underbrace{0,\dots,0}_{s}) ~\vert~ (z_1,\dots,z_n) \in \mathcal{C}_1},
  $$
  and the claim readily follows.
\end{proof}

With this modification in mind, we can now restate and prove~\cref{theo:direct_sum_general_codes_hdx}.
\begin{theorem}[Direct Sum Lifting on HDX (Restatement of~\cref{theo:direct_sum_general_codes_hdx})]\label{theo:direct_sum_general_codes_restated}
  Let $\epsilon_0 < 1/2$ be a constant and $\epsilon \in (0,\epsilon_0)$.  There
  exist universal constants $c,C > 0$ such that for any $\gamma$-HDX
  $X(\le d)$ on ground set $X(1)=[n]$ and $\Pi_1$ uniform, if
  $$
  \gamma \le (\log(1/\epsilon))^{-C \cdot \log(1/\epsilon)}\quad\text{ and }\quad d \ge c \cdot \frac{(\log(1/\epsilon))^2}{\epsilon},
  $$
  then the following holds:

  For every binary code $\mathcal{C}_1$ with $\Delta(\mathcal{C}_1) \ge 1/2-\epsilon_0$ on $X(1)=[n]$,
  there exists a binary lifted code $\mathcal{C}_k=\dsum_{X(k)}(\varphi(\mathcal{C}_1))$ 
  with $\Delta(\mathcal{C}_k) \ge 1/2-\epsilon^{\Omega_{\epsilon_0}(1)}$ on $X(k)$ where $k = O\left(\log(1/\epsilon)\right)$,
  $\varphi$ is an explicit linear projection, and
  \begin{itemize}
    \item{[Efficient List Decoding]} If $\tilde{y}$ is \lict{\sqrt{\epsilon}} to $\mathcal{C}_k$, then we
         can compute the list $\mathcal{L}(\tilde{y},\mathcal{C}_1,\mathcal{C}_k)$ (c.f.~\cref{def:coupled_code_list}) in time
         $$
         n^{\epsilon^{-O\left(1\right)}} \cdot f(n),
         $$
         where $f(n)$ is the running time of a unique decoding algorithm for $\mathcal{C}_1$.
    \item{[Rate]} The rate $r_k$ of $\mathcal{C}_k$ satisfies  $r_k = r_1 \cdot \abs{X(1)}/\abs{X(k)}$ where $r_1$ is the relative
         rate of $\mathcal{C}_1$.
    \item{[Linearity]} If $\mathcal{C}_1$ is linear, then $\varphi$ is the identity and $\mathcal{C}_k= \dsum_{X(k)}(\mathcal{C}_1)$
                       is linear.
  \end{itemize}
\end{theorem}

\begin{proof}   
  By virtue of~\cref{lemma:direct_sum_linear_biased_codes_hdx}, it is
  enough to consider when $\mathcal{C}_1$ is not linear. Note that in
  the proof of~\cref{lemma:direct_sum_linear_biased_codes_hdx} the
  only assumption about linearity of $\mathcal{C}_1$ we used to obtain
  $(1/2-\epsilon_0,1/2-\epsilon)$-robustness was that the sum of two
  codewords is in the code and hence it has small bias. For a general
  code $\mathcal{C}_1$ of constant distance $1/2-\epsilon_0$,
  applying~\cref{claim:zero_padding} we obtain a new code
  $\mathcal{C}_1'$ with this guarantee at the expense of a distance
  $1/2$ times the original one. Naturally, in the current proof we no
  longer obtain a linear lifted code
  $\mathcal{C}_k=\dsum_{X(k)}(\mathcal{C}_1')$. Excluding the two
  previous remarks the proof
  of~\cref{theo:direct_sum_general_codes_hdx} is now the same as the
  proof of ~\cref{lemma:direct_sum_linear_biased_codes_hdx}.
\end{proof}

\section{List Decoding Direct Product Codes}\label{sec:direct_prod}

\subsection{Direct Product Codes}
Having developed a decoding algorithm for direct sum, a promising
strategy for list decoding other lifted codes on expanding objects is
reducing them to instances of direct sum list decoding.  One such
reduction involves the direct product lifting, which was first studied
in the context of samplers by Alon et al.~in~\cite{ABNNR92}.  The direct
product lifting collects the entries of a code on each subset of size
$\ell$.

\begin{definition}[Direct Product Lifting]
Let $\mathcal C_1 \subseteq \F_2^n$ be a base code on $X(1) = [n]$.
The {\deffont direct product lifting} of a word $z \in \F_2^n$ on a collection $X(\ell)$ is $\dprod_{X(\ell)}(z) = (x_{\tee})_{\tee \in X(\ell)}$, where $x_{\tee} = (z_i)_{i \in \tee}$.
The direct product lifting of the entire code is $\dprod_{X(\ell)}(\mathcal C_1) = \{\dprod_{X(\ell)}(z) \mid z \in \mathcal C_1\}$, which is a code of length $|X(\ell)|$ over the alphabet $\F_2^\ell$.
\end{definition}

If $X$ is a HDX, its sampling properties ensure that the direct
product lifting has very high distance.  It follows from the
definition that if the bipartite graph between $X(1)$ and $X(k)$ is an
$(\eta, \delta)$-sampler and the code $\mathcal C_1$ has minimum
distance $\eta$, then the direct product lifting
$\dprod_{X(\ell)}(\mathcal C_1)$ has minimum distance at least $(1 -
\delta)$. Recalling from \cref{fact:exp_to_sampler} that the bipartite
graph between two levels of a HDX can be a sampler with arbitrarily
small parameters if the expansion is good enough, we can reasonably
hope to list decode the direct product lifting on a HDX up to a
distance close to 1.  In fact, Dinur et al.~\cite{DinurHKNT19}
provided a list decoding algorithm accomplishing exactly that.  We
offer a very different approach to the same list decoding problem.

\subsection{Direct Product List Decoding}
We will reduce direct product decoding on $X(\ell)$ to direct sum
decoding on $X(k)$, where $k \approx \ell/2$.  This requires
converting a received word $\tilde x \in (\F_2^{\ell})^{X(\ell)}$ to a
word $\tilde y \in \F_2^{X(k)}$ that we will decode using the direct
sum algorithm.  If we knew that $\tilde x = \dprod_{X(\ell)}(\tilde
z)$ for some $\tilde z \in \F_2^{X(1)}$, we would do so by simply
taking $\tilde y_{\ess} = \sum_{i \in \ess} \tilde z_i$ to be the
direct sum lifting on each edge $\ess$; that is, $\tilde y =
\dsum_{X(k)}(\tilde z)$.

Unfortunately, performing list decoding also involves dealing with
words $\tilde x$ that might not have arisen from the direct product
lifting.  To construct a corrupted instance of direct sum $\tilde y$
from $\tilde x$, we need to assign values to each face $\ess \in X(k)$
based only on the information we have on the faces $X(\ell)$, as there
is no word on the ground set to refer to.  Since different faces
$\tee, \tee' \in X(\ell)$ containing $\ess$ might not agree on $\ess$,
there could be ambiguity as to what value to assign for the sum on
$\ess$.

This is where the $D$-flatness of the distribution $\Pi_{\ell}$ (which holds for the $\gamma$-HDX construction described in \cref{lemma:hdx_existence}) comes in.
Recall that to obtain codewords in $\dprod_{X(\ell)}(\mathcal C_1)$ without weights on their entries, we duplicate each face $\tee \in X(\ell)$ at most $D$ times to make the distribution $\Pi_{\ell}$ uniform.
To perform the same kind of duplication on $X(k)$ that makes $\Pi_k$ uniform, note that each face $\ess \in X(k)$ has $\Pi_k(\ess)$ proportional to $\abs{\{\tee \in X(\ell) \mid \tee \supset \ess\}}$ (where $X(\ell)$ is thought of as a multiset), so we will create one copy of $\ess$ for each $\tee$ containing it.
Thus we can assign a unique $\tee \supset \ess$ to each copy.
By downward closure, the distribution on $X(\ell)$ obtained by choosing $\ess$ uniformly from the multiset $X(k)$ and then selecting its associated face $\tee$ will be uniform, just like $\Pi_{\ell}$.
With this careful duplication process, we are ready to define the function $\rho_k$ that takes a corrupted direct product word $\tilde x$ to a corrupted direct sum word $\tilde y$.

\begin{definition}[Reduction Function]\label{def:reduction_function}
Let $k < \ell$ and $X$ be a HDX where the distribution $\Pi_{\ell}$ is $D$-flat.
Duplicate faces in $X(k)$ so that $\Pi_k$ is uniform, and assign a face $\tee_{\ess} \in X(\ell)$ to each $\ess \in X(k)$ (after duplication) such that $\tee_{\ess}$ is distributed according to $\Pi_{\ell}$ when $\ess$ is selected uniformly from $X(k)$.
The function $\rho_k: (\F_2^{\ell})^{X(\ell)} \to \F_2^{X(k)}$ is defined as
	$$(\rho_k(\tilde x))_{\ess} = \sum_{i \in \ess} (\tilde x_{\tee_{\ess}})_i.$$
\end{definition}

The reduction function $\rho_k$ resolves the ambiguity of which face
$\tee \supset \ess$ to sample the sum from by assigning a different
face to each copy of $\ess$ in a manner compatible with the
distribution $\Pi_{\ell}$.  Observe that if $\tilde x =
\dprod_{X(\ell)}(\tilde z)$ for some $\tilde z \in \F_2^{X(1)}$, then
$\rho_k(\tilde x) = \dsum_{X(k)}(\tilde z)$.

The following lemma shows that performing this reduction from direct
product to direct sum maintains agreement between words.  It
essentially says that if a received word $\tilde x$ exhibits some
agreement with $x \in \dprod_{X(\ell)}(\mathcal C_1)$, then there is a
$k$ for which $\rho_k(\tilde x)$ and $\rho_k(x)$ have agreement larger
than 1/2.

\begin{lemma}[Product-to-sum agreement]\label{lemma:product_to_sum}
Fix $\epsilon > 0$ and $C' > 2$.
Let $z \in \mathcal C_1$, $x = \dprod_{X(\ell)}(z)$, and $\tilde x \in (\F_2^{\ell})^{X(\ell)}$.
If $\Delta(x, \tilde x) \leq 1 - \epsilon$, then there exists a $k$ satisfying
	$$\abs{k - \ell/2} < \frac 12 \sqrt{C' \ell \log(1/\epsilon)}$$
such that
	$$\Delta(y, \tilde y) \leq 1/2 - \epsilon/2 + \epsilon^{C'/2},$$
where $y = \rho_k(x)$ and $\tilde y = \rho_k(\tilde x)$ are words in $\F_2^{X(k)}$.
\end{lemma}

\begin{proof}
For $t \in X(\ell)$ and $\ess \subseteq \tee$, define the function $\chi_{\ess, \tee}: \F_2^{\tee} \to \{-1, 1\}$ by
	$$\chi_{\ess, \tee}(w) = \prod_{i \in \ess} (-1)^{w_i}.$$
For each face $\tee \in X(\ell)$, consider the expectation
	$\E_{\ess \subseteq \tee} [\chi_{\ess, \tee}(x_{\tee} - \tilde x_{\tee})]$,
where $\ess$ is a subset of $\tee$ of any size chosen uniformly.
If $x_{\tee} = \tilde x_{\tee}$, which happens for at least $\epsilon$ fraction of faces $\tee$, the expression in the expectation is always 1.
Otherwise, this expectation is zero, so taking the expectation over the faces yields
	$$\E_{\tee \sim \Pi_{\ell}} \E_{\ess \subseteq \tee}[\chi_{\ess, \tee}(x_{\tee} - \tilde x_{\tee})] = \ProbOp_{t \sim \Pi_{\ell}}[x_{\tee} = \tilde x_{\tee}] \geq \epsilon.$$
We would like to restrict to a fixed size of faces $\ess$ for which this inequality holds; as this will be the size of the direct sum faces, we need to make sure it's large enough to give us the expansion required for decoding later.
Using a Chernoff bound (\cref{fact:Chernoff_bound} with $a = \sqrt{C' \ell \log(1/\epsilon)}$), we see that the size of the faces is highly concentrated around $\ell/2$:
	$$\ProbOp_{\ess \subseteq \tee}\left[\abs{\abs{\ess} - \frac{\ell}{2}} \geq \frac 12 \sqrt{C' \ell \log(1/\epsilon)}\right] \leq 2 e^{-C' \log(1/\epsilon) / 2} \leq 2 \epsilon^{C'/2}.$$

Let $I$ be the interval
	$$I = \left(\frac{\ell}{2} - \frac 12 \sqrt{C' \ell \log(1/\epsilon)}, \frac{\ell}{2} + \frac 12 \sqrt{C' \ell \log(1/\epsilon)} \right).$$

The expectation inequality becomes
\begin{align*}
	\epsilon &\leq \E_{\tee \sim \Pi_{\ell}} [\E_{\ess \subseteq \tee} [\One_{\abs{\ess} \in I} \cdot \chi_{\ess, \tee}(x_{\tee} - \tilde x_{\tee})] + \E_{\ess \subseteq \tee}[\One_{\abs{\ess} \notin I} \cdot \chi_{\ess, \tee}(x_{\tee} - \tilde x_{\tee})]] \\
	&\leq \E_{\tee \sim \Pi_{\ell}} \E_{\ess \subseteq \tee, \abs{\ess} \in I} [\chi_{\ess, \tee}(x_{\tee} - \tilde x_{\tee})] + 2\epsilon^{C'/2}.
\end{align*}

Thus there exists a $k \in I$ such that
	$$\epsilon - 2\epsilon^{C'/2} \leq \E_{\tee \sim \Pi_{\ell}} \E_{\ess \subseteq \tee, \abs{\ess} = k} [\chi_{\ess, \tee}(x_{\tee} - \tilde x_{\tee})].$$
	
Choosing a face $\tee$ and then a uniformly random $\ess \subseteq \tee$ of size $k$ results in choosing $\ess$ according to $\Pi_k$.
Moreover, the edge $\tee_{\ess}$ containing $\ess$ from \cref{def:reduction_function} is distributed according to $\Pi_{\ell}$.
Bearing in mind the definitions of $y_{\ess}$ and $\tilde y_{\ess}$, we have
\begin{align*}
	\epsilon - 2\epsilon^{C'/2} &\leq \E_{\tee \sim \Pi_{\ell}} \E_{\ess \subseteq \tee, \abs{\ess} = k} [\chi_{\ess, \tee}(x_{\tee} - \tilde x_{\tee})] \\
	&= \E_{\ess \sim \Pi_k}[\chi_{\ess, \tee_{\ess}}(x_{\tee_{\ess}} - \tilde x_{\tee_{\ess}})] \\
	&= \E_{\ess \sim \Pi_k}[(-1)^{(\rho_k(x))_{\ess} - (\rho_k(\tilde x))_{\ess}}] \\
	&= \bias(y - \tilde y)
\end{align*}
which translates to a Hamming distance of $\Delta(y, \tilde y) \leq 1/2 - \epsilon/2 + \epsilon^{C'/2}$.
\end{proof}

With \cref{lemma:product_to_sum} in hand to reduce a direct product
list decoding instance to a direct sum list decoding instance, we can
decode by using a direct sum list decoding algorithm as a black box.

\begin{algorithm}{Direct Product List Decoding Algorithm}
{A word $\tilde x \in (\F_2^{\ell})^{X(\ell)}$ with distance at most $(1-\epsilon)$ from $\dprod_{X(\ell)}(\mathcal C_1)$}
{The list $\mathcal L' = \{z \in \F_2^n \mid \Delta(\dprod_{X(\ell)}(z), \tilde x) \leq 1 - \epsilon\}$}\label{algo:product_list_decoding}
\begin{enumerate}
	\item Let $I$ be the interval $\left(\ell/2 - \sqrt{C' \ell \log(1/\epsilon)}/2, \ell/2 + \sqrt{C' \ell \log(1/\epsilon)}/2 \right)$.
	\item For each integer $k \in I$, run the direct sum list decoding algorithm on the input $\tilde y = \rho_k(\tilde x) \in \F_2^{X(k)}$ to obtain a coupled list $\mathcal L_k$ of all pairs $(z, y)$ with $\Delta(y, \tilde y) \leq 1/2 - \epsilon/2 + \epsilon^{C'/2}$. 
	\item Let $\mathcal L = \cup_{k \in I} \{z \in \mathcal C_1 \mid (z, y) \in \mathcal L_k\}$.
	\item Let $\mathcal L' = \{z \in \mathcal L \mid \Delta(\dprod_{X(\ell)}(z), \tilde x) \leq 1 - \epsilon\}$.
	\item Output $\mathcal L'$.
\end{enumerate}
\end{algorithm}

\begin{theorem}[Product-to-sum Reduction]\label{theo:product_to_sum_reduction}
Let $\epsilon > 0$ and $C' > 2$.
Let $\dprod_{X(\ell)}(\mathcal C_1)$ be the direct product lifting of a base code $\mathcal C_1$ on a simplicial complex $X$.
If the direct sum lifting $\dsum_{X(k)}(\mathcal C_1)$ is list decodable up to distance $\left(1/2 - \epsilon/2 + \epsilon^{C'/2}\right)$ in time $\tilde f(n)$ for all $k$ satisfying $\abs{k - \ell/2} < \sqrt{C' \ell \log(1/\epsilon)}/2$, then \cref{algo:product_list_decoding} list decodes $\dprod_{X(\ell)}(\mathcal C_1)$ up to distance $(1 - \epsilon)$ in running time
$$\sqrt{C' \ell \log(1/\epsilon)} \tilde f(n) + |X(\ell)| |\mathcal L|.$$
\end{theorem}

\begin{proof}
Let $\tilde x \in (\F_2^{\ell})^{X(\ell)}$ be a received word and let
$z \in \mathcal C_1$ satisfy $\Delta(\dprod_{X(\ell)}(z), \tilde x) \leq
1 - \epsilon$.  By \cref{lemma:product_to_sum}, there exists a $k \in
I$ such that $\Delta(y, \tilde y) \leq 1/2 - \epsilon/2 +
\epsilon^{C'/2}$.  Thanks to this distance guarantee, the pair $(z,y)$
will appear on the list $\mathcal L_k$ when the direct sum list
decoding algorithm is run for this $k$.  Then $z$ will be on the
combined list $\mathcal L$ and the trimmed list $\mathcal L'$, with
the trimming ensuring that no elements of $\mathcal C_1$ appear on
this list beyond those with the promised distance.  The set
$\{\dprod_{X(\ell)}(z) \mid z \in \mathcal L'\}$ thus contains all words
in $\dprod_{X(\ell)}(\mathcal C_1)$ with distance at most $(1-\epsilon)$
from $\tilde x$.

To obtain the promised the running time, note that
\cref{algo:product_list_decoding} runs the direct sum list decoding
algorithm $\sqrt{C' \ell \log(1/\epsilon)}$ times and then computes the
direct product lifting of each element of $\mathcal L$ in the trimming
step.
\end{proof}

Combining the parameters in the reduction with those required for our
direct sum list decoding algorithm, we obtain the following.  Note
that for very small values of $\epsilon$, we can choose the constant
$C'$ to be close to 2, and we will be list decoding the direct sum
code up to distance $1/2 - \sqrt{\beta} \approx 1/2 - \epsilon/4$.

\begin{corollary}[Direct Product List Decoding]\label{cor:product_list_decoding}
Let $\epsilon_0 < 1/2$ be a constant, and let $\epsilon > 0$, $C' \geq 2 + 4/\log(1/\epsilon)$, and $\beta = (\epsilon/2 - \epsilon^{C'/2})^2$.
There exist universal constants $c, C > 0$ such that for any $\gamma$-HDX $X(\leq d)$ on ground set $[n]$ and $\Pi_1$ uniform, if
	$$\gamma \leq \log(1/\beta)^{-C \log(1/\beta)} \quad \text{ and } \quad d \geq c \cdot \frac{\log(1/\beta)^2}{\beta},$$
then the following holds:

For every binary code $\mathcal C_1$ with $\Delta(\mathcal C_1) \geq 1/2 - \epsilon_0$ on $X(1) = [n]$, there exists a lifted code $\mathcal C_{\ell} = \dprod_{X(\ell)}(\phi(\mathcal C_1))$ on $C_{\ell}$ where $\ell = O(\log(1/\beta))$, $\varphi$ is an explicit linear projection, and

\begin{itemize}
	\item{[Efficient List Decoding]} If $\tilde x$ is $(1-\epsilon)$-close to $\mathcal C_{\ell}$, then we can compute the list of all codewords of $\mathcal C_{\ell}$ that are $(1-\epsilon)$-close to $\tilde x$ in time $n^{\epsilon^{-O(1)}} \cdot f(n)$, where $f(n)$ is the running time of the unique decoding algorithm for $\mathcal C_1$.
	\item{[Rate]} The rate $r_{\ell}$ of $\mathcal C_{\ell}$ satisfies $r_{\ell} = r_1 \cdot \abs{X(1)}/(\ell \abs{X(\ell)})$, where $r_1$ is the relative rate of $\mathcal C_1$.
	\item{[Linearity]} If $\mathcal C_1$ is linear, then $\phi$ is the identity and $\mathcal C_{\ell}$ is linear.
\end{itemize}
\end{corollary}

\begin{proof}
Let $k = \ell/2 - \sqrt{C' \ell \log(1/\epsilon)}/2$.
The choice of parameters ensures that $\dsum_{X(k)}(\mathcal C_1)$ is list
decodable up to distance $1/2 - \sqrt \beta = 1/2 - \epsilon/2 +
\epsilon^{C'/2}$ in running time
$g(n) = n^{\beta^{-O(1)}} f(n)$ by
\cref{theo:direct_sum_general_codes_hdx} (noting that the bound on $C'$ implies
    $\beta \geq \epsilon^2/16$).
Since increasing $k$
increases the list decoding radius of the direct sum lifting, this holds for any
value of $k$ with $\abs{k - \ell/2} \leq \sqrt{C' \ell
  \log(1/\epsilon)}/2$.  By \cref{theo:product_to_sum_reduction}, the
direct product lifting $\dprod_{X(\ell)}(\mathcal C_1)$ is list decodable
up to distance $(1-\epsilon)$ in running time
	$$\sqrt{C' \ell \log(1/\epsilon)} n^{\beta^{-O(1)}} f(n) + \abs{X(\ell)} \abs{\mathcal L}.$$
	
The HDX has $\abs{X(\ell)} \leq \binom{n}{\ell} = n^{O(\log(1/\beta))}$, and the list size
$\abs{\mathcal L}$ is bounded by the sum of the sizes of the lists
$\mathcal L_k$ obtained from each direct sum decoding.  Each of these
lists has $\abs{\mathcal L_k} \leq 1/(2\beta)$ by the Johnson
bound (see \cref{rem:Johnson_list_size}) and the number of lists is
constant with respect to $n$, so the overall running time is dominated
by the first term, $n^{\beta^{-O(1)}} f(n) = n^{\epsilon^{-O(1)}}
f(n)$.

The rate and linearity guarantees follow in the same manner as they do
in \cref{theo:direct_sum_general_codes_hdx}, where the rate
calculation requires a slight modification for dealing with the
increased alphabet size and $\phi$ is the projection from
\cref{claim:zero_padding}.
\end{proof}

Using \cref{cor:product_list_decoding} with HDXs obtained from
Ramanujan complexes as in \cref{cor:ramanujan_instantiation}, we can
perform list decoding with an explicit construction up to distance $(1
- \epsilon)$ with HDX parameters $d = O(\log(1/\epsilon)^2/\epsilon^2)$
and $\gamma = (\log(1/\epsilon))^{-O(\log(1/\epsilon))}$.
The direct product list decoding algorithm of Dinur
et al.~\cite{DinurHKNT19} is based on a more general expanding object
known as a double sampler. As the only known double sampler
construction is based on a HDX, we can compare our parameters
to their HDX requirements of $d = O(\exp(1/\epsilon))$ and
$\gamma = O(\exp(-1/\epsilon))$.

\section{Instantiation II: Direct Sum on Expander Walks}\label{sec:expander_walks}

We instantiate the list decoding framework to the direct sum lifting
where the sum is taken over the collection $X(k)$ of length $k$ walks
of a sufficiently expanding graph $G$. To stress the different nature
of this collection and its dependence on $G$ we equivalently denote
$X(k)$ by $W_G(k)$ and endow it with a natural measure
in~\cref{def:walk_collection}.
\begin{definition}[Walk Collection]\label{def:walk_collection}
  Let $G=(V,E,w)$ be a weighted graph with weight distribution $w
  \colon E \to [0,1]$. For $k \in \mathbb{N}^{+}$, we denote by
  $W_G(k)$ the collection of all walks of length $k$ in $G$,
  i.e.,
  $$
  W_G(k) \coloneqq \set{w = (w_1,\dots,w_k)~\vert~w\textup{ is a walk of length $k$ in $G$}}.
  $$
  We endow $W_G(k)$ with the distribution $\Pi_k$
  arising from taking a random vertex $w_1$ according to the
  stationary distribution on $V$ and then taking $k-1$ steps according
  to the normalized random walk operator of $G$.
\end{definition}

One simple difference with respect to the HDX case is that now we are
working with a collection of (ordered) tuples instead of subsets. The
Propagation Rounding~\cref{algo:prop-rd} remains the same, but we need
to establish the tensorial properties of $W_G(k)$ which is done
in~\cref{subsec:expander_walks_two_step_tensorial}.

The main result of this section follows.
\begin{theorem}[Direct Sum Lifting on Expander Walks]\label{theo:direct_sum_general_codes_expander_walks}
  Let $\epsilon_0 < 1/2$ be a constant and $\epsilon \in (0,\epsilon_0)$.
  There exists a universal constant $C > 0$ such that for any $d$-regular
  $\gamma$-two-sided expander graph $G$ on ground set $W_G(1)=[n]$, if $\gamma \le \epsilon^{C}$,
  then the following holds:  

  For every binary code $\mathcal{C}_1$ with $\Delta(\mathcal{C}_1) \ge 1/2-\epsilon_0$ on $W_G(1)=[n]$,
  there exists a binary lifted code $\mathcal{C}_k=\dsum_{X(k)}(\varphi(\mathcal{C}_1))$ 
  with $\Delta(\mathcal{C}_k) \ge 1/2-\epsilon^{\Omega_{\epsilon_0}(1)}$ on $W_G(k)$ where $k = O\left(\log(1/\epsilon)\right)$,
  $\varphi$ is an explicit linear projection, and
  \begin{itemize}
    \item{[Efficient List Decoding]} If $\tilde{y}$ is \lict{\sqrt{\epsilon}} to $\mathcal{C}_k$, then we
         can compute the  list $\mathcal{L}(\tilde{y},\mathcal{C}_1,\mathcal{C}_k)$ (c.f.~\cref{def:coupled_code_list}) in time
         $$
         n^{\epsilon^{-O\left(1\right)}} \cdot f(n),
         $$
         where $f(n)$ is the running time of a unique decoding algorithm for $\mathcal{C}_1$.
    \item{[Rate]} The rate $r_k$ of $\mathcal{C}_k$ satisfies  $r_k = r_1 / d^{k-1}$ where $r_1$ is the relative
         rate of $\mathcal{C}_1$.
    \item{[Linearity]} If $\mathcal{C}_1$ is linear, then $\varphi$ is the identity and $\mathcal{C}_k= \dsum_{X(k)}(\mathcal{C}_1)$
                       is linear.
  \end{itemize}
\end{theorem}

In particular, we
apply~\cref{theo:direct_sum_general_codes_expander_walks} to the
explicit family of Ramanujan expanders of Lubotzky et
al. from~\cref{theo:ramanujan_existence}.

\begin{theorem}[Lubotzky-Phillips-Sarnak abridged~\cite{LPS88}]\label{theo:ramanujan_existence}
   Let $p \equiv 1 \pmod{4}$ be a prime. Then there exists an explicit
   infinite family of $(p+1)$-regular Ramanujan graphs $G_1,G_2,\dots$
   on $n_1 < n_2 < \cdots$ vertices, i.e., $\sigma_2(G_i) \le
   2 \cdot \sqrt{p}/(p+1)$.
\end{theorem}

In order to construct Ramanujan expanders with arbitrarily good
expansion, we will use the following lemma for finding primes.

\begin{lemma}[From~\cite{Ta-Shma17}]\label{lemma:prime_finder}
  For every $\alpha > 0$ and sufficiently large $n$, there exists an
  algorithm that given $a$ and $m$ relatively prime, runs in time
  $\poly(n)$ and outputs a prime number $p$ with $p \equiv a \pmod m$
  in the interval $[(1-\alpha)n, n]$.
\end{lemma}

This results in~\cref{cor:ramanujan_instantiation_exp}.
\begin{corollary}\label{cor:ramanujan_instantiation_exp}
  Let $\epsilon_0 < 1/2$ be a constant and $\epsilon \in (0,\epsilon_0)$. There
  is an infinite sequence of explict Ramanujan expanders
  $G_1,G_2,\dots$ on ground sets of size $n_1 < n_2 < \cdots$ such that the
  following holds:
  
  For every sequence of binary codes $\mathcal{C}_1^{(i)}$ on $[n_i]$ with
  rate and distance uniformly bounded by $r_1^{(i)}$ and $(1/2-\epsilon_0)$
  respectively, there exists a sequence of binary lifted codes
  $\mathcal{C}_k^{(i)} =\dsum_{X(k)}(\varphi(\mathcal{C}_1^{(i)}))$ on a collection $X_i(k)$
  with distance $\left(1/2-\epsilon^{\Omega_{\epsilon_0}(1)} \right)$ where $\varphi$ is an explicit linear projection and
  \begin{itemize}
    \item{[Efficient List Decoding]} If $\tilde{y}$ is \lict{\sqrt{\epsilon}} to $\mathcal{C}_k$, then we can compute
          the list $\mathcal{L}(\tilde{y},\mathcal{C}_1,\mathcal{C}_k)$ (c.f.~\cref{def:coupled_code_list}) in time $n^{\epsilon^{-O\left(1\right)}} \cdot f(n)$, where
          $f(n)$ is the running time of a unique decoding algorithm of $\mathcal{C}_1$.
    \item{[Explicit Construction]} The collection $W_{G_i}(k)$ is obtained from length $k$ walks on a Ramanujan
          $d$-regular expander $G_i$ where $k = O\left(\log(1/\epsilon)\right)$, $d=8 \cdot \epsilon^{-O(1)}$ and $\gamma = \epsilon^{O(1)}$.
    \item{[Rate]} The rate $r_k^{(i)}$ of $\mathcal{C}_k^{(i)}$ satisfies $r_k^{(i)} \ge r_1^{(i)} \cdot \epsilon^{O(\log(1/\epsilon))}$.
    \item{[Linearity]} If $\mathcal{C}_1^{(i)}$ is linear, then $\varphi$ is the identity and $\mathcal{C}_k^{(i)}= \dsum_{X(k)}(\mathcal{C}_1^{(i)})$
                       is linear.    
  \end{itemize}
\end{corollary}

\begin{proof}
   Using \cref{lemma:prime_finder} with $a=1$ and $m=4$, we see that
   given $n,\alpha$, a prime $p$ such that $p \equiv 1 \pmod{4}$ may
   be found in the interval $[(1-\alpha )n,n]$ for large enough
   $n$. For Ramanujan expanders, the condition that
   $ \gamma \leq \epsilon^C$ translates to $p\geq
   4 \cdot \epsilon^{-2C}$. Choose $\alpha = 1/2$ and $n >
   8 \cdot \epsilon^{-2C}$ so that we find a prime greater than
   $4 \cdot \epsilon^{-2C}$, but at most $8 \cdot \epsilon^{-2C}$.

   Based on this prime, we use the
   above~\cref{theo:ramanujan_existence} to get a family of Ramanujan
   graphs $G_1,G_2,\dots $ with $n_1 < n_2 < \cdots $ vertices, such
   that the degree is bounded by $8\epsilon^{-2C}$.
   Using the parameters of this family in \cref{theo:direct_sum_general_codes_expander_walks}, we obtain the desired claims.
\end{proof}

\subsection{Expander Walks are Two-Step Tensorial}\label{subsec:expander_walks_two_step_tensorial}

To apply the list decoding framework we need to establish the
tensorial parameters of expander walks $W_G(k)$ for a
$\gamma$-two-sided expander graph $G$. Although the tensorial property
is precisely what the abstract list decoding framework uses, when
faced with a concrete object such as $W_G(k)$ it will be easier to
prove that it satisfies a \textit{splittable} property defined
in~\cite{AJT19} for complexes which implies the tensorial property. In
turn, this splittable property is defined in terms of some natural
operators denoted \textit{Swap} operators whose definition is recalled
in~\cref{subsubsec:expander_walk_swap_operator} in a manner tailored
to the present case $X(k) = W_G(k)$. Then,
in~\cref{subsubsec:expander_walks_swap_op_splittable}, we formally
define the splittable property and show that the expansion of the Swap
operator is controlled by the expansion parameter $\gamma$ of $G$
allowing us to deduce the splittable parameters of $W_G(k)$.  Finally,
in~\cref{subsec:expander_walk_two_step_tensorial}, we show how
$W_G(k)$ being splittable gives the tensorial parameters. Some results
are quite similar to the hypergraph case in~\cite{AJT19} (which built
on~\cite{BarakRS11}). The key contribution in this new case of
$W_G(k)$ is observing the existence of these new Swap operators along
with their expansion properties.

\subsubsection{Emergence of Swap Operators}\label{subsubsec:emmergence_of_swap}

To motivate the study of Swap operators on $W_G(k)$, we show how they
naturally emerge from the study of $k$-CSPs. The treatment is quite
similar to the hypergraph case developed in~\cite{AJT19}, but this
will give us the opportunity to formalize the details that are
specific to $W_G(k)$. Suppose that we solve a $k$-CSP instance as
defined in~\cref{subsec:sos-relax} whose constraints were placed on
the tuples corresponding to walks in $W_G(k)$. The result is a local
PSD ensemble $\set{\rv Z}$ which can then be fed to the Propagation
Rounding~\cref{algo:prop-rd}.  It is easy to show that the tensorial
condition of~\cref{eq:tensorial_exp_walk_example} (below) is
sufficient to guarantee an approximation to this $k$-CSP on $W_G(k)$
within $\mu$ additive error. The precise parameters are given
in~\cref{subsub:k_csp_walk_constraints}. For now, we take this
observation for granted and use it to show how the Swap operators
emerge in obtaining the inequality
\begin{equation}\label{eq:tensorial_exp_walk_example}
   \ExpOp_{\Omega} \ExpOp_{w \sim W_G(k)}{ \norm{\set{\rv Z'_{w}} - \set*{\rv Z'_{w_1}}\cdots \set*{\rv Z'_{w_k}}}_1 } \le \mu
\end{equation}
present in the definition of tensoriality.

The following piece of notation will be convenient when referring to
sub-walks of a given walk.

\begin{definition}[Sub-Walk]
  Given $1 \le i \le j \le k$ and $w=(w_1,\dots,w_k) \in W_G(k)$, we
  define the sub-walk $w(i,j)$ from $w_i$ to $w_j$ as
  $$
  w(i,j) \coloneqq (w_i,w_{i+1},\dots,w_j).
  $$
\end{definition}

We will need the following simple observation about marginal
distributions of $\Pi_k$ on sub-walks.

\begin{claim}[Marginals of the walk distribution]\label{claim:walk_marginals}
  Let $k \in \mathbb{N}^+$ and $1 \le i \le j \le k$. Then sampling $w \sim \Pi_k$ in $W_G(k)$ and taking $w(i,j)$
  induces the distribution $\Pi_{j-i+1}$ on $W_G(j-i+1)$.
\end{claim}

\begin{proof}
  Let $w=(w_1,\dots,w_i,\dots,w_j,\dots,w_k) \sim \Pi_k$. Since
  $w_1 \sim \Pi_1$ where $\Pi_1$ is the stationary measure of $G$ and
  $w_2,\dots,w_i$ are obtained by $(i-1)$ successive steps of a random
  walk on $G$, the marginal distribution on $w_i$ is again the
  stationary measure $\Pi_1$.  Then by taking $(j-i)$ successive
  random walk steps from $w_i$ on $G$, we obtain a walk
  $(w_i,\dots,w_j)$ distributed according to $\Pi_{j-i+1}$.
\end{proof}

We also need the notion of a \textit{splitting tree} as follows.
\begin{definition}[Splitting Tree~\cite{AJT19}]
  We say that a binary tree $\tree$ is a $k$-splitting tree if it has
  exactly $k$ leaves and
  \begin{itemize}
    \item the root of $\tree$ is labeled with $k$ and all other vertices are labeled with
          positive integers,
    \item the leaves are labeled with $1$, and
    \item each non-leaf vertex satisfies the property that its label is the sum of
          the labels of its two children.
  \end{itemize}
\end{definition}

The Swap operators arise naturally from the following triangle
inequality where the quantity $\ExpOp_{w \sim W_G(k)}{ \norm{\set{\rv
Z'_{w}} - \prod_{i=1}^k\set*{\rv Z'_{w(i)}}}_1}$ is upper bounded by a
sum of terms of the form
$$
\ExpOp_{w \sim W_G(k_1+k_2)}{ \norm{\set{\rv Z'_{w}} - \set*{\rv Z'_{w(1,k_1)}}\set*{\rv Z'_{w(k_1+1,k_2)}} }_1}.
$$
We view the above expectation as taking place over the edges
$W_G(k_1+k_2)$ of a bipartite graph on vertex bipartition
$(W_G(k_1),W_G(k_2))$. This graph gives rise to a Swap operator which
we formally define later
in~\cref{subsubsec:expander_walk_swap_operator}. The following claim
shows how a splitting tree defines all terms (and hence also their
corresponding graphs and operators) that can appear in this upper
bound.

\begin{claim}[Triangle inequality]
  Let $k \in \mathbb{N}^+$ and $\tree$ be a $k$-splitting tree. Then
  $$
   \ExpOp_{w \sim W_G(k)}{ \norm{\set{\rv Z'_{w}} - \prod_{i=1}^k\set*{\rv Z'_{w(i)}}}_1} \le \sum_{(k_1,k_2)} \ExpOp_{w \sim W_G(k_1+k_2)}{ \norm{\set{\rv Z'_{w}} - \set*{\rv Z'_{w(1,k_1)}}\set*{\rv Z'_{w(k_1+1,k_2)}} }_1},
  $$
  where the sum $\sum_{(k_1,k_2)}$ is taken over all pairs of
  labels of the two children of each internal node of $\tree$.
\end{claim}

\begin{proof}
  We prove the claim by induction on $k$. Let $(k_1,k_2)$ be the
  labels of the children of the root of the splitting tree
  $\tree$. Suppose $\tree_1$ and $\tree_2$ are the corresponding
  splitting trees rooted at these children with labels $k_1$ and
  $k_2$, respectively. By this choice, we have $k=k_1+k_2$. Applying
  the triangle inequality yields  
  \begin{align*}
    \ExpOp_{w \sim W_G(k)}{ \norm{\set{\rv Z'_{w}} - \prod_{i=1}^k\set*{\rv Z'_{w_i}}}_1}  \le &
             \ExpOp_{w \sim W_G(k)}{ \norm{\set{\rv Z'_{w}} - \set*{\rv Z'_{w(1,k_1)}}\set*{\rv Z'_{w(k_1+1,k_2)}} }_1} + \\
          &  \ExpOp_{w \sim W_G(k)}{ \norm{\set*{\rv Z'_{w(1,k_1)}}\set*{\rv Z'_{w(k_1+1,k_2)}} - \prod_{i=1}^{k_1}\set*{\rv Z'_{w_i}} \set*{\rv Z'_{w(k_1+1,k_2)}}}_1} + \\
          &  \ExpOp_{w \sim W_G(k)}{ \norm{\prod_{i=1}^{k_1}\set*{\rv Z'_{w_i}} \set*{\rv Z'_{w(k_1+1,k_2)}} - \prod_{i=1}^k\set*{\rv Z'_{w_i}} }_1}.
  \end{align*}
  Using the marginalization given by~\cref{claim:walk_marginals} on the second and third terms and simplifying, we get
 \begin{align*}
    \ExpOp_{w \sim W_G(k)}{ \norm{\set{\rv Z'_{w}} - \prod_{i=1}^k\set*{\rv Z'_{w_i}}}_1} \le &\ExpOp_{w \sim W_G(k)}{ \norm{\set{\rv Z'_{w}} - \set*{\rv Z'_{w(1,k_1)}}\set*{\rv Z'_{w(k_1+1,k_2)}} }_1} +\\
                                                                                             &\ExpOp_{w \sim W_G(k_1)}{ \norm{\set*{\rv Z'_{w}} - \prod_{i=1}^{k_1}\set*{\rv Z'_{w_i}}}_1} +
                                                                                             \ExpOp_{w \sim W_G(k_2)}{ \norm{\set*{\rv Z'_{w}} - \prod_{i=1}^{k_2}\set*{\rv Z'_{w_i}}}_1} .
  \end{align*}
  Applying the induction hypothesis to the second term with tree
  $\tree_1$ and to the third term with tree $\tree_2$ finishes the
  proof.
\end{proof}

\subsubsection{Swap Operators Arising from Expander Walks}\label{subsubsec:expander_walk_swap_operator}

We define the Swap operator associated to walks on a given graph $G$
as follows.

\begin{definition}[Graph Walk Swap Operator]\label{def:graph_walk_swap}
  Let $G=(V,E,w)$ be a weighted graph. Let $k_1,k_2 \in
  \mathbb{N}^{+}$ be such that $k = k_1 + k_2$. We define the
  graph walk Swap operator
  $$
  \wswap{k_1}{k_2} \colon \mathbb{R}^{W_G(k_2)} \to \mathbb{R}^{W_G(k_1)}
  $$
  such that for every $f \in \mathbb{R}^{W_G(k_2)}$,
  \begin{align*}
      \left(\wswap{k_1}{k_2}(f)\right)(w) \coloneqq \E_{w': ww' \in W(k)}[f(w')],
  \end{align*}
  where $ww'$ denotes the concatenation of the walks
  $w$ and $w'$.  The operator $\wswap{k_1}{k_2}$ can be
  defined more concretely in matrix form such that for every $w
  \in W_G(k_1)$ and $w' \in W_G(k_2)$, 
  \begin{align*}
      \left(\wswap{k_1}{k_2}\right)_{w,w'} \coloneqq \frac{\Pi_k(ww')}{\Pi_{k_1}(w)}.
  \end{align*}  
\end{definition}

\begin{remark}
  Swap operators are Markov operators, so the largest singular value of a Swap operator is bounded by $1$.
\end{remark}

Unlike the Swap operators for HDXs described in~\cite{AJT19}, which are defined using unordered subsets of hyperedges, the Swap operators $\wswap{k_1}{k_2}$ use sub-walks and are thus directed operators.
Instead of analyzing such an operator directly, we will examine the symmetrized version
  $$
  \mathcal{U}(\wswap{k_1}{k_2}) = 
  \begin{pmatrix}
    0 & \wswap{k_1}{k_2}\\
    \left(\wswap{k_1}{k_2}\right)^{\dag} & 0
  \end{pmatrix}
  $$
and show that $\mathcal{U}(\wswap{k_1}{k_2})$ is the normalized random walk operator of an undirected graph.
In particular, $\mathcal{U}(\wswap{k_1}{k_2})$ defines an undirected weighted bipartite graph on the vertices $W_G(k_1) \cup W_G(k_2)$, where each edge $ww'$ in this graph is weighted according to the transition probability from one walk to the other whenever one of $w$, $w'$ is in $W_G(k_1)$ and the other is in $W_G(k_2)$.
This becomes clear when taking a closer look at the adjoint operator $(\wswap{k_1}{k_2})^{\dag}$.

\begin{claim}
  Let $k_1, k_2 \in \N$ and $k = k_1 + k_2$.
  Define the operator
  $\fswap{k_1}{k_2} \colon \mathbb{R}^{W_G(k_1)} \to \mathbb{R}^{W_G(k_2)}$ such that for every $f \in \R^{W_G(k_1)}$,
  $$
  \left(\fswap{k_1}{k_2} (f)\right)(w') \coloneqq \E_{w: ww' \in W(k)} [f(w)]
  $$
  for every $w' \in W_G(k_2)$.
  Then
  $$
  \left(\wswap{k_1}{k_2}\right)^{\dag} = \fswap{k_1}{k_2}.
  $$
\end{claim}

\begin{proof}
  Let $f \in C^{W_G(k_1)}$ and $g \in C^{W_G(k_2)}$.
  We show that $\ip{f}{\wswap{k_1}{k_2} g} = \ip{\fswap{k_1}{k_2} f}{g}$.
  On one hand we have
  \begin{align*}
    \ip{f}{\wswap{k_1}{k_2} g} &= \E_{w \in W_G(k_1)} \left[f(w) \E_{w': ww' \in W_G(k)}  [g(w')] \right] \\
                                               &= \E_{w \in W_G(k_1)} \left[f(w) \sum_{w' \in W_G(k_2)} \frac{\Pi_k(ww')}{\Pi_{k_1}(w)} g(w') \right] \\
                                               &= \sum_{w \in W_G(k_1)} \Pi_{k_1}(w) f(w) \sum_{w' \in W_G(k_2)} \frac{\Pi_k(ww')}{\Pi_{k_1}(w)}  g(w')\\
                                               & = \sum_{ww' \in W_G(k)} f(w) g(w') \Pi_k(ww').
  \end{align*}
  On the other hand we have
  \begin{align*}
  \ip{\fswap{k_1}{k_2} f}{ g} &= \E_{w' \in W_G(k_2)} \left[\E_{w: ww' \in W_G(k)} [f(w)] g(w') \right] \\
                                               &= \E_{w' \in W_G(k_2)} \left[ \sum_{w \in W_G(k_1)} \frac{\Pi_k(ww')}{\Pi_{k_2}(w')} f(w) g(w') \right]\\
                                               &= \sum_{w' \in W_G(k_2)} \Pi_{k_2}(w') \sum_{w \in W_G(k_1)} \frac{\Pi_k(ww')}{\Pi_{k_2}(w')} f(w) g(w')\\
                                               &= \sum_{ww' \in W_G(k)} f(w) g(w') \Pi_k(ww').
  \end{align*}
  Hence, $\fswap{k_1}{k_2} = (\wswap{k_1}{k_2})^{\dag}$ as claimed.
\end{proof}

\subsubsection{Swap Operators are Splittable}\label{subsubsec:expander_walks_swap_op_splittable}

At a high level, the expansion of a certain collection of Swap walks
$\wswap{k_1}{k_2}$ ensures that we can round the SOS solution and this
gives rise to the \textit{splittable} notion, which we tailor to the
$W_G(k)$ case after recalling some notation.

\begin{remark}
  We establish the definitions in slightly greater generality than needed
  for our coding application since this generality is useful for
  solving $k$-CSP instances on $W_G(k)$ for more general graphs $G$
  that are not necessarily expanders
  (c.f.~\cref{subsub:k_csp_walk_constraints}).  Solving these kinds of
  $k$-CSPs might be of independent interest. For the coding
  application, the threshold rank (\cref{def:graph_trank}) will be one,
  i.e., we will be working with expander graphs.
\end{remark}

\begin{definition}[Threshold Rank of Graphs (from~\cite{BarakRS11})]\label{def:graph_trank}
  Let $G = (V,E,w)$ be a weighted graph on $n$ vertices and $\Aye$ be
  its normalized random walk matrix. Suppose the eigenvalues of $\Aye$
  are $1 = \lambda_1 \ge \cdots \ge \lambda_n$. Given a parameter
  $\tau \in (0,1)$, we denote the threshold rank of $G$ by $\Rank_{\ge
    \tau}(\Aye)$ (or $\Rank_{\ge \tau}(G)$) and define it as
  \[
    \Rank_{\ge \tau}(\Aye) \coloneqq \left\vert \{i \mid \lambda_i \ge \tau \} \right\vert.
  \]
\end{definition}

Let $\SwapT(\tree, W_G(\le k))$ be the set of all swap
graphs over $W_G(\le k)$ finding representation in the
splitting tree $\tree$, i.e., for each internal node with leaves
labeled $k_1$ and $k_2$ we associate the undirected Swap operator
$\mathcal U(\wswap{k_1}{k_2})$.

Given a threshold parameter $\tau \le 1$ and a set of normalized
adjacency matrices $\mathcal{A} = \{\Aye_1,\dots,\Aye_s\}$, we define
the threshold rank $\Rank_{\ge \tau}(\mathcal{A})$ of $\mathcal{A}$ as
$$
\Rank_{\ge \tau}(\mathcal{A}) \coloneqq \max_{\Aye \in \mathcal{A}} ~\Rank_{\ge \tau}(\Aye),
$$
where $\Rank_{\ge \tau}(\Aye)$ denotes the usual threshold rank of
$\Aye$ as in~\cref{def:graph_trank}.

\begin{definition}[$(\tree, \tau, r)$-splittability~\cite{AJT19}]
    A collection $W_G(\le k)$ is said to be $(\tree, \tau, r)$-splittable if
    $\tree$ is a $k$-splitting tree and 
    \[ \Rank_{\ge \tau}(\SwapT(\tree, W_G)) \le r.\]
    If there exists some $k$-splitting tree $\tree$ such that
    $W_G(\le k)$ is $(\tree, \tau, r)$-splittable, the
    instance $W_G(\le k)$ will be called a $(\tau,
    r)$-splittable instance.
\end{definition}

We show that the expansion of $\mathcal U(\wswap{k_1}{k_2})$ is inherited from the
expansion of its defining graph $G$. To this end we will have to
overcome the hurdle that $W_G(k) \subseteq V^{k}$ is not necessarily a
natural product space, but it can be made so with the proper representation.

\begin{lemma}\label{lemma:swap_matrix_rep}
  Let $G=(V=[n],E)$ be a $d$-regular graph with normalized random walk
  operator $\Aye_G$. Then for every $k_1,k_2 \in \mathbb{N}^{+}$,
  there are representations of $\wswap{k_1}{k_2}$ and
  $\Aye_G$ as matrices such that
  $$
  \wswap{k_1}{k_2} = \matr A_G \otimes \Jay/d^{k_2-1},
  $$
  where $\Jay \in \mathbb{R}^{[d]^{k_1-1} \times [d]^{k_2-1}}$ is the
  all ones matrix.
\end{lemma}

\begin{proof}
Partition the set of walks $W_G(k_1)$ into the sets $W_1, \dots, W_n$,
where $w \in W_i$ if the last vertex of the walk is $w_{k_1} = i$.
Similarly, partition $W_G(k_2)$ into the sets $W_1', \dots, W_n'$,
where $w' \in W_j'$ if the first vertex of the walk is $w'_1 = j$.
Note that $\abs{W_i} = d^{k_1-1}$ for all $i$ and $\abs{W_j'} =
d^{k_2-1}$ for all $j$.

Now order the rows of the matrix $\wswap{k_1}{k_2}$ so that all of the
rows corresponding to walks in $W_1$ appear first, followed by those
for walks in $W_2$, and so on, with an arbitrary order within each
set.  Do a similar re-ordering of the columns for the sets
$W_1', \dots, W_n'$.  Observe that
$$
\left(\wswap{k_1}{k_2}\right)_{w,w'} = \frac{\Pi_{k_1+k_2}(ww')}{\Pi_{k_1}(w)} = \frac{\one\left[w_{k_1} \text{ is adjacent to } w'_1 \right]}{d^{k_2-1}},
$$
which only depends on the adjacency of the last vertex of $w$ and the
first vertex of $w'$.  If the vertices $i$ and $j$ are adjacent, then
$\left(\wswap{k_1}{k_2}\right)_{w,w'} = 1/d^{k_2-1}$ for every $w \in
W_i$ and $w' \in W_j'$; otherwise,
$\left(\wswap{k_1}{k_2}\right)_{w,w'} = 0$.  Since the walks in the
rows and columns are sorted according to their last and first
vertices, respectively, the matrix $\wswap{k_1}{k_2}$ exactly matches
the tensor product $\matr A_G \otimes \Jay/d^{k_2-1}$, where the rows
and columns of $\matr A_G$ are sorted according to the usual ordering
on $[n]$.
\end{proof}

\begin{corollary}
  Let $G=(V,E)$ be a $\gamma$-two-sided spectral expander with normalized random walk operator $\Aye_G$.
  Then for every $k_1,k_2 \in \mathbb{N}^{+}$,
  $$
  \lambda_2(\mathcal U(\wswap{k_1}{k_2})) \leq \gamma.
  $$
\end{corollary}

\begin{proof}
  To make the presentation reasonably self-contained, we include the proof of the well-known connection between
  the singular values of $\wswap{k_1}{k_2}$ and the eigenvalues
  of $\mathcal U(\wswap{k_1}{k_2})$. Using~\cref{lemma:swap_matrix_rep} and the fact that
   $\sigma_i(\Aye_G \otimes \Jay/d^{k_2-1})
   = \sigma_i(\Aye_G)$, we have $\sigma_i(\wswap{k_1}{k_2}) = \sigma_i(\Aye_G)$.
   Since
   	$$\left(\mathcal U(\wswap{k_1}{k_2})^{\dag}\right) \mathcal U(\wswap{k_1}{k_2}) =
		\begin{pmatrix}
    			\wswap{k_1}{k_2} \left(\wswap{k_1}{k_2}\right)^{\dag} & 0 \\
    			0 & \left(\wswap{k_1}{k_2}\right)^{\dag} \wswap{k_1}{k_2} 
  		\end{pmatrix},$$
  the nonzero singular values of $\mathcal U(\wswap{k_1}{k_2})$ are the same as the nonzero singular values of $\wswap{k_1}{k_2}$.
  As $\mathcal U(\wswap{k_1}{k_2})$ is the random walk operator of a bipartite graph, the spectrum of $\mathcal U(\wswap{k_1}{k_2})$
  is symmetric around $0$ implying that its nonzero eigenvalues are
  	$$\pm \sigma_1(\wswap{k_1}{k_2}), \pm \sigma_2(\wswap{k_1}{k_2}), \ldots = \pm \sigma_1(\Aye_G), \pm \sigma_2(\Aye_G), \ldots$$
  Hence, the second-largest of these is $\lambda_2(\mathcal U(\wswap{k_1}{k_2})) = \sigma_2(\Aye_G) \leq \gamma$.
\end{proof}

Applying this spectral bound on $\mathcal U(\wswap{k}{k})$ to each internal node of any splitting tree readily gives the splittability of $W_G(k)$.

\begin{corollary}\label{cor:cite-ahead-split}
  If $G$ is a $\gamma$-two-sided spectral expander, then for every
  $k \in \mathbb{N}^+$ the collection $W_G(k)$ endowed with
  $\Pi_k$ is $(\gamma, 1)$-splittable (for all choices of
  splitting trees).
\end{corollary}

\subsubsection{Splittable Implies Tensorial}\label{subsec:expander_walk_two_step_tensorial}

By a simple adaptation of an argument in~\cite{AJT19} for hypergraphs
which built on~\cite{BarakRS11}, we can use the splittable property to
obtain tensorial properties for $W_G(k)$. More precisely, we can
deduce~\cref{thm:brs_tree_split}.

\begin{theorem}[Adapted from~\cite{AJT19}]\label{thm:brs_tree_split}
  Suppose $W_G(\le k)$ with $W_G(1) = [n]$ and an
  $(L+2k)$-local PSD ensemble $\rv Z = \set{\rv Z_1, \ldots, \rv Z_n}$
  are given. There exist some universal constants $c_4 \ge 0$ and
  $C'' \ge 0$ satisfying the following: If $L \ge C'' \cdot (q^{4k}
  \cdot k^7 \cdot r/\mu^5)$, $\supp(\rv Z_j) \le q$ for all $j \in [n]$, and $W_G(\le k)$ is $(c_4
  \cdot (\mu/(4 k \cdot q^{k}))^2, r)$-splittable, then
    \begin{equation}
        \ExpOp_{\Omega} \ExpOp_{w \sim W_G(k)}{ \norm{\set{\rv Z'_{w}} - \set*{\rv Z'_{w_1}}\cdots \set*{\rv Z'_{w_k}}}_1 } \le \mu,
        \label{eq:brs_expander_walk}
    \end{equation}
    where $\rv Z'$ is as defined in \cref{algo:prop-rd} on the input of $\set{\rv Z_1, \ldots, \rv Z_n}$ and $\Pi_k$. 
\end{theorem}

Using~\cref{thm:brs_tree_split}, we can establish conditions on a
$\gamma$-two-sided expander graph $G=(V,E,w)$ in order to ensure that
$W_G(k)$ is $(\mu,L)$-two-step tensorial.

\begin{lemma}[Expander walks are two-step tensorial]\label{lemma:exp_walk_double_close_to_product}    
    There exist some universal constants $c' \ge 0$ and $C' \ge 0$
    satisfying the following: If $L \ge c' \cdot (q^{4k} \cdot
    k^7/\mu^5)$, $\supp(\rv Z_j) \le q$ for all $j \in [n]$, and $G$
    is a $\gamma$-two-sided expander for $\gamma \le C' \cdot
    \mu^2/\left(k^2 \cdot q^{2k} \right)$ and size $\ge k$, then
    $W_G(k)$ is $(\mu,L)$-two-step tensorial.
\end{lemma}

\begin{proof}
  The proof is similar to the proof
  of~\cref{lemma:double_close_to_product} for HDXs, so we omit it.
\end{proof}

\subsubsection{Interlude: Approximating $k$-CSP on Walk Constraints}\label{subsub:k_csp_walk_constraints}

Now, we digress to show how using~\cref{thm:brs_tree_split} it is
possible to deduce parameters for approximating $k$-CSPs on $W_G(k)$.
We believe this result might be of independent interest and note that
it is not required in the list decoding application.

\begin{corollary}\label{cor:alg-cooked-baked2}
    Suppose $\Ins$ is a $q$-ary $k$-CSP instance with constraints on $W_G(k)$. There exist absolute constants $C'' \ge 0$
    and $c_4 \ge 0$ satisfying the following:

    If $W_G(k)$ is $(c_4 \cdot (\mu/ (4k \cdot q^{k}))^2, r)$-splittable, then there is an algorithm that runs
    in time $n^{O\parens*{q^{4k} \cdot k^7 \cdot r / \mu^5}}$ based on $(C'' \cdot k^5 \cdot q^{k} \cdot r / \mu^4)$-levels
    of SOS-hierarchy and \cref{algo:prop-rd} that outputs a random assignment $\xi: [n] \to [q]$ that in expectation ensures
    $\SAT_\Ins(\xi) = \OPT(\Ins) - \mu$. 
\end{corollary}

\begin{proof}
    The algorithm will just run \cref{algo:prop-rd} on the local PSD-ensemble $\set{\rv Z_1, \ldots, \rv Z_n}$ 
    given by the SDP relaxation of $\Ins$ 
    strengthened by $L = (C'' \cdot k^5 \cdot q^{2k} / \mu^4)$-levels of SOS-hierarchy and $\Pi_k$, where
    $C'' \ge 0$ is the constant from \cref{thm:brs_tree_split}. $\rv Z$ satisfies
    \begin{equation}
        \SDP(\Ins) = \Ex{w \sim \Pi_k}{\Ex{\set{\rv Z_w}}{\one[\rv Z_w \in \Pred_w]}} \ge \OPT(\Ins). \label{eq:ysatisfiesgood} 
    \end{equation}

    Since the conditioning done on $\set{\rv Z'}$ is consistent with
    the local distribution, by law of total expectation and \cref{eq:ysatisfiesgood} we have
    \begin{equation}
        \ExpOp_{\Omega}\ExpOp_{w \sim \Pi_k}{ \one[ \rv Z'_w \in \Pred_w]} = \SDP(\Ins) \ge \OPT(\Ins).\label{eq:yprimesatgood}
    \end{equation}
    
    By~\cref{thm:brs_tree_split} we know that
    \begin{equation}
        \ExpOp_{\Omega}{\ExpOp_{w \sim \Pi_k}{ \norm{\set{\rv Z'_w} - \set{\rv Z'_{w_1}}\cdots \set{\rv Z'_{w_k}}}_1}} \le \mu\label{eq:distsclose}.
    \end{equation}
    Now, the fraction of constraints satisfied by the algorithm in expectation is
    \[ \ExpOp_{\xi}[\SAT_\Ins(\xi)]= \ExpOp_{\Omega}{\ExpOp_{w \sim \Pi_k}{\Ex{(\xi_1, \ldots, \xi_n) \sim \set{\rv Z_1'} \cdots \set{\rv Z_n'}}{\one[ \xi|_w \in \Pred_w]}}}.\]
    By using \cref{eq:distsclose}, we can obtain 
    \[ \ExpOp_{\xi}[\SAT_\Ins(\xi)] \ge \Ex{\Omega}{\ExpOp_{\set{\rv Z_w}}{ \one[ \rv Z'_w \textrm{ satisfies the constraint on } w]}} - \mu. \]
    Using \cref{eq:yprimesatgood}, we conclude
    \[ \ExpOp_{\xi}[\SAT_\Ins(\xi)] \ge \SDP(\Ins) - \mu = \OPT(\Ins) - \mu.\]
\end{proof}

\subsection{Instantiation to Linear Base Codes}\label{sec:list_dec_xor_expander_walks}

We instantiate the list decoding framework to the direct sum lifting
given by the collection $W_G(k)$ of length $k$ walks on a sufficiently
expanding graph $G=(V,E,w)$. For parity sampling of expander walks, we
will rely on the following fact.

\begin{theorem}[Walks on Expanders are Parity Samplers~\cite{Ta-Shma17} (Restatement of~\cref{theo:expander_parity_sampler})]\label{theo:expander_parity_sampler_restated}
   Suppose $G$ is a graph with second-largest eigenvalue in absolute
   value at most $\lambda$, and let $X(k)$ be the set of all walks of
   length $k$ on $G$.  Then $X(k)$ is a $(\beta_0, (\beta_0 +
   2\lambda)^{\lfloor k/2 \rfloor})$-parity sampler.  In particular,
   for any $\beta > 0$, if $\beta_0 + 2\lambda < 1$ and $k$ is
   sufficiently large, then $X(k)$ is a
   $(\beta_0, \beta)$-parity sampler.
\end{theorem}

First, we instantiate the framework to linear codes which already
encompasses most of the ideas need for general binary codes.
\begin{lemma}[Direct sum lifting of linear biased codes II]\label{lemma:direct_sum_linear_biased_codes_exp_walks}
  Let $\epsilon_0 < 1/2$ be a constant and $\epsilon \in (0,\epsilon_0)$.
  There exists a universal constant $C > 0$ such that for any
  $d$-regular $\gamma$-two-sided expander graph $G$ on ground
  set $W_G(1)=[n]$, if $\gamma \le \epsilon^{C}$,
  then the following holds:

  For every binary \text{$2\epsilon_0$-biased} linear code $\mathcal{C}_1$  on $W_G(1)=[n]$,
  there exists a $2\epsilon$-biased binary lifted linear code
  $\mathcal{C}_k=\dsum_{X(k)}(\mathcal{C}_1)$ on $W_G(k)$ where $k = O\left(\log(1/\epsilon)\right)$ and
  \begin{itemize}
    \item{[Efficient List Decoding]} If $\tilde{y}$ is \lict{\sqrt{\epsilon}} to $\mathcal{C}_k$, then we
         can compute the list $\mathcal{L}(\tilde{y},\mathcal{C}_1,\mathcal{C}_k)$ (c.f.~\cref{def:coupled_code_list}) in time
         $$
         n^{\epsilon^{-O\left(1\right)}} \cdot f(n),
         $$
         where $f(n)$ is the running time of a unique decoding algorithm for $\mathcal{C}_1$.
    \item{[Rate]} The rate $r_k$ of $\mathcal{C}_k$ satisfies  $r_k = r_1 / d^{k-1}$ where $r_1$ is
                  the relative rate of $\mathcal{C}_1$. 
    \item{[Linear]} The lifted code $\mathcal{C}_k$ is linear.
  \end{itemize}
\end{lemma}

\begin{proof}
  The proof is analogous to the one given
  in~\cref{lemma:direct_sum_linear_biased_codes_hdx}. We want to
  define parameters for a $\gamma$-two-sided expander $G=(V,E,w)$ so
  that $W_G(k)$ satisfies strong enough \textit{robust}
  and \textit{tensorial} assumptions and we can
  apply~\cref{theo:list_dec}. In this application, we will rely on
  parity sampling for robustness. If $\dsum_{W_G(k)}$ is a
  $(2\epsilon_0,2\epsilon)$-parity sampler, using the linearity of
  $\mathcal{C}_1$, we obtain a lifted code
  $\mathcal{C}_k=\dsum_{X(k)}(\mathcal{C}_1)$ which is linear and has bias
  $2\epsilon$; thus the lifting is indeed
  $(1/2-\epsilon_0,1/2-\epsilon)$-robust. If we want to fully rely on parity
  sampling in~\cref{theo:list_dec}, the lifting must be a $(\beta_0=1/2+\epsilon_0,
  \beta=2\epsilon)$-parity sampler, which
  is more stringent than the first parity sampling
  requirement.~\footnote{Recall that
    this strengthening is used in our list decoding framework.} To
  invoke~\cref{theo:expander_parity_sampler_restated} and obtain this
  $(\beta_0,\beta)$-parity sampler, we need to choose a
  parameter $\theta$ (where $0 < \theta < (1 -
  \beta_0)/\beta_0$) such that 
  \begin{align*}
    & k \ge 2 \cdot \log_{(1+\theta) \beta_0}(\beta) + 2 \text{ and}\\
    & \gamma \leq \frac{\theta \cdot \beta_0}{2},
  \end{align*}
  which will ensure that
  	$$(\beta_0 + 2\gamma)^{\lfloor k/2 \rfloor} \leq ((1+\theta) \beta_0)^{\lfloor k/2 \rfloor} \leq \beta.$$
  To get a $(\mu,L)$-tensorial collection of walks,~\cref{lemma:exp_walk_double_close_to_product} requires
  $$
  L \ge \frac{c' \cdot 2^{4k} \cdot k^7}{\mu^5} \quad \text{ and } \quad \gamma \le \frac{C' \cdot \mu^2}{k^2 \cdot 2^{2k}}.
  $$
  where we used that our alphabet is binary (i.e., $q=2$) and $c',C' >
  0$ are constants. Finally,~\cref{theo:list_dec} requires $\mu \le
  \epsilon^8/2^{22}$. The conceptual part of the proof is essentially
  complete and we are left to compute parameters.  We choose $\theta =
  1/2 - \epsilon_0$, so that provided $\epsilon_0 < 1/2$ we have
  $(1+\theta)\beta_0 = 3/4 + \epsilon_0 - \epsilon_0^2 < 1$. Combining the
  parity sampling and tensorial requirements and after some simplification,  
  the expansion $\gamma$ is constrained as
  $$
  \gamma \le C''\cdot \min\left(\frac{\epsilon^{16}}{k^{2} \cdot 2^{2k}}, \left(1/4-\epsilon_0^2\right) \right),
  $$
  where $C'' > 0$ is a constant. We deduce that taking $\gamma$ as
  $$
  \gamma \le C''\cdot \frac{\left(1/4-\epsilon_0^2\right) \cdot \epsilon^{16}}{k^{2} \cdot 2^{2k}}
  $$
  is sufficient. Further simplifying the above bound gives $\gamma$ as in
  the statement of the theorem. Now, we turn to
  the SOS related parameter $L$ which is constrained to be
  $$
  L \ge c'' \cdot \frac{2^{4k} \cdot k^7}{\epsilon^{40}},
  $$
  where $c'' > 0$. Note that in this case the exponent $O(L+k)$
  appearing in the running time of~\cref{theo:list_dec} becomes $O(L)$.
  Further simplification of the bound on $L$ leads to a running time of
  $n^{\epsilon^{-O(1)}} \cdot f(n)$ as in the statement of the theorem.
\end{proof}

\subsection{Instantiation to General Base Codes}

The proof of~\cref{theo:direct_sum_general_codes_expander_walks}
follows from~\cref{lemma:direct_sum_linear_biased_codes_exp_walks} in
the same way that~\cref{theo:direct_sum_general_codes_restated}
follows from~\cref{lemma:direct_sum_linear_biased_codes_hdx} in the
case of HDXs.

\bibliographystyle{alphaurl}
\bibliography{macros,madhur}

\appendix

\section{Auxiliary Basic Facts of Probability}

In this section, we collect some basic facts of probability used in
the text.

\begin{fact}[First Moment Bound]\label{claim:first_moment_bound}
  Let $\rv R$ be a random variable in $[0,1]$ with $\E\left[ \rv R
    \right] = \alpha$. Let $\beta \in (0,1)$ be an arbitrary
  \textit{approximation parameter}. Then
  $$
  \ProbOp\left[\rv R \ge (1-\beta)\cdot \alpha \right] \ge \beta \cdot \alpha.
  $$
  In particular,
  $$
  \ProbOp\left[\rv R \ge \frac{\alpha}{2} \right] \ge \frac{\alpha}{2}.
  $$
\end{fact}

\begin{fact}[Chernoff Bound~\cite{ME17}]\label{fact:Chernoff_bound}
  Let $\rv R_1,\dots,\rv R_n$ be independent and identically
  distributed random variables where $\rv R_i$ is uniformly
  distributed on $\set{\pm 1}$. For every $a > 0$,
  $$
  \ProbOp\left[\abs{\sum_{i=1}^n \rv R_i} \ge a \right] \le 2\cdot \exp\left(-\frac{a^2}{2n}\right).
  $$
\end{fact}

\begin{fact}[Hoeffding Bound~\cite{ME17}]\label{fact:Hoeffding_bound}
  Let $\rv R_1,\dots,\rv R_n$ be independent random variables such
  that $\E\left[\rv R_i\right] = \mu$ and $\ProbOp\left[ a \le \rv R_i \le
    b\right] = 1$ for $i \in [n]$. For every $\beta > 0$,
  $$
  \ProbOp\left[\abs{\frac{1}{n} \sum_{i=1}^n \rv R_i - \mu} \ge \beta \right] \le 2 \cdot \exp\left(-\frac{2 \cdot \beta^2 \cdot n}{(a-b)^2}\right).
  $$
\end{fact}

\section{Further Properties of Liftings}

We show that a uniformly random odd function $g \colon \set{\pm 1}^k
\to \set{\pm 1}$ yields a parity lifting w.v.h.p. in $k$. Thus, parity
liftings abound and we are not restricted to $k$-XOR in the framework.
In fact, SOS abstracts the specific combinatorial properties of the
lifting function being able to handle them in a unified way.

\begin{lemma}
  Let $k \in \mathbb{N}^{+}$ be odd. For every
  $p,\beta, \theta > 0$ satisfying $\theta \ge \sqrt{\log(2/\beta)}/\sqrt{pk}$,
  $$
  \ProbOp_{g}\left[\abs{\E_{x \sim \textup{Bern}(p)^{\otimes k}} \left[ g(\chi^{\otimes k}(x)) \right]} \ge \beta\right] \le 2 \cdot  k \cdot \exp\left(-\beta^2 \cdot \binom{k}{\lfloor (1-\theta) p k \rfloor}/ 8 \right),
  $$
  where $g \colon \set{\pm 1}^{k} \to \set{\pm 1}$ is a uniformly
  random odd function and $\chi \colon (\mathbb{F}_2,+) \to (\set{\pm 1},\cdot)$ is
  the non-trivial character.
\end{lemma}

\begin{proof}
  It is enough to consider $p \in (0,1/2]$ since the case $p \in [1/2,1)$ can be reduced to the
  current case by taking the complement of the bit strings appearing in this analysis.
  Applying the Hoeffding bound~\cref{fact:Hoeffding_bound} yields
  \begin{align*}
    \E_{x \sim \textup{Bern}(p)^{\otimes k}} \left[ g(x) \right] &= \E_{w \sim \textup{Binom}(k,p)} \left[ g(\chi^{\otimes k}(x)) \one_{w \in [pk\pm C\cdot pk]} \right] \pm 2 \cdot \exp(-C^2)\\
                                                        &= \E_{w \sim \textup{Binom}(k,p)} \left[ g(\chi^{\otimes k}(x)) \one_{w \in [pk\pm C\cdot pk]} \right] \pm \frac{\beta}{2},    
  \end{align*}
  where the last equality follows from choosing $C = \theta 
  \sqrt{p k}$ and the assumption that $\theta \ge
  \sqrt{\log(2/\beta)}/\sqrt{pk}$.
  
  Since $p \le 1/2$, $\ell = \binom{k}{\lfloor (1-\theta) \cdot p\cdot k \rfloor}$ is a
  lower bound on the number of binary strings of the Boolean
  $k$-hypercube in a single layer of Hamming weight in the interval
  $[pk\pm C\cdot pk]$. A second application of the Hoeffding
  bound~\cref{fact:Hoeffding_bound} gives that the bias within this
  layer is
  $$
  \ProbOp_{g}\left[ \abs{\E_{x \in \mathbb{F}_2^k \colon \norm{x} = \ell}\left[g(\chi^{\otimes k}(x))\right]}  \ge \beta/2 \right] \le 2 \cdot \exp\left(\beta^2 \cdot \ell/8 \right).
  $$
  By union bound over the layers the result follows.
\end{proof}

\section{Derandomization}

We show how to derandomize the list decoding framework (which amounts
to derandomize~\cref{algo:cover_retrieval}) when the lifting function
is a parity sampler and it satisfies a bounded degree condition
(cf~\cref{lemma:deran_bd_degree}). We observe that this is the setting
of our two concrete instantiations, namely, for HDXs and expander
walks. In the former case, we work with $D$-flat distributions and in
the latter case with walk length and graph degree that are both
functions of $\epsilon$. Roughly speaking, we show that replacing a
random sample by the majority works as long as parity sampling is
sufficiently strong.

\begin{lemma}[Majority Word]\label{lemma:derand_maj}
  Let $z^* \in \set{\pm 1}^{X(1)}$ where $X(1)=[n]$. Suppose that $y^* = \lift_{X(k)}(z^*)$
  satisfy
  $$
  \E_{z \sim \set{\rv Z^{\otimes}\vert_{(S,\sigma)}}}\left[\left\vert \E_{\ess \sim \Pi_k} y_{\ess}^* \cdot \lift(z)_{\ess} \right\vert \right] \ge 3 \cdot \epsilon,
  $$
  and
  \begin{equation}\label{lemma:deran_bd_degree}
    \ProbOp_{\ess \sim \Pi_k}\left[\ess \ni i\right] \le \frac{g(\epsilon)}{n}.
  \end{equation}
  If also $\lift_{X(k)}$ is a $(1-\xi,2\epsilon)$-parity sampler for some $\xi \in (0,1)$, $\xi \ge 2 \exp{\left(- C \cdot \epsilon^2 \cdot g(\epsilon)^2 \cdot n \right)} = o_n(1)$
  where $C > 0$ is an universal constant and $\xi \ge 1/(n (1- \xi - o_n(1)))$, then
  $$
  \abs{\E_{i \in [n]} z^*_i \cdot z_i'} \ge 1 - 7 \sqrt{\xi},
  $$
  where $z' \in \set{\pm 1}^n$ is the majority defined as $z_i' = \argmax_{b\in \set{\pm 1}} \Pr_{\set{\rv Z^{\otimes}\vert_{(S,\sigma)}}}[\rv Z_i=b]$.
\end{lemma}

\begin{proof}
  Define $f(z) \coloneqq \abs{\E_{\ess \sim \Pi_k} y_{\ess}^* \cdot
  \lift(z)_{\ess}}$. Then, using~\cref{lemma:deran_bd_degree} we claim
  that $f(z)$ is $O(g(\epsilon)/n)$-Lipschitz with respect to $\ell_1$
  since
  $$
  \abs{f(z) - f(\tilde{z})} \le \sum_{i \in X(1)} 2 \cdot \ProbOp_{\ess \sim \Pi_k}\left[\ess \ni i\right] \cdot \abs{z_i-\tilde{z}_i} \le O\left(\frac{g(\epsilon)}{n}\right) \cdot \norm{z-\tilde{z}}_1.
  $$

  Since the underlying distribution of $\set{\rv
    Z^{\otimes}\vert_{(S,\sigma)}}$ is a product distribution on $\set{\pm 1}^n$ and $f$
  is $O(g(\epsilon)/n)$-Lipschitz, applying Hoeffding's inequality yields
  $$
  \ProbOp_{z \sim \set{\rv Z^{\otimes}\vert_{(S,\sigma)}}}\left[f(z) \le \epsilon \right] \le \ProbOp_{z \sim \set{\rv Z^{\otimes}\vert_{(S,\sigma)}}}\left[\left\vert f(z) - \E_{z \sim \set{\rv Z^{\otimes}\vert_{(S,\sigma)}}} f(z)\right\vert \ge \epsilon \right] \le \exp{\left(-\Theta(g'(\epsilon) \cdot n )\right)},
  $$
  where $g'(\epsilon) = \epsilon^2 \cdot g(\epsilon)^2$.

  Using the assumption that $\lift$ is a $(1-\xi,2\epsilon)$-parity sampler, we obtain
  $$
  \E_{z \sim \set{\rv Z^{\otimes}\vert_{(S,\sigma)}}}\left[\abs{\ip{z^*}{z}} \right] \ge 1-\xi - 2 \exp{\left(-\Theta(g'(\epsilon) \cdot n )\right)}.
  $$

  By Jensen's inequality, 
  $$
  \E_{z \sim \set{\rv Z^{\otimes}\vert_{(S,\sigma)}}}\left[\ip{z^*}{z}^2 \right] \ge \left( \E_{z \sim \set{\rv Z^{\otimes}\vert_{(S,\sigma)}}}\left[\abs{\ip{z^*}{z}} \right]\right)^2 \ge (1- \xi - 2 \exp{\left(-\Theta(g'(\epsilon) \cdot n )\right)})^2.
  $$
  Using indepdence, we get
  $$
  \E_{z \sim \set{\rv Z^{\otimes}\vert_{(S,\sigma)}}}\left[ \E_{i,j \in [n]} z^*_i z_i z_j z_j^*\right] \le \E_{i,j \in [n]} z^*_i \E\left[ z_i \right] \E\left[ z_j \right] z_j^* + \frac{1}{n} = \left(\E_{i \in [n]} z^*_i \E\left[ z_i \right]\right)^2 + \frac{1}{n}.
  $$
  Thus, in particular $\abs{\E_{i \in [n]} z^*_i \E\left[ z_i \right]} \ge (1-\xi - o_n(1)) - 1/((1 - \xi - o_n(1))n) \ge 1-3 \xi$ which implies
  \begin{align*}
    1- 3 \xi  & \le \abs{\E_{i \in [n]} z^*_i\left(\Pr_{\vert_{(S,\sigma)}}[\rv Z_i=1] - \Pr_{\vert_{(S,\sigma)}}[\rv Z_i=-1]\right)} \\
    & \le \E_{i \in [n]} \abs{\Pr_{\vert_{(S,\sigma)}}[\rv Z_i=1]-\Pr_{\vert_{(S,\sigma)}}[\rv Z_i=-1]}.
  \end{align*}
  Since
  $$
  \E_{i \in [n]} 1-\abs{\Pr_{\vert_{(S,\sigma)}}[\rv Z_i=1]-\Pr_{\vert_{(S,\sigma)}}[\rv Z_i=-1]} \le 3 \xi,
  $$
  Markov's inequality yields
  $$
  \ProbOp_{i \in [n]}\left[ 1 - \sqrt{\xi} \ge \abs{\Pr_{\vert_{(S,\sigma)}}[\rv Z_i=1]-\Pr_{\vert_{(S,\sigma)}}[\rv Z_i=-1]}  \right] \le 3 \sqrt{\xi}.
  $$

  Now, let $z' \in \set{\pm 1}^n$ be as in the statement of the lemma. Then,
  $$
  1- 3 \xi - 4 \sqrt{\xi}\le \abs{\E_{i \in [n]} z^*_i \cdot z_i'}.
  $$
  Hence, we conclude that $\abs{\E_{i \in [n]} z^*_i \cdot z_i'} \ge 1 - 7 \sqrt{\xi}$.
\end{proof}

\begin{remark}
  The parity sampling requierment might be slightly stronger with this
  derandomized version but it does not change the asymptotic nature of
  our results. More precisely, we are only asking for
  $(1-\xi,2\epsilon)$-parity sampler for a different constant value $\xi
  > 0$.
\end{remark}

\end{document}